\documentclass[11pts]{article}

\usepackage{amsmath}
\usepackage{amsthm}
\usepackage{amssymb}
\usepackage{caption}
\usepackage{subcaption}
\usepackage{natbib}
\usepackage{graphicx}

\usepackage[left=25.4mm,right=25.4mm,top=30mm,bottom=25.4mm,a4paper]{geometry}

\newtheorem{theorem}{Theorem}
\newtheorem{example}{Example}
\newtheorem{remark}[theorem]{Remark}
\newtheorem{proposition}[theorem]{Proposition}
\newtheorem{lemma}[theorem]{Lemma}

\newcommand{\E}{\mathbb E}
\newcommand{\e}{\mathrm e}

\newcommand{\D}{\mathrm{d}}
\newcommand{\F}{\mathcal F}

\newcommand{\I}{\mathrm i}
\newcommand{\PP}{\mathbb P}
\newcommand{\Var}{\mathrm{Var}}

\begin{document}

\title{Modeling microstructure price dynamics with symmetric Hawkes and diffusion model using ultra-high-frequency stock data}
\date{}

\author{Kyungsub Lee\footnote{Department of Statistics, Yeungnam University, Gyeongsan, Gyeongbuk 38541, Korea} and Byoung Ki Seo\footnote{Corresponding author, School of  Management Engineering, UNIST(Ulsan National Institute of Science and Technology), Ulsan 44919, Korea}}
\maketitle

\begin{abstract}
This study examine the theoretical and empirical perspectives of the symmetric Hawkes model of the price tick structure.
Combined with the maximum likelihood estimation, the model provides a proper method of volatility estimation specialized in ultra-high-frequency analysis.
Empirical studies based on the model using the ultra-high-frequency data of stocks in the S\&P 500 are performed.
The performance of the volatility measure, intraday estimation, and the dynamics of the parameters are discussed.
A new approach of diffusion analogy to the symmetric Hawkes model is proposed with the distributional properties very close to the Hawkes model.
As a diffusion process, the model provides more analytical simplicity when computing the variance formula, incorporating skewness and examining the probabilistic property.
An estimation of the diffusion model is performed using the simulated maximum likelihood method and shows similar patterns to the Hawkes model.
\end{abstract}

\section{Introduction} 
The extensive observations and analysis of ultra-high-frequency financial data has become increasingly available due to the development of computing schemes, massive storage devices, and electronic trade systems in the financial markets.
The ultra-high-frequency data includes the price dynamics and various types of trade orders recorded in seconds or with a shorter time resolution.
Therefore, there has been growing attention in the necessity for proper analysis and modeling of ultra-high-frequency financial data among practitioners and theorists.

One of the important subjects of modeling ultra-high-frequency data is the price dynamics in micro level with tick structures.
To describe the micro structure of the price dynamics and order flows, the Hawkes process \citep{Hawkes1971point, Hawkes1971} has been used to consider the non-time-homogeneous features of the duration between price changes or orders such as clustering and mutual effect.
The Hawkes process belongs to the class of point processes and is defined by constructing the conditional intensity processes as a function of previous events.

\cite{Hewlett2006} examined the model of the arrival times of trades and the price impacts based on a symmetric bivariate Hawkes process.
\cite{Large2007} examined the market resilience after large trades using the limit order book data and mutually excited multivariate Hawkes processes.
\cite{Bowsher2007} introduced a generalized Hawkes model to analyze the relationship between the trading times and mid price changes.
With mutually excited Hawkes processes that have a strong microscopic mean reversion property, \cite{Bacry2013} constructed a model that accounts for the market microstructure noise and the Epps effect.

On the other hand, \cite{Fonseca2014Self} focused on the clustering behaviors of trades using self-excited Hawkes processes with an application to the generalized method of moments estimations.
\cite{Fonseca2014} provided the moment conditions and autocorrelation functions of self and mutually excited Hawkes processes to exhibit both clustering and mean reversion.
\cite{Bacry2014} proposed a multivariate Hawkes process to model the price dynamics and the market impact of market orders to account for the various stylized facts of the market microstructure.
For more previous financial studies on market microstructure or price dynamics based on point processes or intensity modeling, 
the reader should refer to \cite{Bauwens2009}, \cite{Embrechts2011}, \cite{Bacry2012}, \cite{Zheng2014} and \cite{ChoeLee2014}.
The Hawkes process has also been applied to modeling the credit and contagion risk, see \cite{Errais2010}, \cite{Ait2010} and \cite{Dassios2012}.

This paper focuses on the tick price dynamics and volatility estimation.
The realized volatility estimator \citep{Barndorff2002a,Barndorff2002b,ABDL} in the ultra-high-frequency dynamics can be biased; 
when one uses every sample of ultra-high-frequency financial data to calculate the finite sum approximation of the integrated volatility due to the microstructure noise and clustering property, see \cite{Hansen}.
The adjustment methods of the bias in a nonparametric fashion \citep{Zhang2005, Ait2005, Ait2011} have been introduced.
In these approaches, one supposes that the observed price process consists of the latent efficient price and noise term around the efficient price process.
In contrast, in the Hawkes models or diffusion approach introduced in this paper and related literatures, one  models the observed price movements directly, which may include the noise, and compute the closed form formula for the variance of the return and analyze the properties of the variance.

Empirical studies to compare the volatilities calculated by Hawkes modeling and realized quadratic variation using the stock prices of the S\&P 500 were performed.
Because the Hawkes model approach incorporates all of the arrival times of the price change within a millisecond time resolution, such richness of data provides the efficiency of the volatility estimation.
This paper reports the relative efficiency of the Hawkes volatility compared to the realized volatility in simulation studies.
Therefore, owing to the rich information in ultra-high-frequency data combined with efficient likelihood estimation methods, one can estimate the parameters and volatilities within a relatively short time period of observation.
This is one of important features of the Hawkes model, and with this property, 
this paper presents the empirical results of the intraday volatility dynamics based on the Hawkes model.
By observing the intraday volatility variation in every moment, one can respond to sudden market movements more effectively.

In addition, a diffusion counterpart of the Hawkes model for the micro price dynamics is introduced.
The diffusion model consists of the square root processes for both volatility and drift.
The proposed diffusion model has similar properties to the symmetric Hawkes model of price process such as the strong correlation of the mean process over the time lag on a small time scale
and hence it incorporates the market microstructure noise.
This paper reports that the diffusion models generate the distribution very close to the corresponding Hawkes models using the Kolmogorov forward equation.
As a diffusion model, it is simpler to compute variance formula, able to introduce the leverage parameter which explain the skewness and provides the insight about the distributional property of return.
In addition, using simulated likelihood estimation method, the model parameters and volatility of the equity returns are examined.

The remainder of the paper is organized as follows.
Section~\ref{Sect:Hawkes} introduces the Hawkes model for the micro price dynamics with the basic setup similar to \cite{Hewlett2006}.
Section~\ref{Sect:diffusion} proposes and discusses the diffusion analogy of the symmetric Hawkes model.
Section~\ref{Sect:empirical} shows the empirical results with the symmetric Hawkes model and the corresponding diffusion model.
The daily and intraday variation of the Hawkes parameters and volatility with several stock data of the S\&P 500 are shown.
Section~\ref{Sect:concl} concludes the paper.
The proofs and further explanations are gathered in the Appendix.

\section{Hawkes process for tick dynamics}~\label{Sect:Hawkes}
\subsection{Point process}
This section starts with the introduction of the Hawkes process, which belongs to the class of point processes, (see, \cite{Daley}).
A point process, $N$, is formally defined on a state space, $\mathcal X$, as a mapping from a probability space $(\Omega, \PP)$ to $\mathcal N$, where $\mathcal N$ denotes the space of all counting measures on  the $\sigma$-field of $\mathcal X$'s Borel sets, $\mathcal B_{\mathcal X}$.
The space $\mathcal X$ is a complete separable metric space and  to study the tick-dynamics of a stock price movements, this paper focuses on the case that $\mathcal X = \mathbb R$, the time domain.
As a counting measure, $N(A,\omega)$ has a non-negative integer value for any measurable set $A\in\mathcal B_{\mathcal X}$ and is finite for any bounded measurable $A$.
Using the Dirac measure, $\delta_x$, defined for every $x\in\mathcal X$, the counting measure is represented by
$$ N = \sum_{i} k_i \delta_{x_i}$$
where $\{x_i\}$ is a countable set with at most finitely many $x_i$ in any bounded Borel set and $k_i$ is a positive integer.
This paper only considers the simple counting measure, i.e., $k_i =1$ for all $i$.

A point process $N$ can be regarded as a stochastic process by letting $N(t,\omega) = N((-\infty,t],\omega)$.
Consider a filtered probability space $(\Omega,\{\F_t\}, \PP)$, $-\infty<t\leq T$, where the $\sigma$-field $\F_t$ is generated by $N(t)$.
The Hawkes process is an orderly stationary point process $N$ constructed by modeling the conditional intensity, $\lambda$.
The conditional intensity function is represented as an adapted process to $\{\F_t\}$ such that $\lambda(t)\D t = \E[N(t+\D t)-N(t)|\F_{t}]$. 
For an $M$-dimensional Hawkes process $(N_1, \ldots, N_M)$, each intensity, $\lambda_i(t)$ of $N_i$ is assumed to be
$$\lambda_i(t) = \mu_i + \sum_{j=1}^M \int_{-\infty}^{t} \phi_{i,j}(t-u) \D N_j(u) $$
where $\phi_{i,j}(t-u)$ is normally a deterministic function and called kernel.
The integration of the r.h.s. is the stochastic integration defined pathwise.
To apply the stochastic integration theory in the later, the Hawkes and intensity processes are considered to be right continuous processes with left limits.

\subsection{Self and mutually excited Hawkes}\label{Subsec:full}

This subsection briefly reviews the self and mutually excited Hawkes model.
Consider a two dimensional Hawkes process $(N_1, N_2)$ with exponential decay kernels in the conditional intensities with constants $\mu_i$, $\alpha_{ij}$ and $\beta_{ij}$, for $0<t$:
\begin{align}
\lambda_1 (t) &= \mu_1 + \int_{-\infty}^{t} \alpha_{11} \e^{-\beta_{11}(t-u)}\D N_1(u)+ \int_{-\infty}^{t} \alpha_{12} \e^{-\beta_{12}(t-u)}\D N_2(u)\nonumber\\
&= \mu_1 + \lambda_{11}(0)\e^{-\beta_{11} t} + \lambda_{12}(0)\e^{-\beta_{12} t} + \int_{0}^{t} \alpha_{11} \e^{-\beta_{11}(t-u)}\D N_1(u)+ \int_{0}^{t} \alpha_{12} \e^{-\beta_{12}(t-u)}\D N_2(u),\label{Eq:lambda1}
\end{align}
and
\begin{align}
\lambda_2 (t) &= \mu_2 + \int_{-\infty}^{t} \alpha_{21} \e^{-\beta_{21}(t-u)}\D N_1(u) + \int_{-\infty}^{t} \alpha_{22} \e^{-\beta_{22}(t-u)}\D N_2(u)\nonumber\\
&= \mu_2 + \lambda_{21}(0)\e^{-\beta_{21} t} + \lambda_{22}(0)\e^{-\beta_{22} t} + \int_{0}^{t} \alpha_{21} \e^{-\beta_{21}(t-u)}\D N_1(u)+ \int_{0}^{t} \alpha_{22} \e^{-\beta_{22}(t-u)}\D N_2(u)\label{Eq:lambda2}
\end{align}
where
\begin{align*}
\lambda_{ij}(t)= \int_{-\infty}^{t} \alpha_{ij} \e^{-\beta_{ij}(t-u)}\D N_j(u).
\end{align*}
In this paper, this model is called the fully characterized self and mutually excited Hawkes process compared to the symmetric Haweks process introduced later.
Note that $\lambda_{11}$ and $\lambda_{22}$ are self-excited components, $\lambda_{12}$ and $\lambda_{21}$ are mutually excited components, and every parameter such as $\alpha_{ij}$ and $\beta_{ij}$, can have a different value.
This model was proposed by \cite{Bacry2013} and was studied for a simplified version focused on the self-excited term.
The self and mutually excited Hawkes model and its moment properties are studied in \cite{Fonseca2014}.

The components of the intensity processes, $\lambda_{ij}$, can be rewritten by 
\begin{equation}
\lambda_{ij}(t)= q_{ij} \int_{-\infty}^{t}  \beta_{ij}\e^{-\beta_{ij}(t-u)}\D N_j(u)\label{Eq:normalized}
\end{equation}
where $q_{ij} := \frac{\alpha_{ij}}{\beta_{ij}}$ and the integrand, $\beta_{ij}\e^{-\beta_{ij}(t-u)}$, is a normalized decaying function in the sense that
$$ \int_0^\infty \beta_{ij}\e^{-\beta_{ij}\tau} \D\tau = 1.$$
The coefficients, $q_{ij} $, form a branching matrix, $Q = \{ q_{ij} \}_{i,j=1,2}$ 
and if the spectral radius, the maximum of the absolute eigenvalues of $Q$, is less than 1, then the Hawkes process is well defined \citep{Hawkes1974, Bremaud1981}.

The stock price process can be assumed to be represented by the difference between two Hawkes processes,
\begin{equation}
S_t = S_0 + \delta \{N_1(t) - N_2(t) - (N_1(0) - N_2(0))\} \label{Eq:price}
\end{equation}
where $\delta$ denotes the unit size of the price movement in the tick structure of price dynamics.
(In the previous subsection, $\delta_x$ was used to denote the Dirac measure.
On the other hand, without the subscript, $\delta$ is a constant that represents the tick size.)
The process $N_1$ represents the up movements of the price process and $N_2$ represents the down movements.

However, the fully characterized Hawkes model is too complicated not only in the number of parameters but also in the fact that the model becomes four dimensional problems when dealing with the moment conditions as explained in \ref{Sect:intensity}.
(Nonetheless, we will provide some empirical results with the fully characterized model in Section~\ref{Sect:empirical}.)
In the next subsection, we consider a simpler version.

\subsection{Symmetric Hawkes process}\label{subsect:symmetric}

This subsection explains the symmetric Hawkes model for the price dynamics.
The empirical study shows that the symmetric version also well represents the basic properties of the tick dynamics.
To simplify the model from the fully characterized version, the parameter condition is imposed as
$$ \alpha_c := \alpha_{12} = \alpha_{21},\quad \alpha_s := \alpha_{11} = \alpha_{22}$$
$$\beta:=\beta_{11} = \beta_{12} = \beta_{21} = \beta_{22}, \quad \mu:=\mu_1=\mu_2.$$
Then 
\begin{align}
\lambda_1 (t) &= \mu + \int_{-\infty}^{t} \alpha_s \e^{-\beta(t-u)}\D N_1(u)+ \int_{-\infty}^{t} \alpha_c \e^{-\beta(t-u)}\D N_2(u) \label{Eq:symmetric1}\\
&= \mu + (\lambda_1(0) - \mu)\e^{-\beta t} + \int_{0}^{t} \alpha_s \e^{-\beta(t-u)}\D N_1(u)+ \int_{0}^{t} \alpha_c \e^{-\beta(t-u)}\D N_2(u)\nonumber\\
\lambda_2 (t) &= \mu + \int_{-\infty}^{t} \alpha_c \e^{-\beta(t-u)}\D N_1(u) + \int_{-\infty}^{t} \alpha_s \e^{-\beta(t-u)}\D N_2(u) \label{Eq:symmetric2}\\
&= \mu + (\lambda_2(0) - \mu)\e^{-\beta t} + \int_{0}^{t} \alpha_c \e^{-\beta(t-u)}\D N_1(u) + \int_{0}^{t} \alpha_s \e^{-\beta(t-u)}\D N_2(u).\nonumber
\end{align}
This can also be written as
\begin{align*}
\D \lambda_1(t) &= \beta(\mu-\lambda_1(t))\D t + \alpha_s \D N_1(t) + \alpha_c \D N_2(t) \\
&=\left\{ \beta\mu + (\alpha_s - \beta)\lambda_1(t) + \alpha_c \lambda_2(t) \right\} \D t + \alpha_s (\D N_1(t) -\lambda_1(t) \D t) + \alpha_c (\D N_1(t) -\lambda_1(t) \D t) \\
\D \lambda_2(t) &= \beta(\mu-\lambda_2(t))\D t + \alpha_c \D N_1(t) + \alpha_s \D N_2(t)\\
&=\left\{ \beta\mu + \alpha_c\lambda_1(t) + (\alpha_s - \beta) \lambda_2(t) \right\} \D t + \alpha_c (\D N_1(t) -\lambda_1(t) \D t) + \alpha_s (\D N_1(t) -\lambda_1(t) \D t).
\end{align*}
Note that
$$\alpha_c \lambda_{11}(t) = \alpha_s \lambda_{21}(t), \quad \alpha_s \lambda_{12}(t) = \alpha_c \lambda_{22}(t).$$
By setting $\beta_{11} = \beta_{12}$ and $\beta_{21} = \beta_{22}$, the processes $(N_1, N_2, \lambda_1,\lambda_2)$ are Markov and the differential equation system of the expected intensities becomes two dimensional.

By the differential forms of $\lambda_i$, 
\begin{equation}
\begin{bmatrix}
\ell'_1(t|s)\\
\ell'_2(t|s)
\end{bmatrix}
= 
\begin{bmatrix}
\alpha_{s}-\beta & \alpha_{c} \\
\alpha_{c} & \alpha_{s}-\beta
\end{bmatrix}
\begin{bmatrix}
\ell_1(t|s) \\
\ell_2(t|s)
\end{bmatrix}
+
\begin{bmatrix}
\beta \mu \\
\beta \mu
\end{bmatrix} \label{Eq:system}
\end{equation}
where $\ell_i(t|s) = \E_s[\lambda_i(t)]$ and the derivatives are with respect to $t$.
Let 
$$
M= \begin{bmatrix}
\alpha_{s}-\beta & \alpha_{c} \\
\alpha_{c} & \alpha_{s}-\beta
\end{bmatrix}.
$$
The eigenvalues of $M$ are
$$ (\xi_1, \xi_2) = (-\beta-\alpha_c+\alpha_s, -\beta+\alpha_c+\alpha_s),$$
and the corresponding eigenvectors are $(-1,1)$ and $(1,1)$, respectively.
If the eigenvalues are all negative, then the solution to the system converges to the particular solution as time approaches infinity.
This is equivalent to the condition that the spectral radius of the branching matrix is less than one 
where, in the sense of parametrization in Eq.~\eqref{Eq:normalized}, the branching matrix is
$$ Q = \begin{bmatrix} q_s & q_c \\ q_c & q_s \end{bmatrix}$$
with $q_s := \alpha_s/\beta$ and $q_c := \alpha_c/\beta$.

The solution of system \eqref{Eq:system} is
$$
\begin{bmatrix}
\E_s[\lambda_1(t)]\\
\E_s[\lambda_2(t)]
\end{bmatrix}
=
\frac{ - \lambda_1(s) + \lambda_2(s)}{2}\e^{\xi_1(t-s)}\begin{bmatrix} -1\\1 \end{bmatrix} + \frac{\lambda_1(s) + \lambda_2(s)}{2}\e^{\xi_2(t-s)}\begin{bmatrix} 1\\1 \end{bmatrix} -\frac{\mu\beta}{\xi_2}\left(1-\e^{\xi_2(t-s)}\right)\begin{bmatrix} 1\\1 \end{bmatrix}.
$$
The long-run expectations of the intensities as $t \rightarrow \infty$, i.e., the particular solution of the system \eqref{Eq:system} is 
$$
\dfrac{\mu\beta}{\beta-(\alpha_s + \alpha_c)}
\begin{bmatrix}
1 \\ 1 
\end{bmatrix}
=
-\frac{\mu\beta}{\xi_2}
\begin{bmatrix}
1 \\ 1 
\end{bmatrix}
.
$$
In the latter, for computational ease, it is usually assumed that the intensity processes are in the stationary state at time 0, i.e.,
\begin{equation}
\lambda_1(0) = \lambda_2(0) = \frac{\mu \beta}{\beta - \alpha_s - \alpha_c} = -\frac{\mu \beta}{\xi_2}. \label{Eq:initial}
\end{equation}

The formula for the variance of the return generated by the symmetric Hawkes model is quite simple, as represented in Proposition~\ref{Prop:vol}.
The simplicity largely depends on the symmetry of the parameter setting and the assumption of the stationary state condition at time 0.
Indeed, the stationary condition does not significantly affect the result on the variance formula in the high-frequency price dynamics modeling
as the expectations of the intensities quickly converge.
The formula was derived independently but \cite{Fonseca2014} reported a similar result (however, different from the exponential term below).

\begin{proposition}\label{Prop:vol}
Assume that the price process, $S$, follows the difference of two symmetric Hawkes processes defined by Eqs.~\eqref{Eq:price},\eqref{Eq:symmetric1}, and \eqref{Eq:symmetric2}.
Under the stationarity condition of the intensity processes at time 0 as in Eq.~\eqref{Eq:initial},
the variance of the return is represented by
\begin{align*}
\Var\left(\frac{S_t - S_0}{S_0}\right) = 
\frac{2\delta^2\lambda_1(0)}{S_0^2\xi_1^2}\left\{ \beta^2 t - 2(\alpha_s - \alpha_c)\beta\left(\frac{\e^{\xi_1 t}-1}{\xi_1}\right) + (\alpha_s -\alpha_c)^2\left( \frac{\e^{2\xi_1 t}-1}{2\xi_1} \right) \right\}
\end{align*}
\end{proposition}
\begin{proof}
See \ref{Sect:var_Hawkes}.
\end{proof}

\begin{remark}\label{Rmk:reparm}
If $t$ is sufficiently large, then the variance is approximated by
\begin{align*}
\Var\left(\frac{S_t - S_0}{S_0}\right) &\approx \frac{2\delta^2\beta^2 \lambda_1(0) t}{S_0^2\xi_1^2 }  = \frac{2\delta^2 \mu t}{S_0^2(1-q_s + q_c)^2(1-q_s-q_c)}\\
&=\frac{2\delta^2 \mu t}{S_0^2 \left(1- \frac{\alpha_s}{\beta} + \frac{\alpha_c}{\beta}\right)^2 \left(1-\frac{\alpha_s}{\beta}-\frac{\alpha_c}{\beta}\right)}.
\end{align*}
In this approximation, the following parameterization of the symmetric Hawkes process is useful for the volatility estimation.
The parameter $\mu$ is represented by a formula consisting of the annualized daily volatility and the other parameters in the Hawkes model and are given by
\begin{align*}
\mu &=- \frac{\sigma_{\mathrm{ann}}^2}{T}\cdot\frac{\xi_1^2\xi_2}{2\beta^3 \delta^2_r}\\
&= -\frac{\sigma_{\mathrm{ann}}^2}{T}\cdot\frac{(1-q_s + q_c)^2(1-q_s-q_c)}{2\delta^2_r}\end{align*}
where $\sigma_{\mathrm{ann}}$ denotes the annualized volatility, $\delta_r = \delta/S_0$ and $T$ is one year.
\end{remark}

\begin{remark}
Proposition 2 in \cite{Fonseca2014} showed the formula for the mean signature plot:
$$ \frac{\nu^2}{2}\Lambda \left(\kappa^2 + (1-\kappa^2)\frac{(1-\e^{-\tau \gamma}}{\gamma \tau} \right).$$
Based on the definition of the mean signature plot,
by setting $\tau=t$ and multiplying the mean signature plot by $t/S(0)^2$,
the meaning of the formula is the same as the $\Var\left(\frac{S_t - S_0}{S_0}\right)$ in Proposition 1 of our paper.
If we rewrite \cite{Fonseca2014}'s formula with the notations used in our paper, then we have
$$ \frac{2\delta^2 \lambda_1(0) }{S_0^2 \xi_1^2} \left\{\beta^2 t + \left( (\alpha_s - \alpha_c)^2 - 2\beta(\alpha_s-\alpha_c) \right) \left(\frac{\e^{\xi_1 t}-1}{\xi_1}\right) \right\}.$$
The result of this formula is different from the formula for Proposition \label{Prop:vol} in the exponential term.
However, as discussed in Remark~\ref{Rmk:reparm}, the exponential term is negligible if $t$ is large.

\end{remark}

\subsection{Simulation study}

In this subsection, simulation studies are performed with the symmetric Hawkes processes.
With predetermined parameter settings, 500 sample paths of the price processes defined by the difference between the two symmetric Hawkes processes with 5.5 hours' time horizon are generated.
For each path, the maximum likelihood estimation is performed using the realized arrival times of the simulated path.
Table~\ref{Table:simulation} lists the results.
The detailed information about the simulation method, see \ref{Sect:simul} and
for the likelihood estimation, see \ref{Sect:likelihood}.
The table consists of two panels with different parameter settings.

The row `mean' is for the sample mean of the likelihood estimates of 500 samples.
The row `std.' is for the sample standard deviations of the estimates.
The column `H. vol' is for the mean of the volatility estimates calculated by the likelihood estimates of $\mu, \alpha_s, \alpha_c, \beta$ using Proposition~\ref{Prop:vol}.
This is compared with the theoretic volatility computed by Proposition~\ref{Prop:vol} in the row of `True'.
The column `TSRV' reports the two scale realized volatility (TSRV) proposed by \cite{Zhang2005},
which is known to be an unbiased estimator in the presence of independent market microstructure noise.
For the TSRV computation, the small time scale is 1 second and the large time scale is 5 minutes.

The Hawkes volatility and TSRV both are quite close to the true value of the volatility.
The standard deviations of the Hawkes volatility are smaller than the standard deviations of the TSRV, implying the efficiency of the maximum likelihood estimation.
More precisely, in the maximum likelihood estimation of the Hawkes model, all the information about the time arrivals of events are used without missing single events over the observed period.
On the other hand, in the computation of the realized volatility under the equidistant setting, it is needed to choose specific points that belong to the sub-grids of the interval.

\begin{table}
\caption{Simulation study with 500 samples}\label{Table:simulation}
\centering
\begin{tabular}{cccccccccc}
\hline
& $\mu$ & $\alpha_s$ & $\alpha_c$ & $\beta$ &  H. vol & TSRV\\
\hline
True & 0.0100 & 0.4000 & 0.5000 & 1.5000 & 0.1171 & 0.1171\\
mean & 0.0100 & 0.4021 & 0.5027 & 1.5024 & 0.1177 & 0.1165 \\
std. & (0.0005) & (0.0394) & (0.0428) & (0.0841) &  (0.0057) & (0.0114)\\
\hline
True & 0.0500 & 0.6500 & 0.2000 & 1.7000 & 0.3396 & 0.3396\\
mean & 0.0500 & 0.6514 & 0.2012 & 1.7027 & 0.3400 & 0.3370 \\
std. & (0.0014) & (0.0282) & (0.0144) & (0.0643) & (0.0103) & (0.0283)\\
\hline
\end{tabular}
\end{table}

The likelihood function of the symmetric Hawkes model may not be concave but is concave when $\beta$ is fixed.
For any given observed jump times $t_i$, the log likelihood function of the up jump over interval $[0,T]$ is
\begin{align*}
\log L_1(T) ={}& \int_0^T \log \lambda_1 (u) \D N_1(u) - \int_0^T \lambda_1(u) \D u \\
={}& \sum_{t_i < T} \left( \log \lambda_{1}(t_i) - \int_{t_{i-1}}^{t_i} \lambda_{1}(u) \D u \right) - \int_{t_N}^{T} \lambda_{1}(u) \D u \\
={}& \sum_{t_i < T} \left( \log \lambda_{1}(t_i) - \frac{\e^{\beta \tau_i}-1}{\beta}\lambda_{1}(t_i)  \right) - \frac{\e^{\beta (T-t_N)}-1}{\beta}\lambda_{1}(T)
\end{align*}
where $t_N$ is the last jump time up to $T$ and $\tau_i= t_{i+1}-t_i$.
Using Eq.~\eqref{Eq:symmetric1}, the term $\log \lambda_{1}(t_i) - \frac{\e^{\beta \tau_i}-1}{\beta}\lambda_{1}(t_i)$ is represented by
\begin{align*}
&\log \lambda_{1}(t_i) - \frac{\e^{\beta \tau_i}-1}{\beta}\lambda_{1}(t_i)\\
={}& \log \left\{ \lambda_{1}(0)\e^{-\beta t_i } + \mu (1-\e^{-\beta t_i }) + \alpha_s \int_0^{t_i}  \e^{-\beta (t_i -u)} \D N_1 (u)  + \alpha_c \int_0^{t_i}  \e^{-\beta (t_i -u)} \D N_2 (u) \right\}\\
&-  \frac{\e^{\beta \tau_i}-1}{\beta} \left\{ \lambda_{1}(0)\e^{-\beta t_i } + \mu (1-\e^{-\beta t_i }) + \alpha_s\int_0^{t_i}  \e^{-\beta (t_i -u)} \D N_1 (u) +  \alpha_c \int_0^{t_i} \e^{-\beta (t_i -u)} \D N_2 (u) \right\}.
\end{align*}
When $\beta$ is fixed, then the term is represented by
$$ \log \lambda_{1}(t_i) - \frac{1-\e^{-\beta \tau_i}}{\beta}\lambda_{1}(t_i) = \log(c_{i,0} + c_{i,1} \mu + c_{i,2} \alpha_s + c_{i,3} \alpha_c) - \frac{\e^{\beta \tau_i}-1}{\beta}(c_{i,0} + c_{i,1} \mu + c_{i,2} \alpha_s + c_{i,3} \alpha_c)$$
for some constants $c_{i,0}, c_{i,1}, c_{i,2},$ and $c_{i,3}$.
By simple calculation, we have the negative semidefinite Hessian matrix of the term with respect to $\mu, \alpha_s, \alpha_c$ is
\begin{align*}
H_{1,i} = 
\frac{1}{\lambda^2_{1}(t_i)}
\begin{bmatrix}
- c_{i,1}^2 & -c_{i,1} c_{i,2} & -c_{i,1} c_{i,3}  \\
- c_{i,1}c_{i,2} & -c_{i,2}^2 & -c_{i,2} c_{i,3}  \\
-c_{i,1} c_{i,3} & -c_{i,2} c_{i,3} & -c_{i,3}^2  \\
\end{bmatrix}.
\end{align*}
Similarly, we define $H_{2,i}$,
and the Hessian matrix of the $\log L(T)$ is $H = \sum_{t_i < T} (H_{1,i} + H_{2,i})$ which is also negative semidefinite and implies the log-likelihood function is conditionally concave when $\beta$ is fixed.

This means that if we compute the log-likelihood for every value of a reasonable set of $\beta$,
(a numerical procedure will perform this task well, because the log-likelihood function is concave for any fixed $\beta$),
and by comparing the computed values, we can find the maximum log-likelihood.
Therefore, we set a possible interval for $\beta$, for example, $\beta \in [1,3]$, and with sufficiently small step size, for example, 0.0001, we can find the estimates which make the log-likelihood close enough to the 
maximum log-likelihood.
The examples with the above simulation set of Table~\ref{Table:simulation} is shown in Figure~\ref{Fig:LLF}.
Although the above method can guarantee finding the maximum value, because it is time-consuming, we generally use the numerical procedure such as the BFGS algorithm based on the newton method to find the maximum likelihood estimates.
Although the proof of the BFGS algorithm's global convergence for the nonconvex function is not yet known, it is also known that global convergence works well in most cases \citep{li2001global}.
For more information about the algorithm and its implementation, consult \cite{broyden1970convergence} and \cite{ nash2014best}.

In this simulation study, two methods show quite close results.
For simulation set 1, 
the estimates computed by fixing $\beta$ are $\mu = 0.0099, \alpha_s = 0.6590, \alpha_c = 0.0.4864, \beta =2.0646$
and the estimates through the BFGS algorithm are $\mu = 0.0099, \alpha_s = 0.6590, \alpha_c = 0.4864, \beta =2.0346$.
For simulation set 2, 
the estimates computed by fixing $\beta$ are $\mu = 0.0502, \alpha_s = 0.6273, \alpha_c = 0.2085, \beta =1.6861$
and the estimates through the BFGS algorithm are $\mu = 0.0502, \alpha_s = 0.6272, \alpha_c = 0.2084, \beta =1.6860$.
In other simulation examples not recorded here, the BFGS algorithm always yields very similar results when compared with the method of fixing $\beta$.
Since the method of fixing $\beta$ is relatively time-consuming, 
by assuming that the BFGS algorithm provides very accurate estimates,
we use the BFGS algorithm in future estimations.

\begin{figure}
\includegraphics[width=0.45\textwidth]{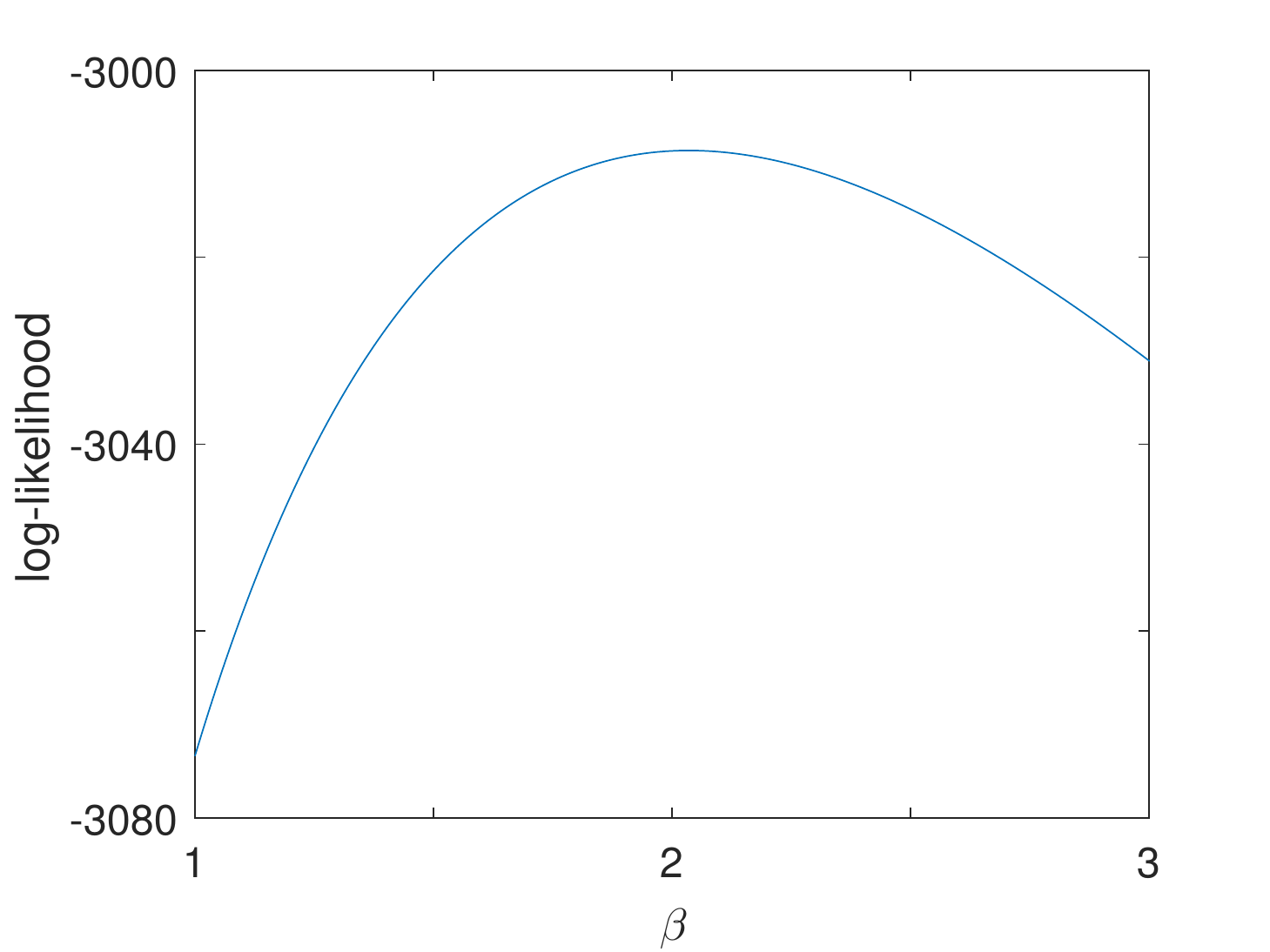}
\includegraphics[width=0.45\textwidth]{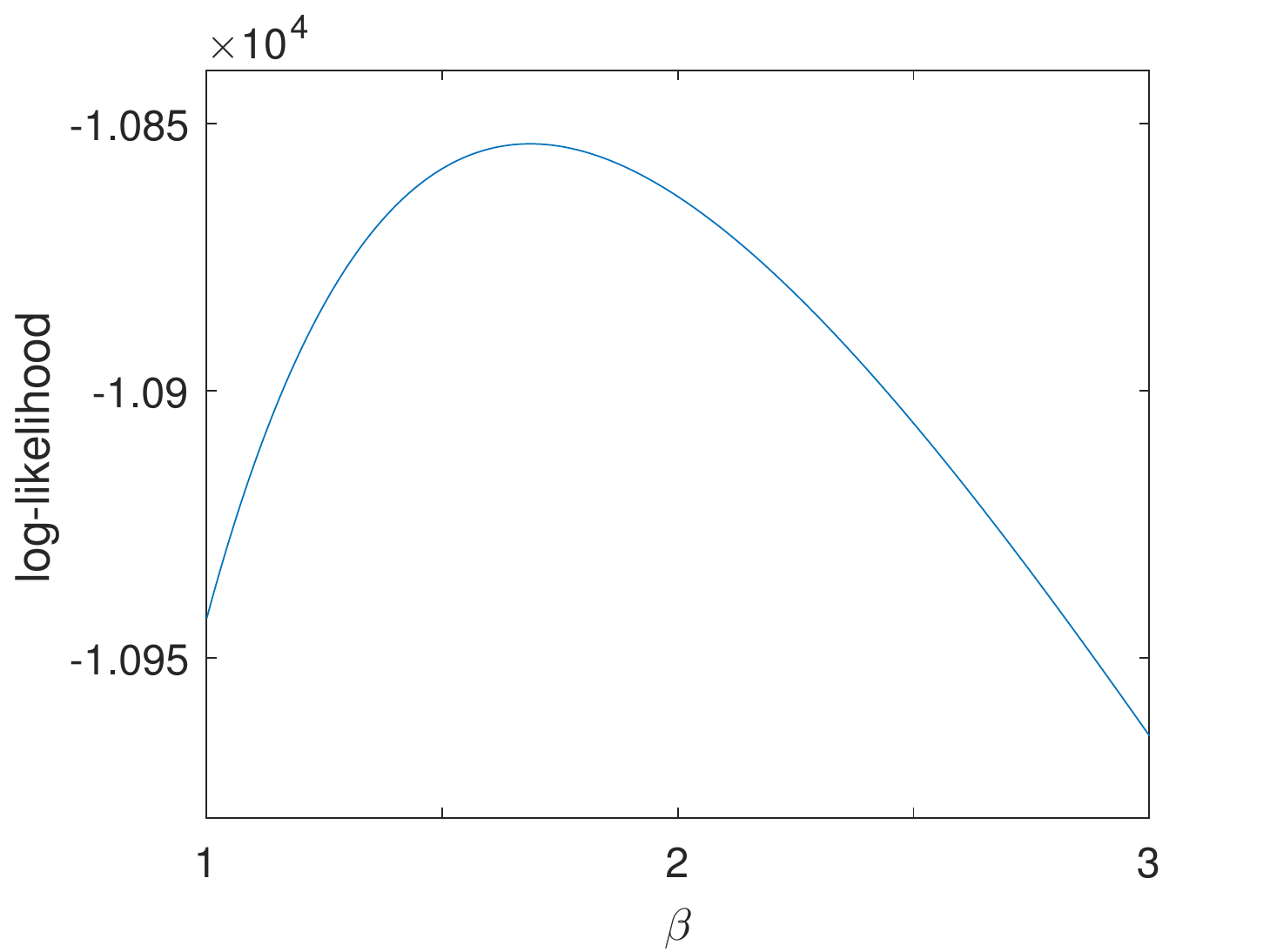}
\caption{Maximum log-likelihood function when $\beta$ is fixed for simulation set 1 (left) and 2 (right)}\label{Fig:LLF}
\end{figure}

\section{Diffusion analogy}\label{Sect:diffusion}

\subsection{Diffusion model}\label{Subsect:diffusion}
This subsection proposes a new diffusion approach for the tick structure.
The diffusion model is analogous to the symmetric Hawkes model and has a similar probabilistic property.

When the price process is represented by the difference of the two Hawkes process, the increment of the price process can be rewritten as
\begin{align*}
\Delta S(t) &= \Delta \{ \delta (N_1(t)-N_2(t)) \}\\
&= \delta (\lambda_1(t) - \lambda_2(t)) \Delta t + \delta\sqrt{\lambda_1(t)}\frac{\Delta N_1(t) - \lambda_1(t) \Delta t}{\sqrt{\lambda_1(t)\Delta t}}\sqrt{\Delta t} - \delta\sqrt{\lambda_2(t)}\frac{\Delta N_2(t) - \lambda_2(t) \Delta t}{\sqrt{\lambda_2(t)\Delta t}}\sqrt{\Delta t}.
\end{align*}
Based on the empirical studies, a sufficient number of price changes were observed during, e.g., one minute, 
and hence the normal approximation to the Poisson distribution
$$ \frac{\Delta N_i(t) - \lambda_i(t) \Delta t}{\sqrt{\lambda_i(t)\Delta t}} \sim N(0, 1)$$
can be considered.
Therefore, it is natural to consider a diffusion analogy to the symmetric Hawkes model such as
$$ \delta \sqrt{\lambda_1(t)}\frac{\D N_1(t) - \lambda_1(t) \D t}{\sqrt{\lambda_1(t)}} - \delta \sqrt{\lambda_2(t)}\frac{\D N_2(t) - \lambda_2(t) \D t}{\sqrt{\lambda_2(t)}} \approx \delta \sqrt{\lambda_1(t)} \D B_1(t) + \delta \sqrt{\lambda_2(t)} \D B_2(t) $$
for some independent Brownian motions $B_1$ and $B_2$.

By the independence, the infinitesimal variance is
$$\Var_t \left( \delta\sqrt{\lambda_1(t)} \D B_1(t) + \delta\sqrt{\lambda_1(t)} \D B_2(t) \right) = \delta^2 (\lambda_1(t) + \lambda_2(t)) \D t$$
which can be written
$$ \delta\sqrt{\lambda_1(t) + \lambda_2(t)}\D W^s_t =  \delta \sqrt{\lambda_1(t)} \D B_1(t) + \delta \sqrt{\lambda_2(t)} \D B_2(t) $$
for some Brownian motion $W^s$.
In the left hand side, $\delta^2 (\lambda_1(t) + \lambda_2(t))$ as the instantaneous variance $V_t$ of the price process.
In addition, by treating $ \delta(\lambda_1(t) - \lambda_2(t))$ as the mean process $n_t$ of the price process, a diffusion analogy of the price process can be derived as follows:
\begin{equation}
\D S_t = n_t \D t + \sqrt{V_t}\D W^s_t. \label{Eq:S_diff}
\end{equation}

Now the diffusion analogies of $n_t$ and $V_t$ are constructed.
This is because, by the definition of $\lambda_i(t)$, 
\begin{align*}
\D \{ \delta(\lambda_1(t) - \lambda_2(t)) \} ={}& (\alpha_s - \alpha_c - \beta)\delta(\lambda_1(t) - \lambda_2(t))\D t \\
&+ (\alpha_s-\alpha_c) \left\{ \delta\sqrt{\lambda_1(t)}\frac{\D N_1(t) - \lambda_1(t) \D t}{\sqrt{\lambda_1(t)}} - \delta\sqrt{\lambda_2(t)}\frac{\D N_2(t) - \lambda_2(t) \D t}{\sqrt{\lambda_2(t)}} \right\},
\end{align*}
let
\begin{align*}
\D n_t &= (a_s - a_c - b)n_t\D t + (a_s-a_c)\sqrt{V_t} \D W^s_t\\
&=: -\kappa_1 n_t \D t + \phi \sqrt{V_t} \D W^s_t
\end{align*}
where $a_s, a_c, b$ are the diffusion counterparts of $\alpha_s, \alpha_c, \beta$, respectively.

The micro structure of the price dynamics are slightly different from the macro dynamics
as the non-zero drift term in the price process is observed.
The drift term in the micro dynamics is also called the microstructure noise and related to the mutually excited feature in the Hawkes model.

In addition, because
\begin{align*}
\D\{\delta^2(\lambda_1(t) + \lambda_2(t))\} ={}& \{ 2\beta \mu \delta^2 + (\alpha_s + \alpha_c - \beta)\delta^2(\lambda_1(t) + \lambda_2(t)) \} \D t \\
&+ \delta (\alpha_s+\alpha_c)  \left\{ \delta\sqrt{\lambda_1(t)}\frac{\D N_1(t) - \lambda_1(t) \D t}{\sqrt{\lambda_1(t)}} + \delta \sqrt{\lambda_2(t)}\frac{\D N_2(t) - \lambda_2(t) \D t}{\sqrt{\lambda_2(t)}} \right\},
\end{align*}
under the similar argument of the drift, such as 
\begin{align*}
V_t &= \delta^2 (\lambda_1(t) + \lambda_2(t)), \\
\sqrt{V_t} \D W^v_t &= \delta \sqrt{\lambda_1(t)}\D B_1 (t) - \delta \sqrt{\lambda_2(t)}\D B_2 (t)\\
&= \delta\sqrt{\lambda_1(t)}\frac{\D N_1(t) - \lambda_1(t) \D t}{\sqrt{\lambda_1(t)}} + \delta \sqrt{\lambda_2(t)}\frac{\D N_2(t) - \lambda_2(t) \D t}{\sqrt{\lambda_2(t)}},
\end{align*}
the familiar square root variance process as introduced in \cite{Heston1993} is derived:
\begin{align*}
\D V_t &= (b-a_s-a_c)\left\{\frac{2bm\delta^2}{b-a_s-a_c} - V_t \right\} \D t + \delta(a_s+a_c)\sqrt{V_t} \D W^v_t\\
&=: \kappa_2 (\theta -V_t) \D t + \gamma \sqrt{V_t} \D W^v_t.
\end{align*}
where $m$ is the diffusion counterpart of $\mu$.

In this reasoning, the correlation $\rho$ that satisfies $ \D [W^s, W^v]_t = \rho_t \D t$ is represented by
$$ \rho_t = \frac{\lambda_1(t) - \lambda_2(t)}{\lambda_1(t) + \lambda_2(t)}.$$
In addition, if there is no jump for a sufficiently long interval and hence $\lambda_1(t) \rightarrow \lambda_{\infty}$ and $\lambda_2(t) \rightarrow \lambda_{\infty}$, then $\rho_t \rightarrow 0 $.

Even though $\rho$ is represented by $\lambda$s or converges to zero,
we consider that this constraint is better to be relaxed for flexibility of the model.
Note that the above derivation is not an exact mathematical justification, but rather to provide an intuition to construct a diffusion model for the micro structure of price dynamics.
For example, if necessary, the asymmetry in the price dynamics is simply introduced by a constant leverage parameter $\rho$ such that 
$$ \D [W^s, W^v]_t = \rho \D t$$
as in the typical macro level price dynamics modeling.

Overall, the price, mean and variance process are as follows:
\begin{align*}
\D S_t &= n_t \D t + \sqrt{V_t} \D W^s_t,\\
\D n_t &= -\kappa_1 n_t \D t + \phi \sqrt{V_t} \D W^s_t,\\
\D V_t &= \kappa_2 (\theta -V_t) \D t + \gamma \sqrt{V_t} \D W^v_t, \quad \D [W^s, W^v]_t = \rho \D t
\end{align*}
with the parameter relations:
\begin{align*}
\kappa_1 &= b - a_s + a_c \\
\kappa_2 &= b - a_s - a_c \\
\theta &= \frac{2bm\delta^2}{b-a_s-a_c} \\
\gamma &= \delta(a_s+a_c) \\
\phi &= a_s - a_c  = \frac{\gamma}{\delta} - \kappa_1 + \kappa_2.
\end{align*}
Note that with this analogy, $\kappa_1$ corresponds to $-\xi_1$ in the symmetric Hawkes model and $\kappa_2$ corresponds to $-\xi_2$.
The diffusion model is not an exact mathematical limit of the (symmetric) Hawkes model\
but has a very close distributional property with the Hawkes model.
For recent studies about the limit theorem of the Hawkes process, consult \cite{jaisson2015limit}.

\subsection{Basic property}

The diffusion model has several advantages.
First, by the forward Kolmogorov equation, the joint probability density function $f(s,n,v,t)$ of the diffusion model with $s=S_t, n=n_t, v=V_t$ at time $t$ satisfies the following partial differential equation
\begin{align*}
\frac{\partial f}{\partial t} ={}& -n \frac{\partial f}{\partial s} + \kappa_1 \frac{\partial}{\partial n}nf - \kappa_2 \frac{\partial}{\partial v}(\theta - v) f + \frac{v}{2}\frac{\partial^2 f}{\partial s^2} + \frac{\phi^2 v}{2}\frac{\partial^2 f}{\partial n^2} + \frac{\gamma^2}{2}\frac{\partial^2}{\partial v^2}vf \\
&+ \phi v\frac{\partial^2 f}{\partial s \partial n} + \gamma\rho\phi \frac{\partial^2}{\partial n \partial v}vf + \gamma\rho \frac{\partial^2}{\partial s \partial v}vf
\end{align*}
and the density function can be computed via a numerical procedure such as finite difference method.

More precisely,
because the variable $s$ comes up only in the derivative operators in the above equation, to reduce the dimension of the PDE, consider the Fourier transform of $f$ with respect to s.
That is
$$ \hat f(n,v,t;\psi) = \int_{-\infty}^{\infty} f(s,n,v,t)\e^{-\I \psi s}\D s$$
and the Fourier transforms of $\frac{\partial f}{\partial s}$ and $\frac{\partial^2 f}{\partial s^2}$ are $\I \psi \hat f$ and $-\psi^2 \hat f$, respectively.
Thus, by applying the Fourier transform to the PDE, 
\begin{align*}
\frac{\partial \hat f}{\partial t} ={}& (\gamma\rho\phi + \kappa_1 n + \I \phi \psi v) \frac{\partial \hat f}{\partial n} + \frac{\phi^2 v}{2}\frac{\partial^2 \hat f}{\partial n^2} + \left\{-\kappa_2(\theta-v)+\gamma^2 + \I\gamma\rho\psi v \right\}\frac{\partial \hat f}{\partial v} + \frac{\gamma^2}{2}v\frac{\partial^2 \hat f}{\partial v^2} \\
&+ \gamma\rho\phi v \frac{\partial^2 \hat f}{\partial n \partial v} + \left\{ \I\psi (\gamma \rho - n) - \frac{\psi^2}{2}v + \kappa_1 + \kappa_2 \right\} \hat f.
\end{align*}
The transformed function $\hat f$ can be computed by numerical procedures and
the probability density function of the price is generated by applying inverse Fourier transform to the computed $\hat f$.
The distributions of the diffusion model and the simulated histograms of the corresponding Hawkes process are compared in Figure~\ref{Fig:hist}.
The parameter settings are 
$$\mu = m = 0.09, \alpha_s = a_s = 0.6, \alpha_c = a_c = 0.3, \beta = b = 2.5, \delta = 0.2, S_0 = 1000$$
and the time horizon is 30 seconds.

\begin{figure}
\centering
\includegraphics[width=0.45\textwidth]{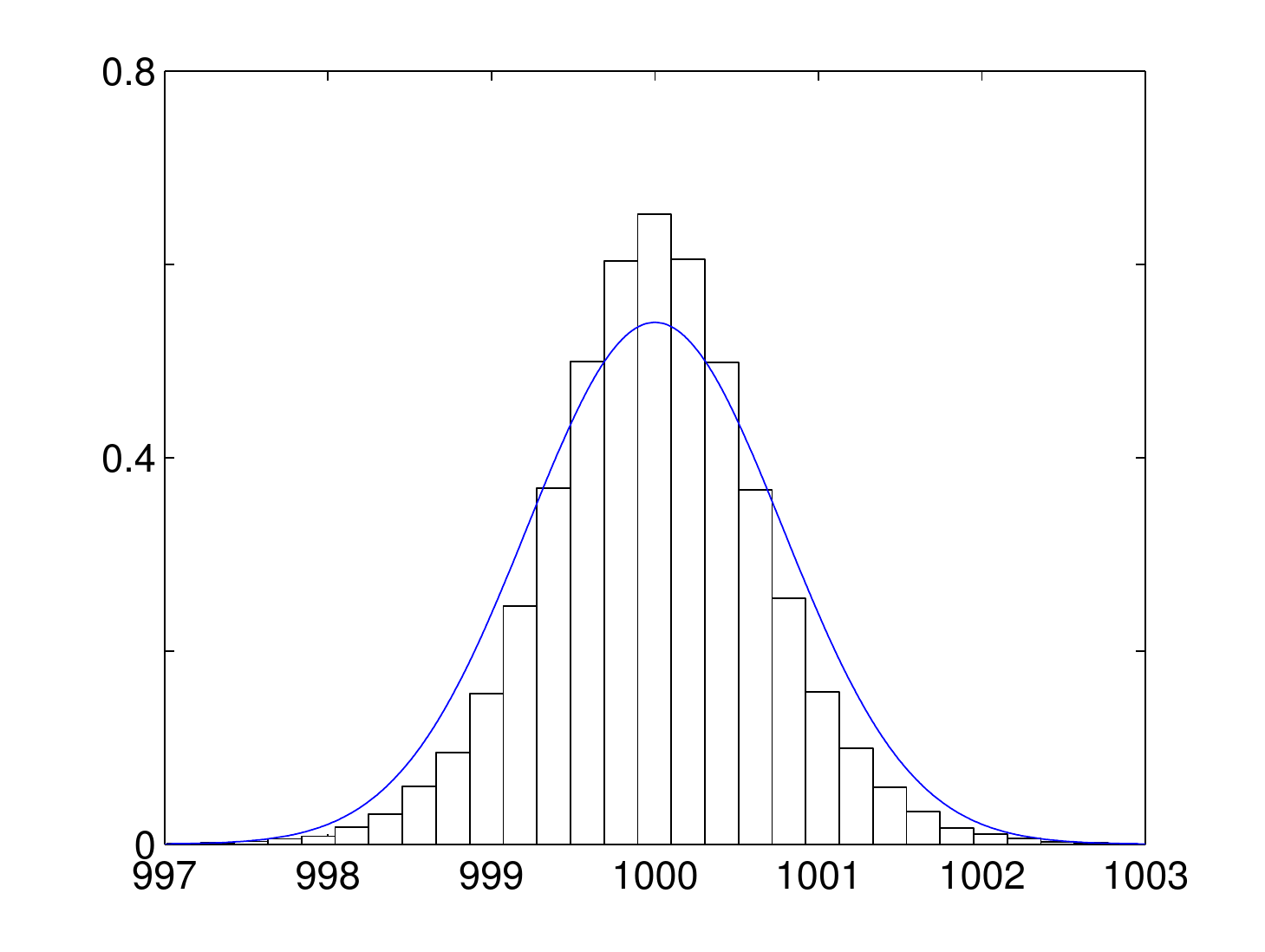}
\caption{Numerically computed probability density function of the price driven by the diffusion model and a histogram of the Hawkes model price by the simulation with 30 seconds (right)}\label{Fig:hist}
\end{figure}

Second, the derivation of the variance formula in the diffusion model is relatively simple compared to the symmetric Hawkes model due to the analytical simplicity of the diffusion processes.
To derive the variance formula of the return, for simplicity, it is assumed that the variance process $V_t$ is in the stationary state at time 0.
This is a similar assumption that the intensity processes in the symmetric Hawkes model are in the stationary state at time 0.
If $V_t$ is in the stationarity state at time 0, then 
$$ V_0 = \theta = \frac{2bm\delta^2}{b-a_s-a_c} $$
and since $\E[V_s] = \theta$,
$$ \int_0^t \E[V_s] \D s = \frac{2bm\delta^2t}{b-a_s-a_c}. $$
Similarly, if the mean process $n_t$ is in the stationarity state at time 0, then $n_0 =0$.

\begin{proposition}\label{Prop:voldiff}
Assume that the price process $S$ follows Eq.~\eqref{Eq:S_diff} and the instantaneous variance $V$ and mean processes $n$ are in the stationary state at time 0.
The variance of the return is
$$\Var\left( \frac{S_t - S_0}{S_0} \right) = \frac{1}{S^2_0} \left\{ \frac{\phi ^2\theta\left( - \e^{-2\kappa_1 t} + 4 \e^{-\kappa_1 t} - 3 +2  \kappa_1 t\right)}{2 \kappa_1^3} + \frac{2\phi \theta \left(\kappa_1 t-1+ \e^{-\kappa_1 t}\right) }{\kappa_1^2} + \theta t\right\}.$$
\end{proposition}
\begin{proof}
See \ref{append:proof1}.
\end{proof}
If $t$ is sufficiently large, then the variance is approximated by
$$\Var\left( \frac{S_t - S_0}{S_0} \right) \approx \frac{1}{S^2_0} \left( \frac{\phi ^2\theta   t}{ \kappa_1^2} + \frac{2\phi \theta t}{\kappa_1} + \theta t\right) = \frac{b^2 \theta t}{S_0^2 \kappa_1^2}$$
which is analogous to Remark~\ref{Rmk:reparm}.
When $\phi = 0$, i.e., the drift of the price process is zero, the variance of the return is simply $\theta t/S_0^2$.
The volatility of the diffusion model computed by Proposition~\ref{Prop:voldiff} and the volatility of the symmetric Hawkes model computed by Proposition~\ref{Prop:vol} were compared with the following parameter settings
$$\alpha_s = a_s = 1.2, \alpha_c=a_c=0.3, \beta=b=2.2, \mu=m=0.01, \delta/S_0 = 0.002$$
in Figure~\ref{Fig:volcomp}.
The volatilities were not annualized to show the increasing shape with time.
The two volatilities are quite close to each other.

\begin{figure}
        \centering
        \begin{subfigure}[b]{0.45\textwidth}
                \includegraphics[width=\textwidth]{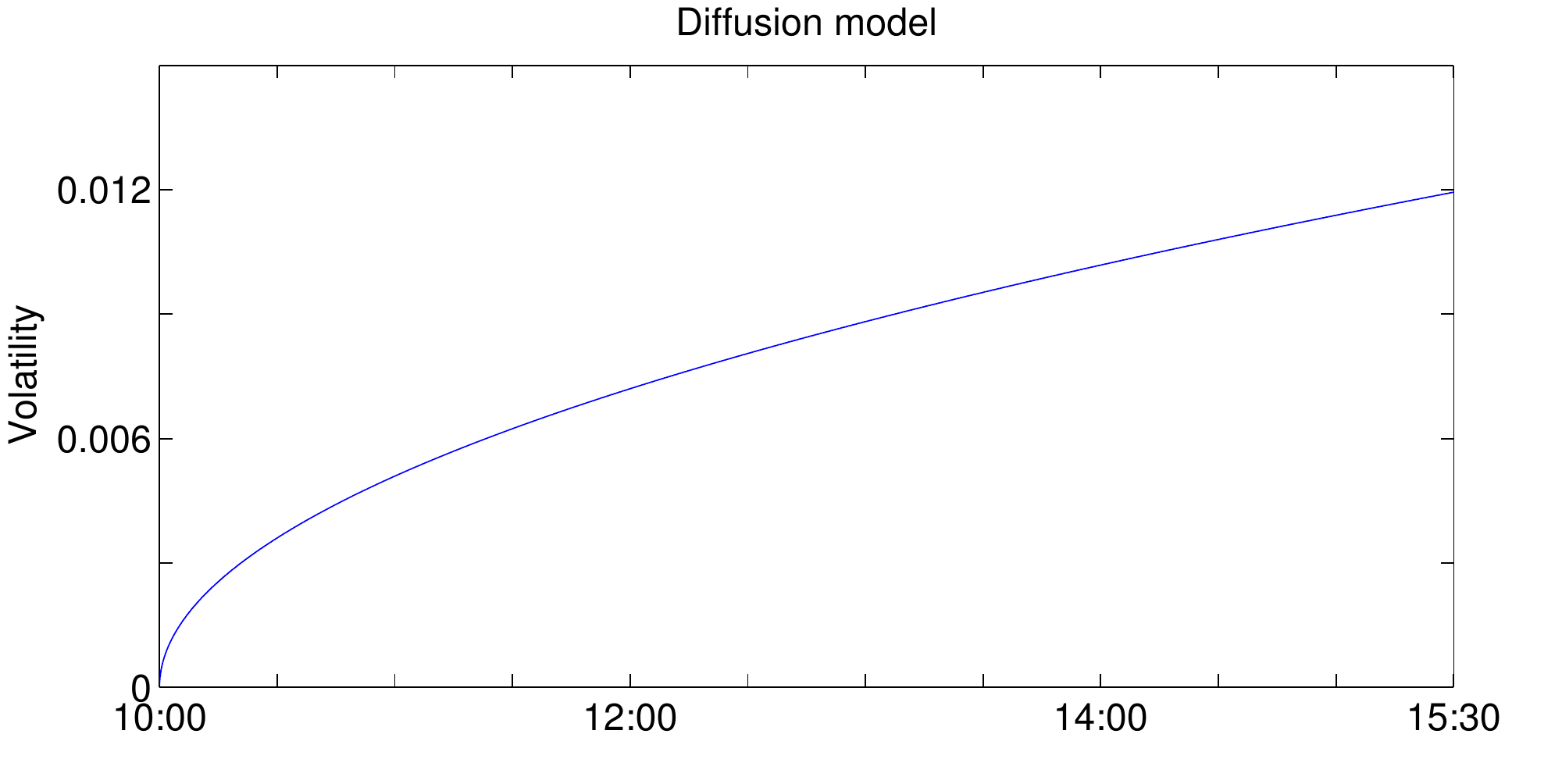}
                \caption{Volatility by the diffusion model}
                \label{Fig:vol_Diffusion}
        \end{subfigure}
        \centering
        \begin{subfigure}[b]{0.45\textwidth}
                \includegraphics[width=\textwidth]{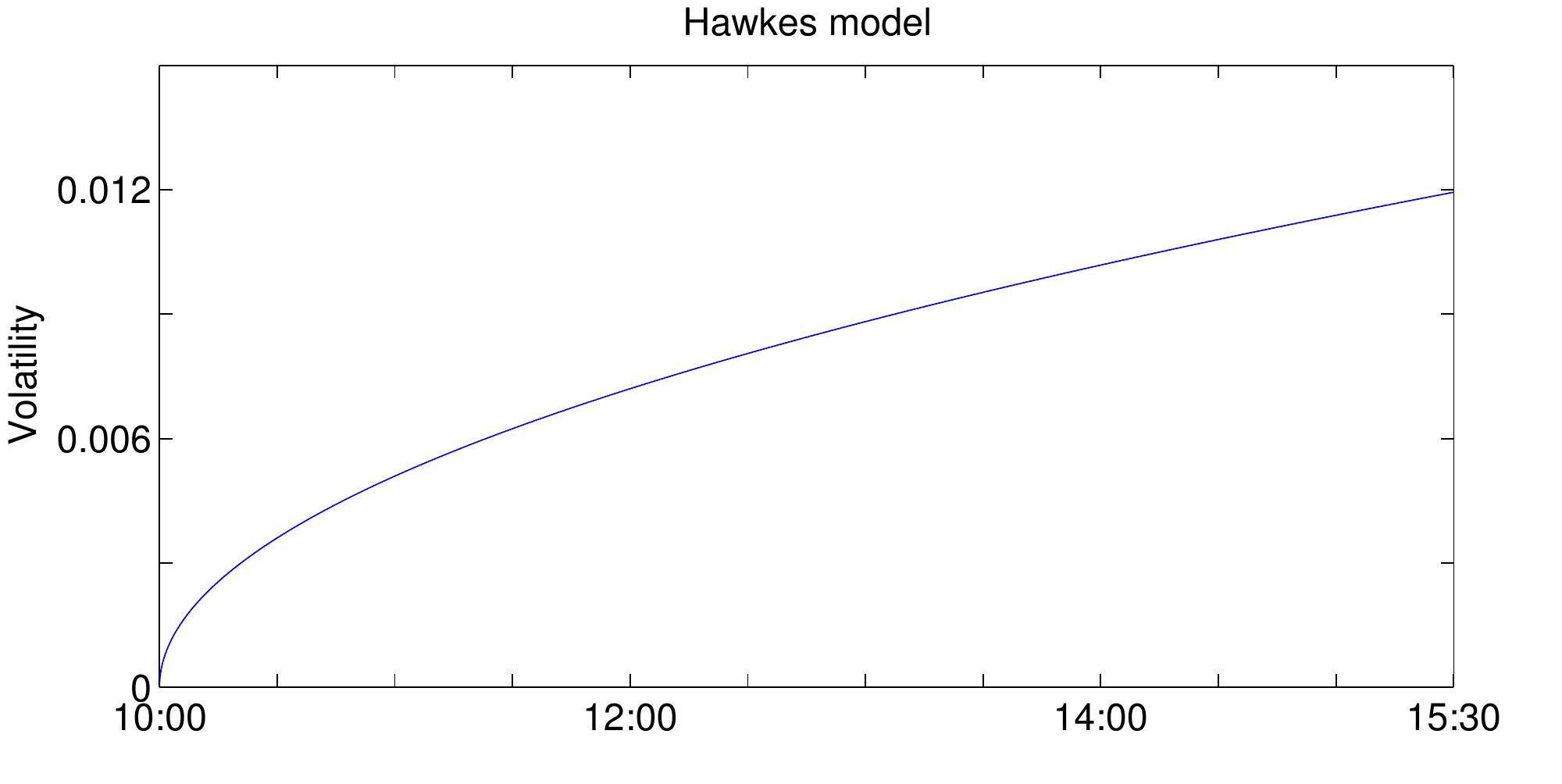}
                \caption{Volatility by the symmetric Hawkes model}
                \label{Fig:vol_Hawkes}
        \end{subfigure}
        \caption{The comparison between the volatility computed by the diffusion model and the symmetric Hawkes model}\label{Fig:volcomp}
\end{figure}

Figure~\ref{fig3_vol_surface_diff}  shows the annualized volatility surface as a function of $\kappa_1$ and $\phi$ with a fixed $\theta = 4\times10^{-9}$.
With a fixed $\kappa_1$, with increasing $\phi = a_s - a_c$, (when the self-excited coefficient $a_s$ is larger than the mutually excited coefficient $a_c$)
the volatility increases.
This result is expected because the self-excited coefficient is related to trade clustering.

The increasing rate of the volatility with respect to $\phi$ depends on the level of $\kappa_1$.
Because $\kappa_1 = b -\phi$, with a fixed $\phi$, a large $\kappa_1$ implies a large $b$ and a short persistence.
In addition, a small $\kappa_1$ implies a small $b$ and a long persistence.
Therefore, when $\phi<0$, i.e., the mutually excited effect is larger than the self excited effect, 
a longer persistence (smaller $\kappa_1$) of the mutually excited effect implies a smaller volatility and a shorter persistence implies a larger volatility.
On the other hand, when $\phi>0$, i.e., the self excited effect is larger than the mutually excited effect,
a longer persistence of the self excited effect implies a larger volatility and a shorter persistence implies a smaller volatility.
This contrast is visualized in Figure~\ref{fig3_vol_surface_diff} with the different sign of the slope of the volatility with respect to $\kappa_1$ depending on whether $\phi>0$.

\begin{figure}
\centering
\includegraphics[width=0.5\textwidth]{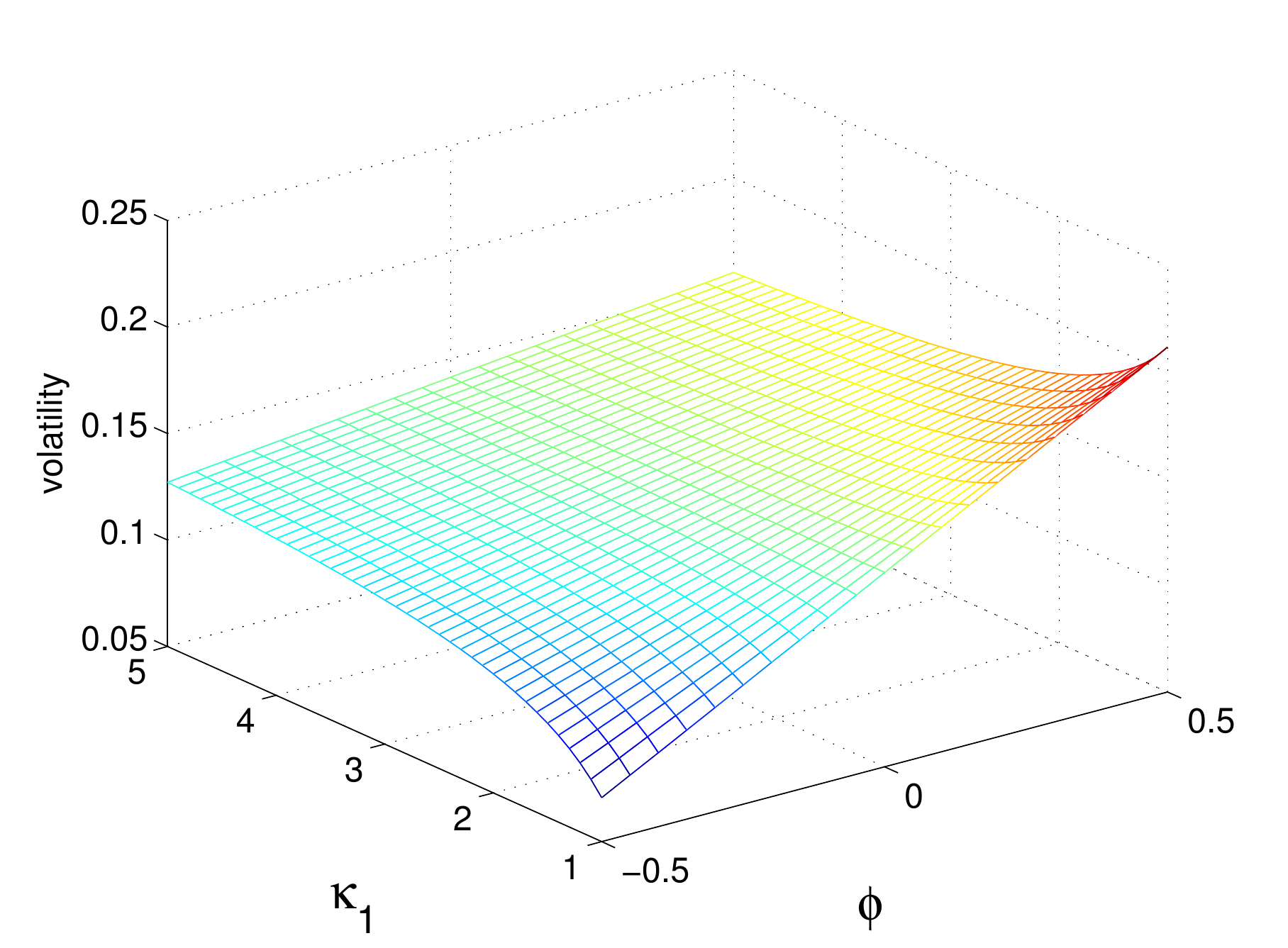}
\caption{Annualized volatility surface as a function of $\kappa_1$ and $\phi$}\label{fig3_vol_surface_diff}
\end{figure}

Because there are many studies focused on the behavior of the realized variance in the presence of the market microstructure noise,
this study also examined the realized variance under the diffusion model.
Consider a discretized time interval $[0,T]$ with time step $\tau$.
For convenience, let $T/\tau$ be an integer.
The realized variance is defined by the finite sum approximation to the quadratic variation of the return over a time interval.
The signature plot over a fixed interval is the realized variance over the interval defined as a function of $\tau$:
$$ \widehat C(\tau) = \frac{1}{T} \sum_{n=0}^{T/\tau} (R_{(n+1)\tau} - R_{n\tau})^2 $$
where $R = (S_t - S_0)/S_0$ is the return process.
(Depending on the context,  $R$ could be the log-return process.)

The above formula is indeed the definition of the realized variance, which is the consistent estimator of the true variance of the return in the absence of microstructure noise by the semimartingale theory.
On the other hand, empirical studies showed that the realized variance depends on the size of the partition $\tau$ due to the microstructure noise or clustering property \citep{Hansen, Fonseca2014Self}.
For the diffusion model, under the stationarity assumption of the time series of the squared return, $(R_{(n+1)\tau} - R_{n\tau})^2$, and 
the mean signature plot is
\begin{align*}
C(\tau) &= \E[ \widehat C(\tau)] = \frac{1}{\tau} \E [ (R_{(n+1)\tau} - R_{n\tau})^2 ] \\
&= \frac{1}{\tau S^2_0} \left\{ \frac{\phi ^2\theta\left( - \e^{-2\kappa_1 \tau} + 4 \e^{-\kappa_1 \tau} - 3 +2  \kappa_1 \tau\right)}{2 \kappa_1^3} + \frac{2\phi \theta \left(\kappa_1 \tau-1+ \e^{-\kappa_1 \tau}\right) }{\kappa_1^2} + \theta \tau\right\}.
\end{align*}

Figure~\ref{Fig:signature} shows the mean signature plots with various $\phi$ in the left and $\kappa_1$ in the right.
For parameter settings, $S_0 = 1, \theta = 2\times10^{-8}$ and $\kappa_1 = 0.5$ in the left and $\phi = -0.3$ in the right.
With negative values of $\phi$, implying $a_s < a_c$ and a more pronounced self-excited effect, the mean signature plot increases as $\tau$ approaches zero.
On the other hand, with positive $\phi$, implying $a_s < a_c$ and a more pronounced self-excited effect, the mean signature plot decreases as $\tau$ approaches zero.
In both cases, when $\tau$ is too small, there is bias between the realized variance and the true variance,
which is in contrast to the traditional understanding in statistics that a more exact result is obtained with a large sample size.
In addition, with a sufficiently large $\tau$, the expected realized variances converge.

\begin{figure}
\centering
\includegraphics[width=0.45\textwidth]{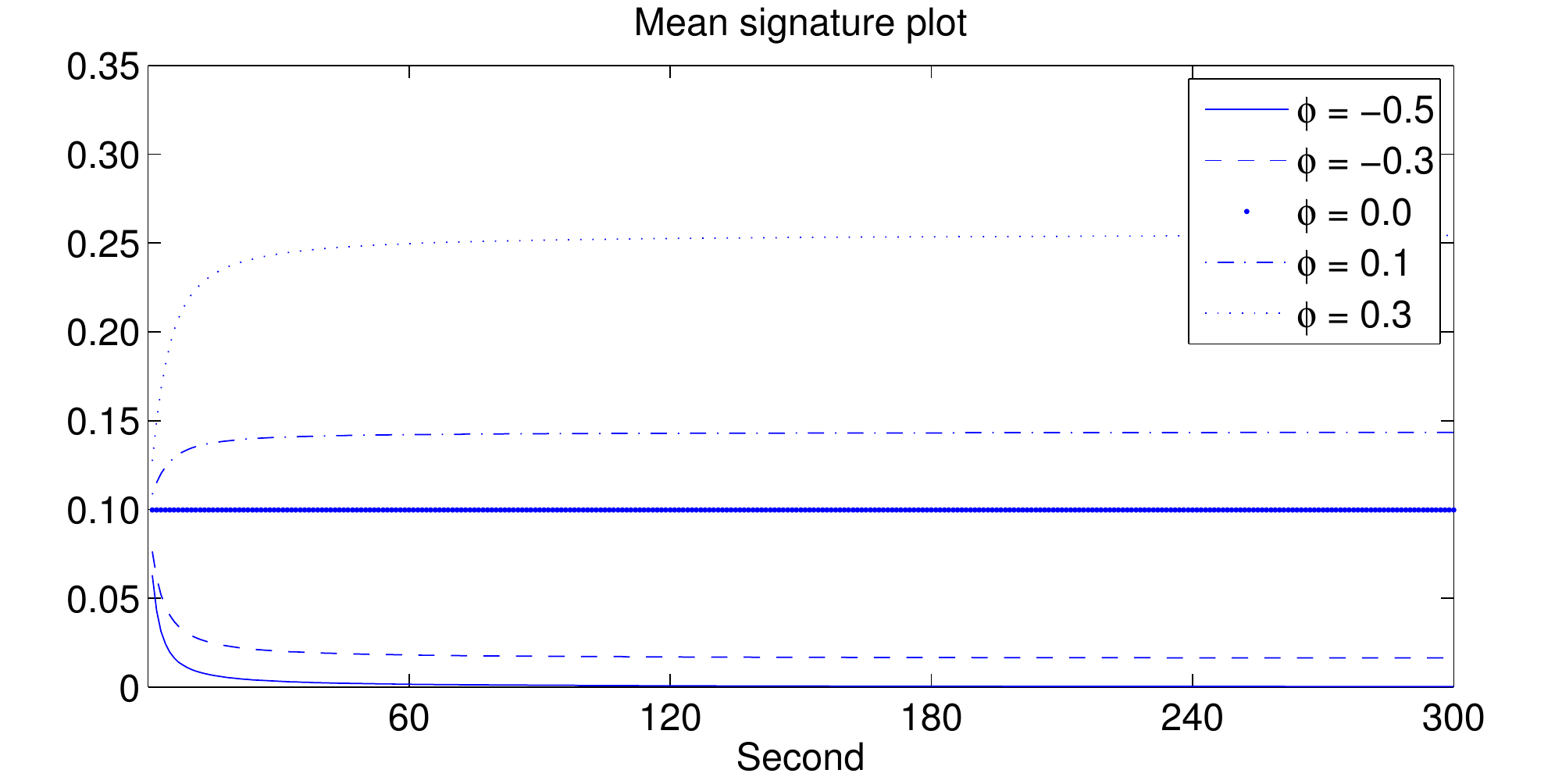}
\includegraphics[width=0.45\textwidth]{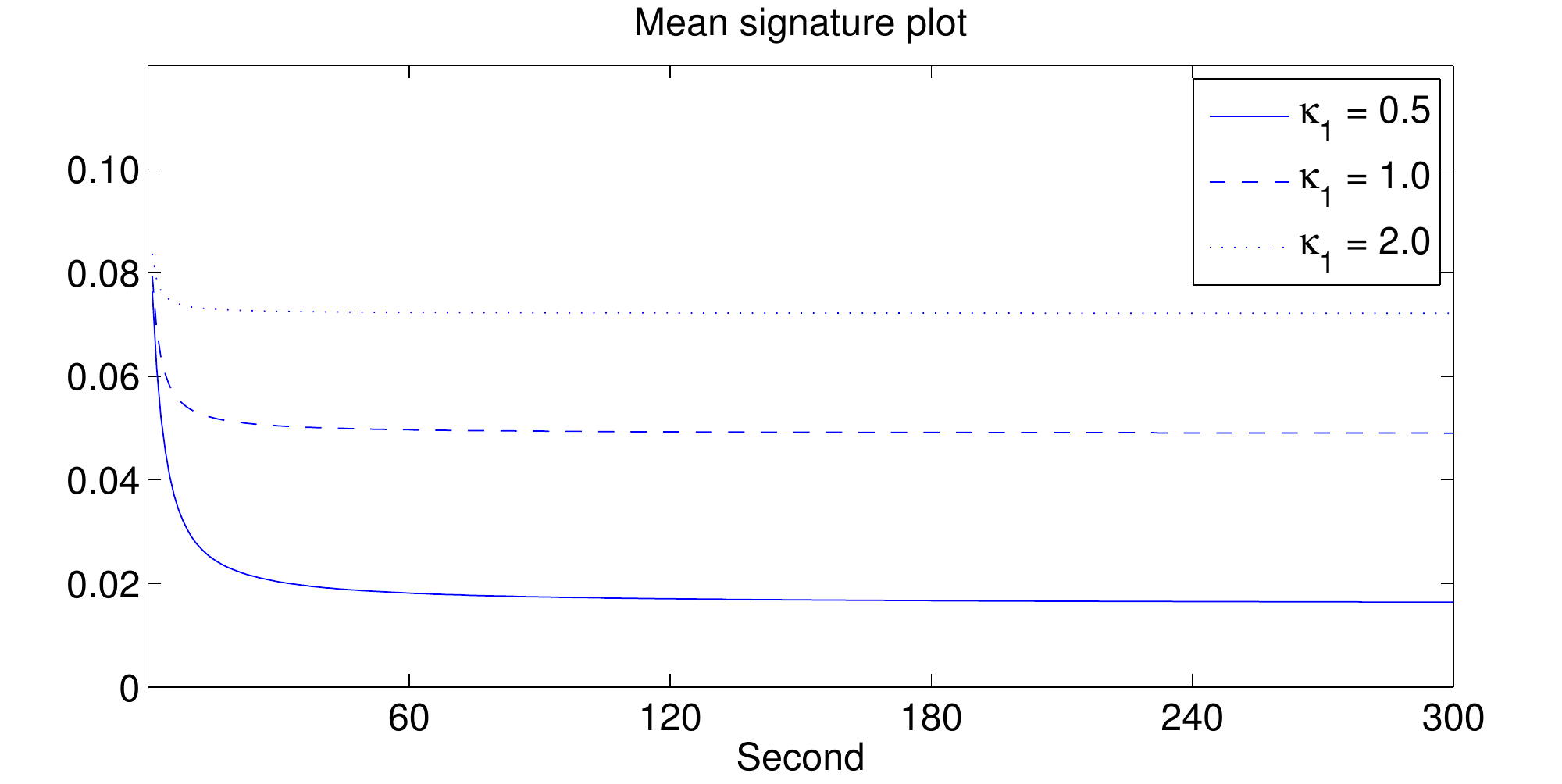}
\caption{Mean signature plot with fixed $\kappa_1=0.5$ with various $\phi$ (left) and fixed $\phi=-0.3$ and various $\kappa_1$ (right)}\label{Fig:signature}
\end{figure}

Third, as mentioned before, the asymmetry in the price distribution can be introduced easily 
with the leverage parameter $\rho$.
This method is a natural extension of the method used to introduce asymmetry in a macro level price dynamics.
The asymmetry in the Hawkes model is an ongoing research topic, for example, consult \cite{omar2016microstructural}.
In our notation and setting, the asymmetric Hawkes model in \cite{omar2016microstructural} can be regarded as
the Hawkes model of Eqs~\eqref{Eq:lambda1}~and~\eqref{Eq:lambda2} with
$$ \alpha_{12} = \eta \alpha_{21}, \quad \alpha_{22} = \alpha_{11} + (\eta-1) \alpha_{21} $$ 
where $\eta$ is a newly introduced parameter.
It is believed that there are many possible ways to incorporate asymmetry into the Hawkes model.

In general, estimating $\rho$ in the diffusion models is not trivial \citep{ait2013leverage}.
One method for estimating $\rho$ is to use the method of moment as in \cite{Lee2016}.
Let $[X,Y]$ denote the quadratic covariation process between the processes $X$ and $Y$, i.e.,
\begin{align*}
[X,Y]_t &= X_t Y_t - \int_0^t X_{s-} \D Y_s - \int_0^t Y_{s-} \D X_s\\
&= X_0 Y_0 + \lim_{||\pi_n|| \rightarrow 0} \sum_i (X_{i+1} - X_{i})(Y_{i+1} - Y_{i})
\end{align*}
for a sequence of random partitions $\pi_n$ with a limit in probability.
The third moment variation of the return $[R^2,R]_t$ introduced by \cite{ChoeLee} is a useful quantity to measure the skewness of the return distribution.
In addition, the tractability of the diffusion process enable us to easily derive the following formula. 

\begin{proposition}\label{Prop:tmv}
Under the stationarity condition of the variance process with time 0, the following moment condition can be derived
\begin{align*}
\E[ [R^2,R]_t] ={}& \rho K
\end{align*}
where
\begin{align*}
K ={}& \frac{1}{S_0^3} \left[ \frac{2\gamma \theta}{\kappa_2^2}\left(  \kappa_2 t - 1 + \e^{-\kappa_2 t} \right) + \frac{2\gamma\phi\theta}{\kappa^3}
\left( \frac{\kappa_1^2}{2} t^2 - \kappa_1 t + 1 - \e^{-\kappa_1 t} \right) \right.\\
&- \frac{ 2\gamma\theta\phi }{\kappa_1^2 \kappa_2 (\kappa_1 +\kappa_2)}\left\{ (-\kappa_1^2-\kappa_2^2-\kappa_1\kappa_2)t + \frac{1}{2} (\kappa^2\kappa_2 + \kappa_1\kappa_2^2)t^2 \right.\\
&\left. \left. - \frac{  \kappa_2^2+\kappa_1\kappa_2}{\kappa_1} (\e^{-\kappa_1 t}-1) - \frac{\kappa^2+\kappa_1\kappa_2}{\kappa_2} (\e^{-\kappa_2 t}-1) + \frac{\kappa_1\kappa_2}{\kappa_1+\kappa_2}(\e^{-(\kappa_1+\kappa_2) t}-1)\right\}\right].
\end{align*}
\end{proposition}
\begin{proof}
See \ref{sect:proof2}.
\end{proof}
If the drift part in the price process is zero, i.e., $\phi=0$, then the expectation of the third moment variation is simply
\begin{align*}
\E[ [R^2,R]_t] = \frac{2\gamma \rho \theta}{S_0^3 \kappa_2^2  }\left(  \kappa_2 t - 1 + \e^{-\kappa_2 t} \right) \approx \frac{2\gamma \rho \theta}{S_0^3 \kappa_2  } t
\end{align*}
where the approximation is for a sufficiently large $t$.

\begin{example}
By Proposition~\ref{Prop:tmv}, $\frac{1}{N}\frac{\sum \widehat{[R^2,R]}_i}{K} \rightarrow \rho$ as the sample size increases where $\widehat{[R^2,R]}_i$ denotes the realized finite sum approximation of the third moment variation.
The convergence of the estimates of $\rho$ in Figure~\ref{Fig:rho_simul} were plotted in a simulation study with parameter settings 
$\kappa_1 = 1.15, \phi= 0.45, \theta=2.8\times10^{-4}, \kappa_2=0.85, \gamma=0.0375,\rho=-0.5$.
The sample mean of $\widehat{[R^2,R]}_i /K$ converges to $\rho$.
In the simulation result, the sample mean is $-0.4935$ with the standard error of $0.0757$.

\begin{figure}
\centering
\includegraphics[width=0.5\textwidth]{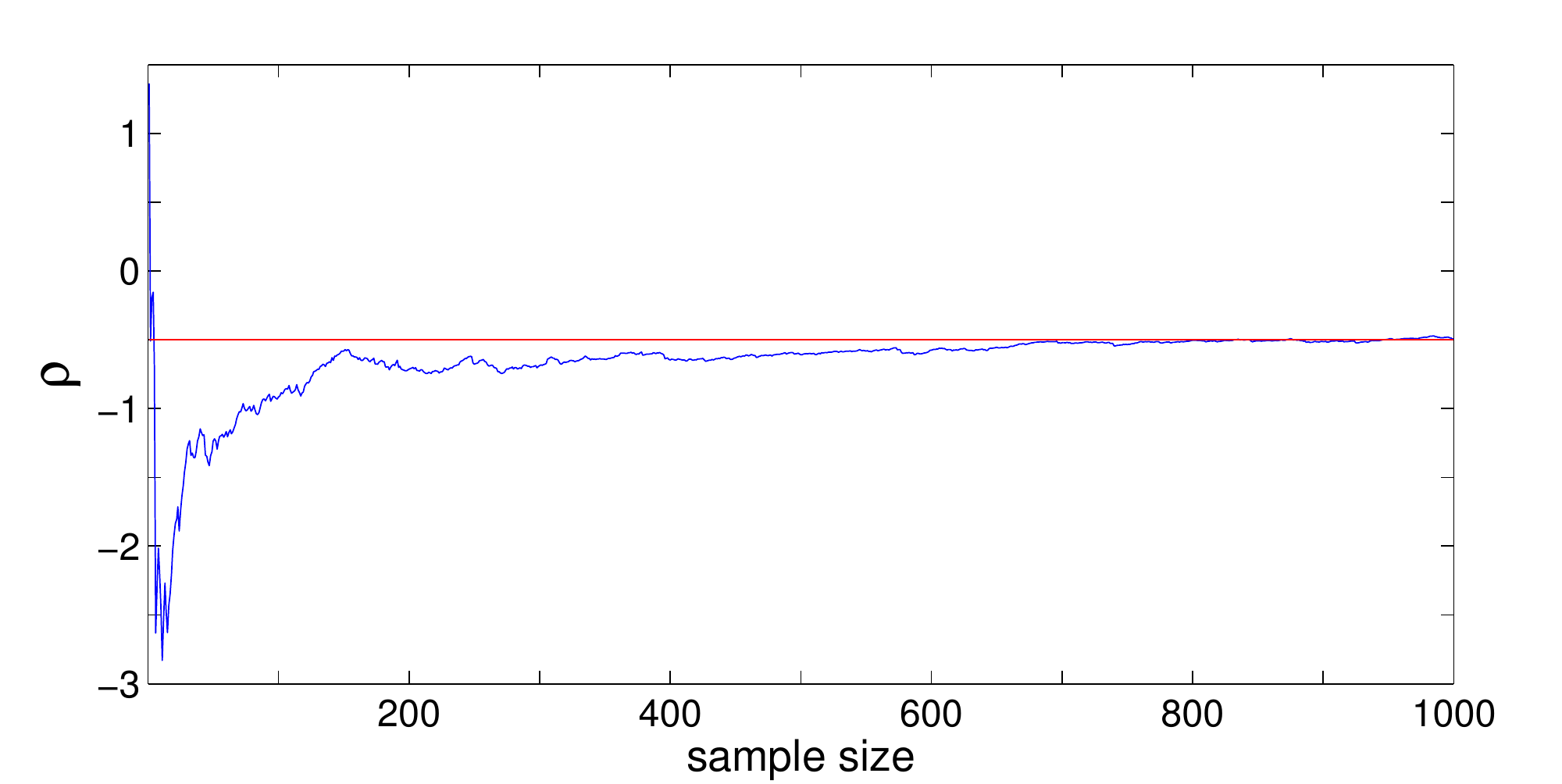}
\caption{Convergence of the estimates of $\rho$}\label{Fig:rho_simul}
\end{figure}
\end{example}

However, it should be noted that the number of samples should be sufficient for the convergence.
If the number of samples is not sufficient, it is better to use the approximate likelihood method or the simulated likelihood estimate discussed in \ref{Subsect:SLE}.

\begin{figure}
\centering
\includegraphics[width=0.6\textwidth]{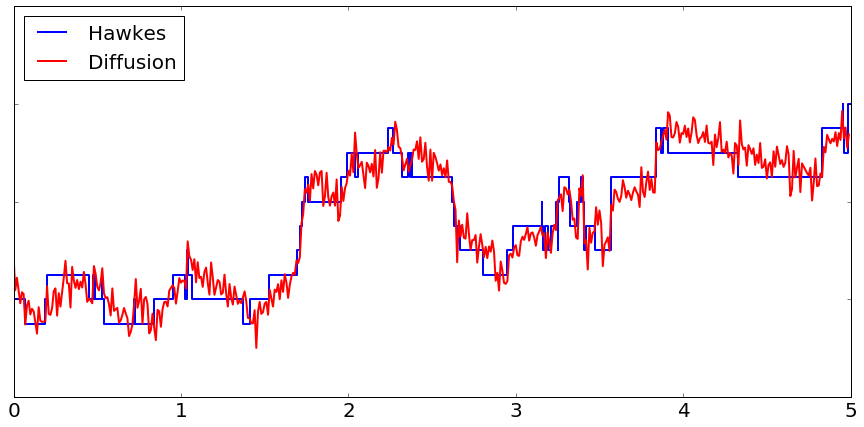}
\caption{Hawkes model and diffusion analogy}\label{Fig:HD}
\end{figure}

\subsection{Comparison}

Both the Hawkes and the diffusion models well describe the microstructure of price dynamics such as trade clustering or microstructure noise.
The Hawkes model directly describes the tick-by-tick structure of the asset price and 
data is applied to the model  without further assumptions or data corrections.
The model's closed-form formula of the log-likelihood function and 
quite reliable numerical algorithms to find the maximum the applicable.

On the other hand, the diffusion approach naturally extends the methodology traditionally used to describe asset price movements.
Note that the diffusion model in our paper is not a rigorous mathematical transform of the Hawkes model.
We use the derivation to provide an intuition not a mathematical proof.
Thus, one can argue about the legitimacy of the model, for example, the introduction of $\rho$ which we regarded as a constant.

Nevertheless, the model inherits the advantages of typical diffusion models.
Based on the It\'{o} calculus and PDE approach, the derivations of useful formula such as moment conditions and distributional property are simpler than the Poisson based Hawkes models.
Since the diffusion model has been extensively studied for a long time, it is expected that there will be a more convenient aspect to apply the existing theory or extend the model.

Meanwhile, the maximum likelihood estimation for the diffusion model is generally more complicated because the closed-form formula for the density function is not available in many cases.
In the absence of the closed-form likelihood function, 
the expansion based likelihood function approach \citep{ait2008closed}, simulation based method \citep{Brandt2002} or the generalized method of moment \citep{garcia2011estimation, bollerslev2011dynamic} are used to estimate the parameters.

\subsection{Simulated likelihood estimation}\label{Subsect:SLE}

Because the exact likelihood formula of the diffusion process in this paper is barely available, 
the estimation is based on the simulation method proposed by \cite{Brandt2002}.
Briefly explaining the method, the interval between two observed points, $t_i$ and $t_{i+1}$, are divided into subintervals with a length $N$.
The $M$ number of paths are simulated from $t_i$ up to $N-1$ subintervals using the discretized version of the diffusion model.
The mean of the transition probability functions from the last values of the simulated paths to the observed value at $t_{i+1}$, which is approximated by the normal distribution based on the discretization, 
becomes the maximum simulated likelihood.

In Empirical studies, the data is reformulated to apply the diffusion model
because the original data is based on a tick structure.
The large interval, i.e., $t_{i+1} - t_{i}$, is set to one minute where a sufficiently large number of events are observed for the approximation.
Figure~\ref{Fig:HD} presents the procedure, with every one minute, the observed price is the base point to construct a diffusion process, which lies behind the tick structure. 
Within the interval, the paths of the discretized version of the diffusion model are simulated with 60 subintervals.

\section{Empirical study}\label{Sect:empirical}
\subsection{Data}
For empirical studies, ultra high-frequency data of 10 stocks in the S\&P 500 are used.
As raw data in the first place, we reorganize the data in the following way:
\begin{itemize}
\item
The historical data consists of the best bid, ask quotes of the stocks, and their dynamics over trading time with various exchanges.
\item
The mid-price dynamics of the best bid and ask quotes of each stock reported in the New York Stock Exchange (NYSE) from 10:00 to 15:30 are selected to avoid the seasonal effects observed in early or late in the market.
\item
In the original raw data, the time stamps have 1 second resolutions. If the prices changes are reported several times for one second, the price changes with equidistant intervals are redistributed over one second.
\item
The mid-price increments and decrements have a unit size of change that is the half of the minimal bid ask spread.
If a price increment or decrement is larger than the minimal unit size, 
the change is considered to be the sum of the successive movements with the minimal size.
In recent data, the percentage of the minimal change is very high in many symbols, as listed in Table~\ref{Table:min}.
In addition, Table~\ref{Table:min_tran} lists the percentage of minimal change of transacted prices where similar patterns to the percentage of the mid-prices are observed.
\item
The symbols in the table represents:\\
BAC - Bank of America Corp,
CVX - Chevron,
GE - General Electric Co.,
IBM - International Business Machines,
JPM - JP Morgan Chase \& Co,
KO - The Coca-Cola Company,
MCD - McDonald's Corp,
T - AT\&T Inc,
VZ - Verizon Communications Inc,
XOM - Exxon Mobil Corp
\end{itemize}

\begin{table}
\caption{minimal tick percentage (\%) - mid price}\label{Table:min}
\centering
\begin{tabular}{ccccccccccc}

\hline
& BAC & CVX & GE & IBM & JPM & KO & MCD & T & VZ & XOM\\
\hline
2007 & 79.91 & 61.88 & 83.29 & 55.70 & 78.16 & 75.41 & 73.71 & 79.15 & 84.57 & 70.14\\
2008 & 87.83 & 59.01 & 79.93 & 43.65 & 59.68 & 67.48 & 58.07 & 68.58 & 69.48 & 68.57\\
2009 & 88.19 & 70.36 & 93.84 & 56.98 & 72.42 & 80.10 & 84.51 & 82.14 & 82.06 & 86.79\\
2010 & 79.78 & 87.83 & 98.74 & 77.48 & 95.77 & 94.61 & 83.88 & 82.40 & 82.71 & 86.98\\
2011 & 99.53 & 72.16 & 99.21 & 52.73 & 96.96 & 89.44 & 86.92 & 98.23 & 89.07 & 90.35\\
\hline
\end{tabular}
\end{table}

\begin{table}
\caption{Minimal tick percentage (\%) - transacted price}\label{Table:min_tran}
\centering
\begin{tabular}{ccccccccccc}

\hline
& BAC & CVX & GE & IBM & JPM & KO & MCD & T & VZ & XOM\\
\hline
2007 & 92.70 & 69.86 & 97.26 & 64.40 & 90.34 & 89.81 & 87.23 & 94.74 & 93.11 & 79.78\\
2008 & 84.60 & 50.88 & 89.81 & 51.80 & 71.21 & 76.61 & 63.89 & 86.87 & 82.79 & 60.93\\
2009 & 98.89 & 72.34 & 98.39 & 59.98 & 88.03 & 89.15 & 80.42 & 97.44 & 93.67 & 82.47\\
2010 & 99.63 & 81.07 & 99.61 & 80.88 & 95.58 & 92.59 & 85.65 & 99.15 & 98.19 & 92.36\\
2011 & 99.81 & 62.08 & 99.68 & 57.04 & 96.76 & 91.76 & 80.58 & 99.01 & 97.12 & 84.48\\
\hline
\end{tabular}
\end{table}

\subsection{Dynamics of parameters and performance of volatility measure}~\label{subsect:paramters}
The parameters of the symmetric Hawkes process were estimated, as explained in \ref{Sect:likelihood} using the mid-price dynamics of the stocks quoted in NYSE.
The estimations are employed on a daily basis because there are enough samples even in a day and the aim is to demonstrate the daily change in the parameters.
Table~\ref{Table:GE2011} lists one of the results with GE for each day from January 3 to 27, 2011.
The estimates of $\mu, \alpha_s, \alpha_c, \beta$ and their numerically computed standard errors in the parentheses are reported.
In the estimation, the unit time, $t=1$, is one second.
The averaged daily estimates of $\mu, \alpha_s, \alpha_c, \beta$ for the different stocks over a month, January 2011, are also reported.

\begin{table}
\caption{Symmetric Hawkes estimation result, GE, January 2011}\label{Table:GE2011}
\centering
\begin{tabular}{cccccccccccc}

\hline
Date & $\mu$ & $\alpha_s$ & $\alpha_c$ & $\beta$ & H.vol  &TSRV & RRV\\
\hline
0103 & 0.0067 & 0.4661 & 1.3576 & 2.2596 & 0.0957  & 0.1289 & 0.1103\\
 & 	(0.0004) &	(0.0609) & (0.0958) & (0.1160)\\
0104 & 0.0082 &	0.4853 & 1.3941 & 2.5297 & 0.1139  & 0.1468 & 0.1344\\
 & (0.0005)	& (0.0494) & (0.0889) & (0.1104)\\
0105 & 0.0112 &	0.4741 & 1.1698 & 2.2281 & 0.1339  & 0.1825 & 0.1619\\
 &(0.0006)& (0.0402) & (0.0673) & (0.0939)\\
0106 & 0.0091 &	0.5599 & 1.0112 & 2.1822 & 0.1265  & 0.1656 & 0.1391\\
& (0.0005) & (0.0471) &	(0.0675) & (0.0958)\\
0107 & 0.0163 &	0.6973 & 0.5968 & 1.9959 & 0.1933  & 0.1998 & 0.1932\\
 & (0.0007) & (0.0407) & (0.0391) & (0.0747)\\
0110 & 0.0132 & 0.4978 & 0.7434 & 1.8366 & 0.1520  & 0.1730 & 0.1553\\
 & (0.0006) & (0.0360) & (0.0465) & (0.0780)\\
0111 & 0.0081 & 0.6959 & 0.7522 & 2.1448 & 0.1310  & 0.1414 & 0.1254\\
 & (0.0005) & (0.0593) & (0.0645) & (0.1125)\\
0112 & 0.0098 &	0.4210 & 0.7322 & 1.8399 & 0.1181  & 0.1440 & 0.1328\\
 & (0.0005) & (0.0406) & (0.0559) & (0.0965)\\
0113 & 0.0097 &	0.6512 & 0.3476 & 1.7471 & 0.1533  & 0.1275 & 0.1434\\
 & (0.0005) & (0.0526) & (0.0379) & (0.1035)\\
0114 & 0.0092 &	0.6173 & 0.5303 & 1.8327 & 0.1390  & 0.1328 & 0.1344\\
 & (0.0005) & (0.0533) & (0.0498) & (0.1097)\\
0118 & 0.0097 &	0.5122 & 0.4779 & 1.6361 & 0.1351  & 0.1235 & 0.1161\\
 & (0.0005) & (0.0449 & (0.0434) & (0.0976) \\
0119 & 0.0168 &	0.4737 & 0.5202 & 1.5618 & 0.1774 & 0.1842 & 0.1772\\
 & (0.0007) & (0.0331) & (0.0351) & (0.0772)\\
0120 & 0.0181 & 0.6944 & 0.5065 & 1.9431 & 0.2071  & 0.1913 & 0.1928\\
 & (0.0007) & (0.0428) & (0.0359) & (0.0835)\\
0121 & 0.0316 &	0.5334 & 0.5986 & 1.7352 & 0.2358  & 0.2168 & 0.2299\\
 & (0.0010) & (0.0262) & (0.0282) & (0.0556) \\
0124 & 0.0120 &	0.4355 & 0.4512 & 1.4181 & 0.1408  & 0.1249 & 0.1388\\
 & (0.0006) & (0.0352) & (0.0358) & (0.0836) \\
0125 & 0.0209 &	0.6164 & 0.4259 & 1.6312 & 0.2144  & 0.2018 & 0.2096\\
 & (0.0008) & (0.0321) & (0.0275) & (0.0640)\\
0126 & 0.0143 &	0.5053 & 0.5590 & 1.6587 & 0.1532  & 0.1146 & 0.1326 \\
 & (0.0007) & (0.0365) & (0.0390) & (0.0815)\\
0127 & 0.0147 &	0.5029 & 0.3925 & 1.4351 & 0.1687  & 0.1807 & 0.1578\\
 & (0.0007) & (0.0371) & (0.0317) & (0.0868) \\
\hline
\end{tabular}
\end{table}

\begin{table}
\centering
\caption{Averaged estimation result of symmetric Hawkes model, January 2011}\label{Table:2011}
\begin{tabular}{ccccc}
\hline
Symbol & $\mu$ & $\alpha_s$ & $\alpha_c$ & $\beta$ \\
\hline
GE & 0.0141 & 0.5480 & 0.6766 & 1.8476 \\
IBM & 0.1489 & 1.0057 & 0.5037 & 2.0986 \\
JPM & 0.0672 & 0.6330 & 0.4767 & 1.5806 \\
KO & 0.0357 & 0.6669 & 0.3153 & 1.4814 \\
MCD & 0.0478 & 0.7201 & 0.4177 & 1.6641 \\
T & 0.0153 & 0.4593 & 0.5157 & 1.4506 \\
VZ & 0.0216 & 0.7887 & 0.4669 & 1.8206 \\
XOM & 0.0691 & 0.5280 & 0.3482 & 1.2808 \\
\hline
\end{tabular}
\end{table}

In the column, `H.vol', the annualized daily volatility estimates computed by the estimates of the Hawkes parameters and using the formula in Proposition~\ref{Prop:vol} are presented.
In the column, `TSRV', the two scaled realized volatilities introduced by \cite{Zhang2005} are compared
and in the column, `RRV', the volatility estimates proposed by \cite{robert2011new} based on the uncertainty zones model are presented.
The table shows that the Hawkes volatility, TSRV and RRV have similar values all over the reported time.

Figure~\ref{Fig:EstimationGE2011} plots the dynamics of the parameters of GE, 2011.
The estimation results show evidence that the parameters of the Hawkes process, particularly $\mu$, changes with time.
The dynamics of $\mu$ with time shows the typical movements of positively autocorrelated time series, which is strongly associated with the macro feature of the volatility movement, such as the GARCH effect and stochastic volatility.
In addition, a comparison of Figures \ref{Fig:ge2011_vol}, \ref{Fig:ge2011_TSRV} and \ref{Fig:ge2011_mu} verifies that the dynamics of $\mu$ is related significantly to the dynamics of the volatility.
When the parameter $\mu$ of a day is large, the volatility of the day is large and 
when the parameter $\mu$ of a day is small, the volatility of the day is small.

In the figure, the movements of the other parameters $\alpha_s, \alpha_c$ and $\beta$, do not appear to be meaningful compared to the movement of $\mu$.
The plots also show that the volatilities computed by the symmetric Hawkes modeling and TSRV show similar patterns over the observed time period.

\begin{figure}
        \centering
        \begin{subfigure}[b]{0.45\textwidth}
                \includegraphics[width=\textwidth]{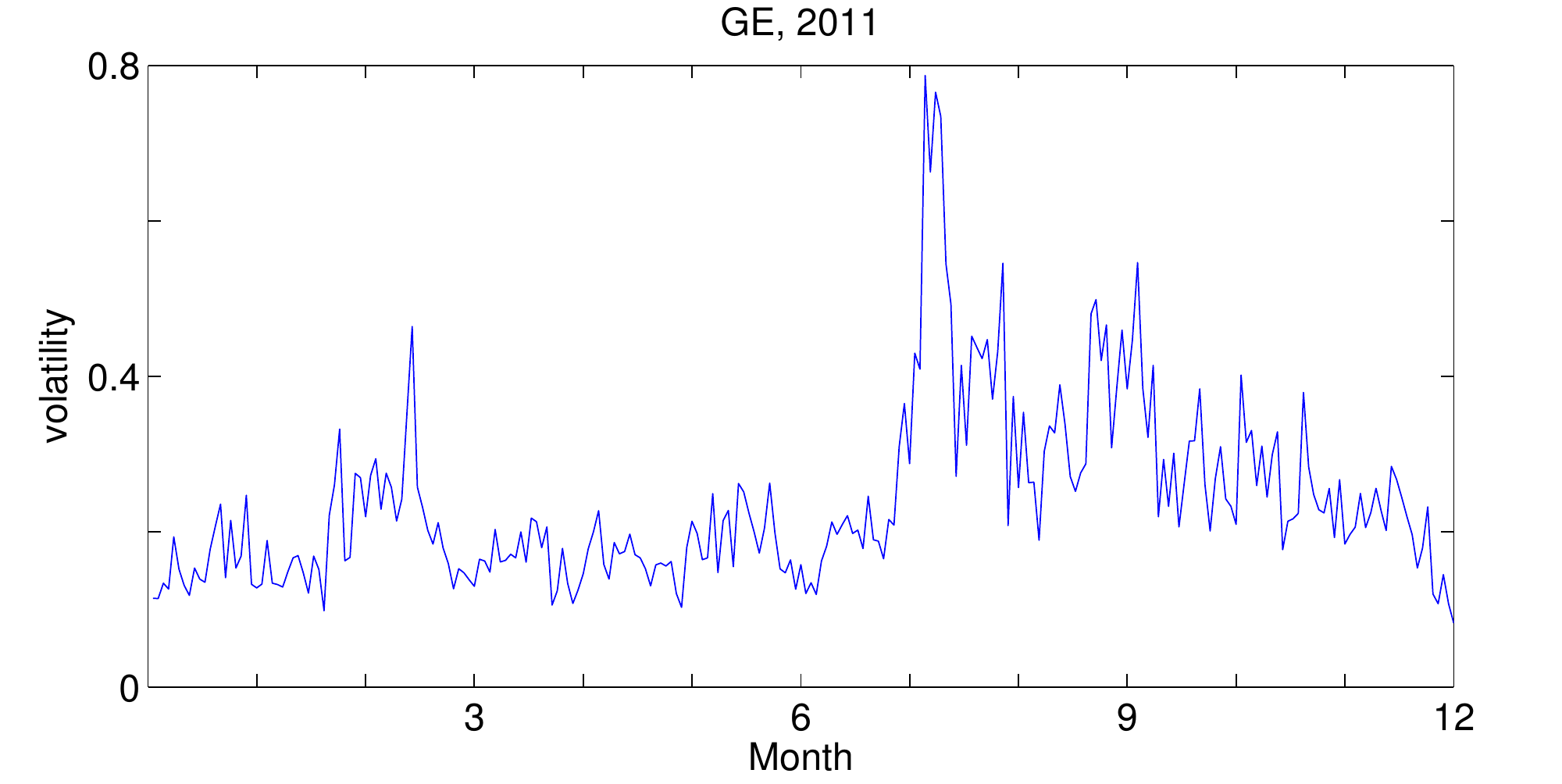}
                \caption{volatility}
                \label{Fig:ge2011_vol}
        \end{subfigure}
        \centering
        \begin{subfigure}[b]{0.45\textwidth}
                \includegraphics[width=\textwidth]{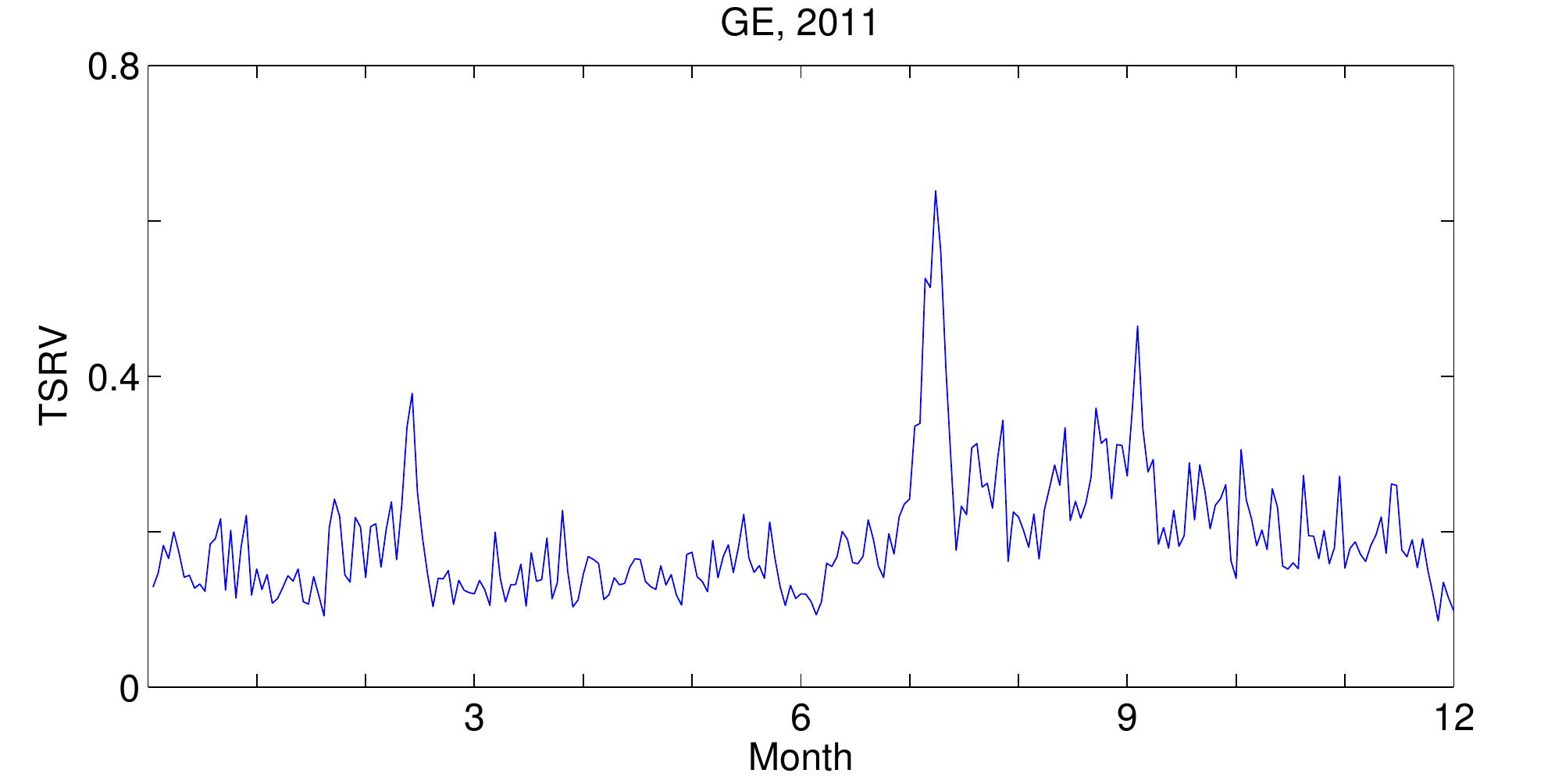}
                \caption{TSRV}
                \label{Fig:ge2011_TSRV}
        \end{subfigure}
	    \centering
        \begin{subfigure}[b]{0.45\textwidth}
                \includegraphics[width=\textwidth]{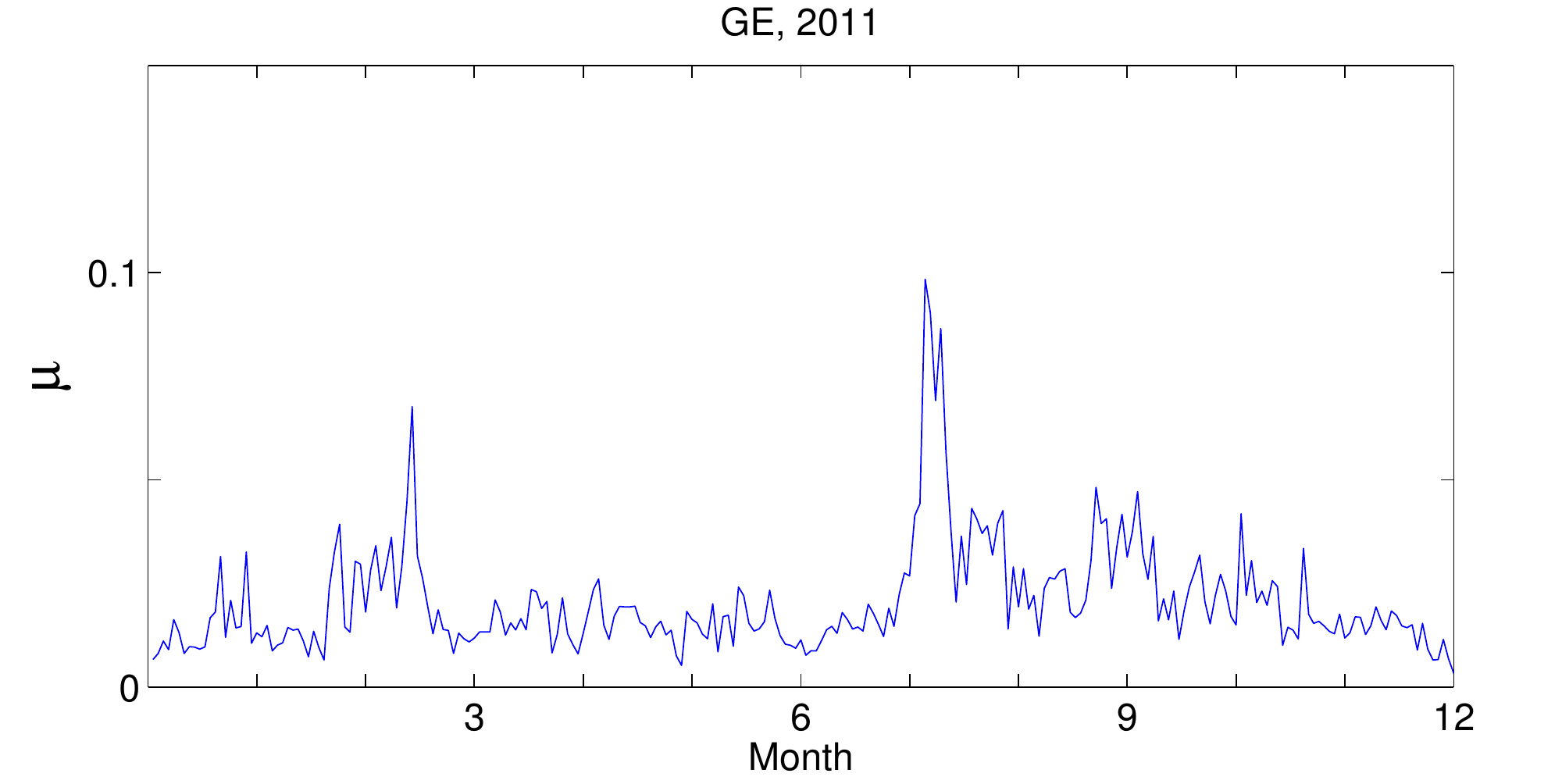}
                \caption{$\mu$}
                \label{Fig:ge2011_mu}
        \end{subfigure}
	    \centering
        \begin{subfigure}[b]{0.45\textwidth}
                \includegraphics[width=\textwidth]{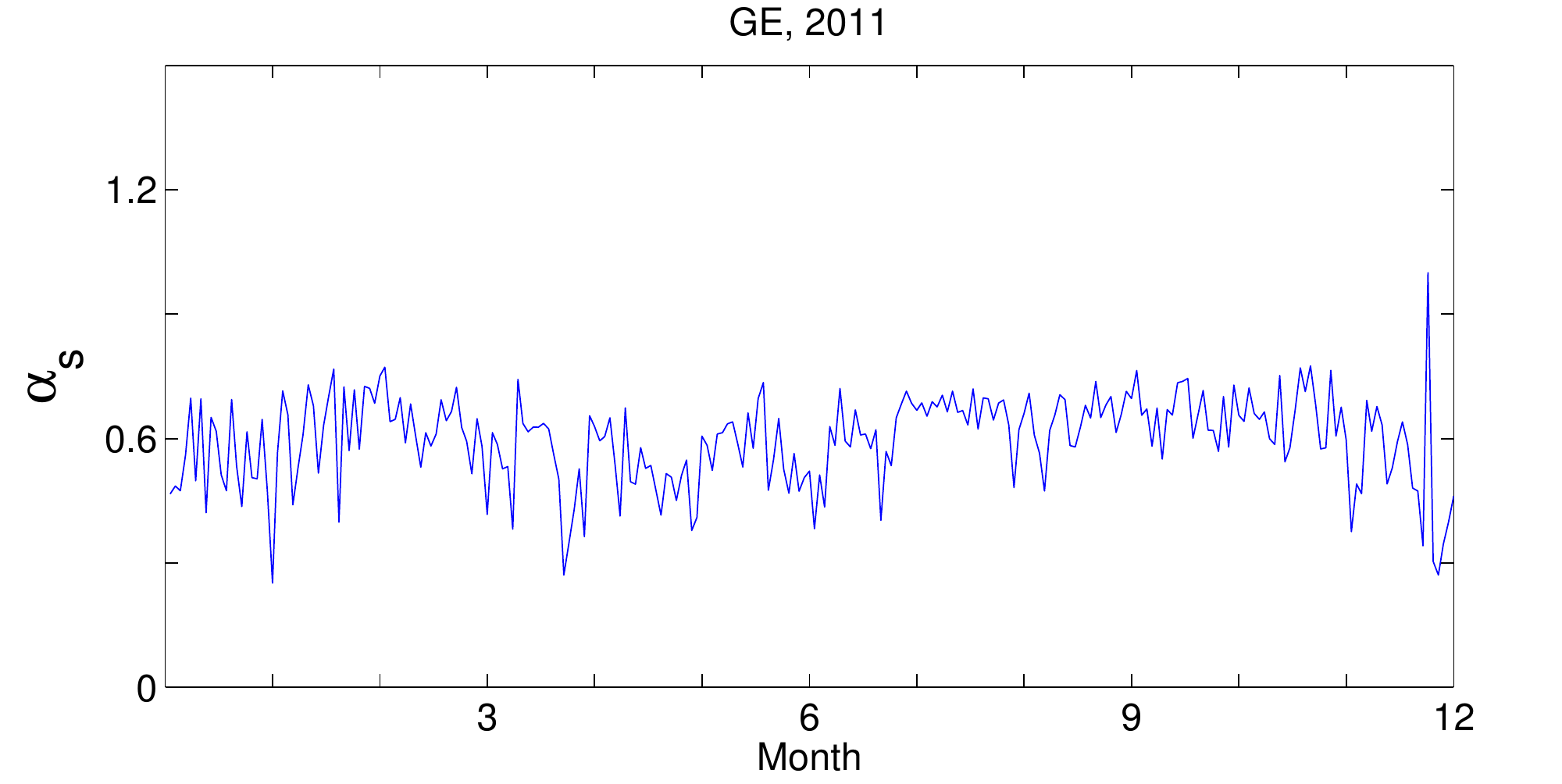}
                \caption{$\alpha_s$}
                \label{Fig:ge2011_alphas}
        \end{subfigure}
	    \centering
        \begin{subfigure}[b]{0.45\textwidth}
                \includegraphics[width=\textwidth]{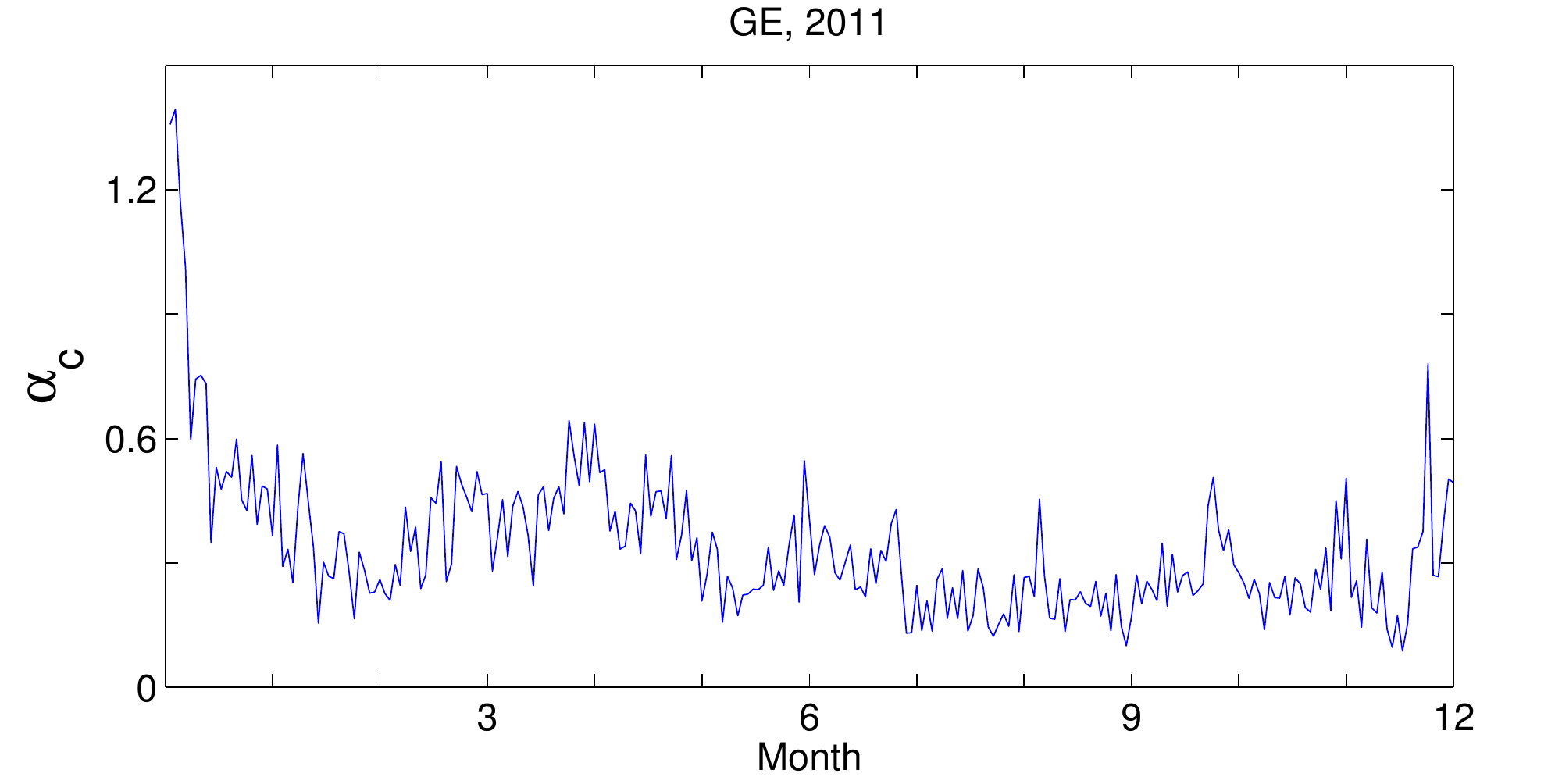}
                \caption{$\alpha_c$}
                \label{Fig:ge2011_alphac}
        \end{subfigure}
	    \centering
        \begin{subfigure}[b]{0.45\textwidth}
                \includegraphics[width=\textwidth]{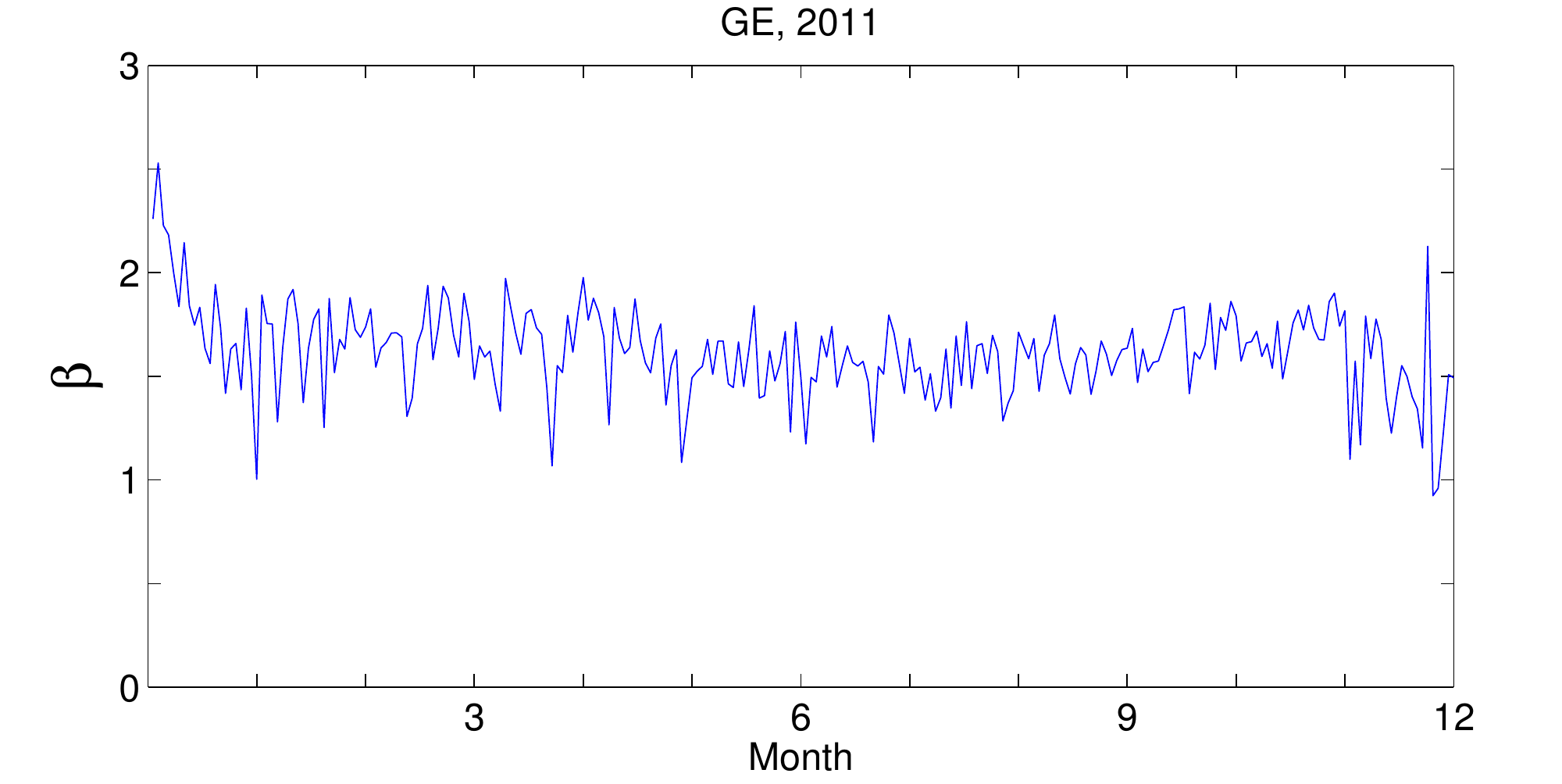}
                \caption{$\beta$}
                \label{Fig:ge2011_beta}
        \end{subfigure}
		\caption{Symmetric Hawkes estimation result, GE, 2011}\label{Fig:EstimationGE2011}
\end{figure}

Figure~\ref{Fig:EstimationGE2010} presents the parameter and volatility dynamics of GE, 2010.
Similar to the previous case, the behaviors $\mu$ and TSRV are similar.
The day of peaked volatilities in the figure is the May 6, 2010 Flash Crash where the equity prices fell rapidly.
At the day of the Flash Crash, the two estimated volatilities had different values and TSRV was much larger than the Hawkes volatility.

\begin{figure}
        \centering
        \begin{subfigure}[b]{0.45\textwidth}
                \includegraphics[width=\textwidth]{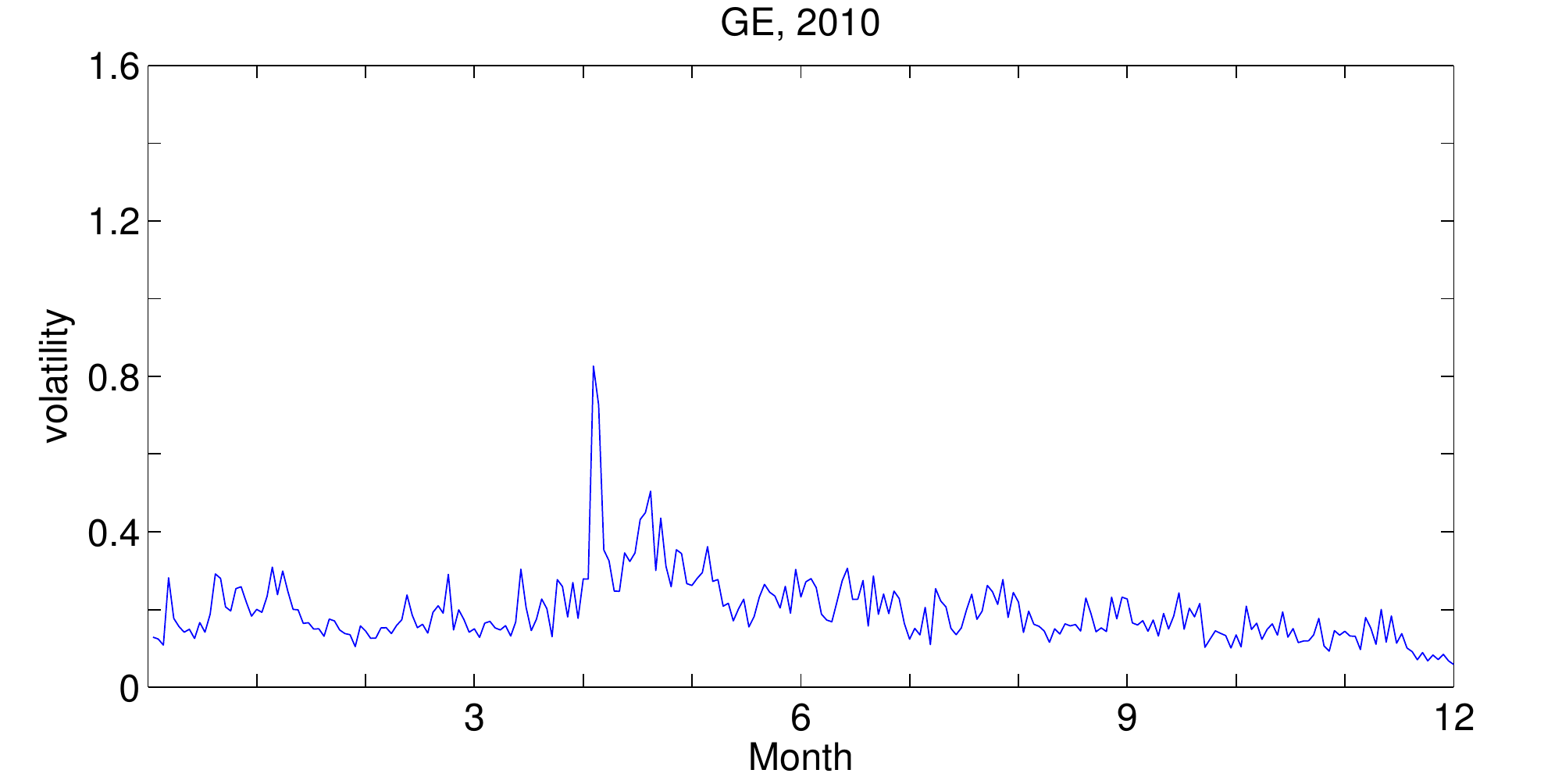}
                \caption{volatility}
                \label{Fig:ge2010_vol}
        \end{subfigure}
        \centering
        \begin{subfigure}[b]{0.45\textwidth}
                \includegraphics[width=\textwidth]{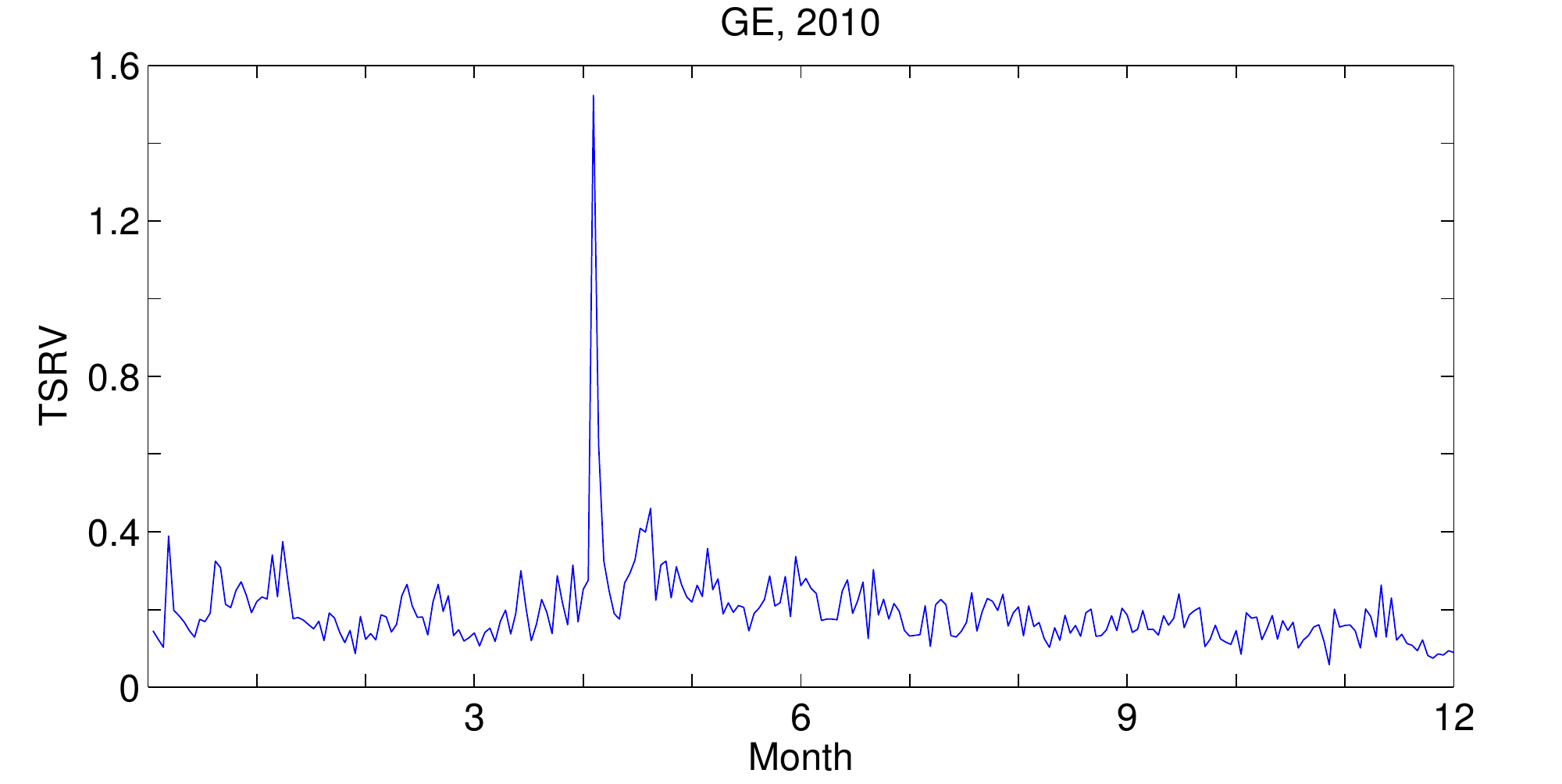}
                \caption{TSRV}
                \label{Fig:ge2010_TSRV}
        \end{subfigure}
	    \centering
        \begin{subfigure}[b]{0.45\textwidth}
                \includegraphics[width=\textwidth]{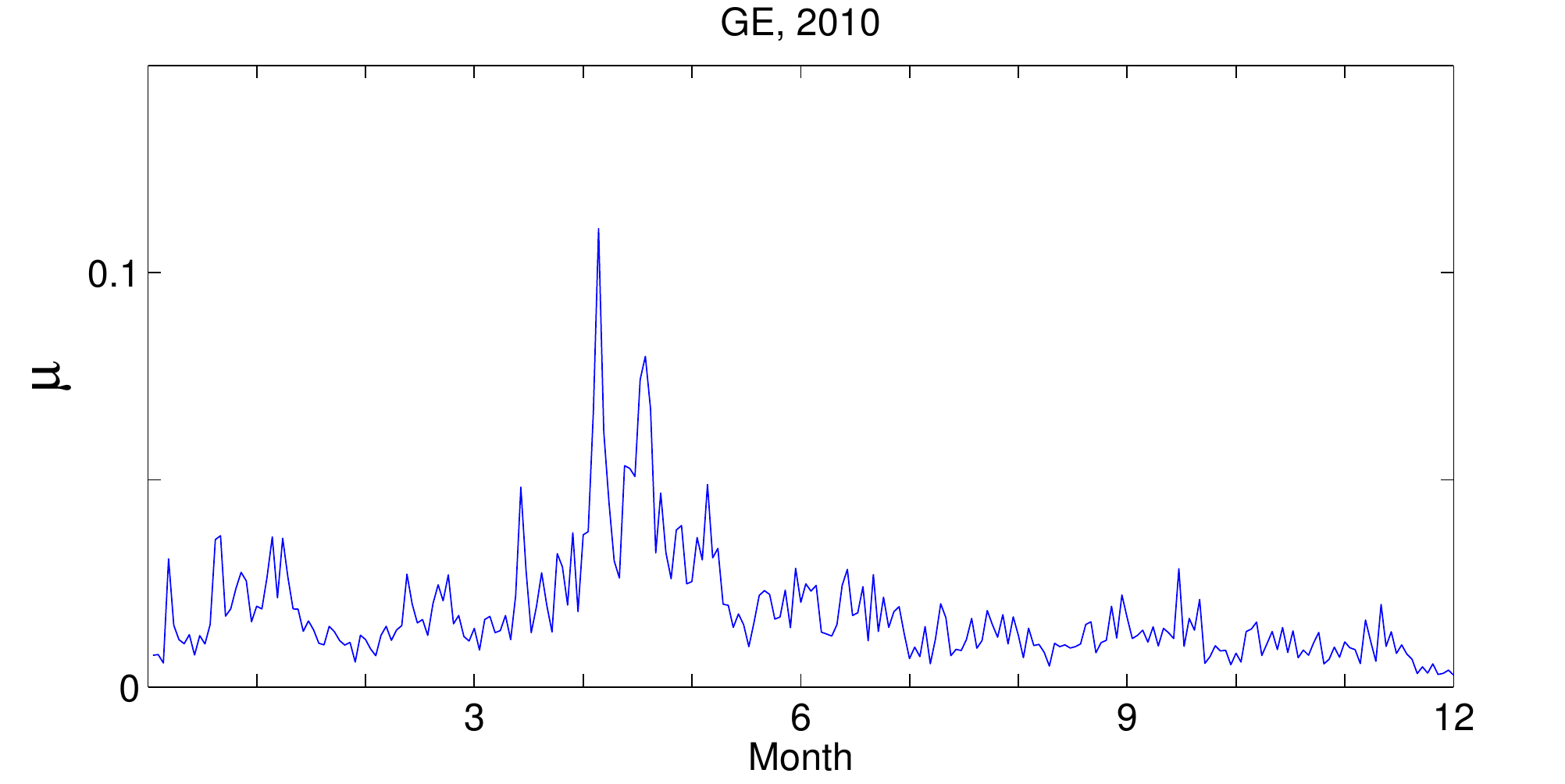}
                \caption{$\mu$}
                \label{Fig:ge2010_mu}
        \end{subfigure}
	    \centering
        \begin{subfigure}[b]{0.45\textwidth}
                \includegraphics[width=\textwidth]{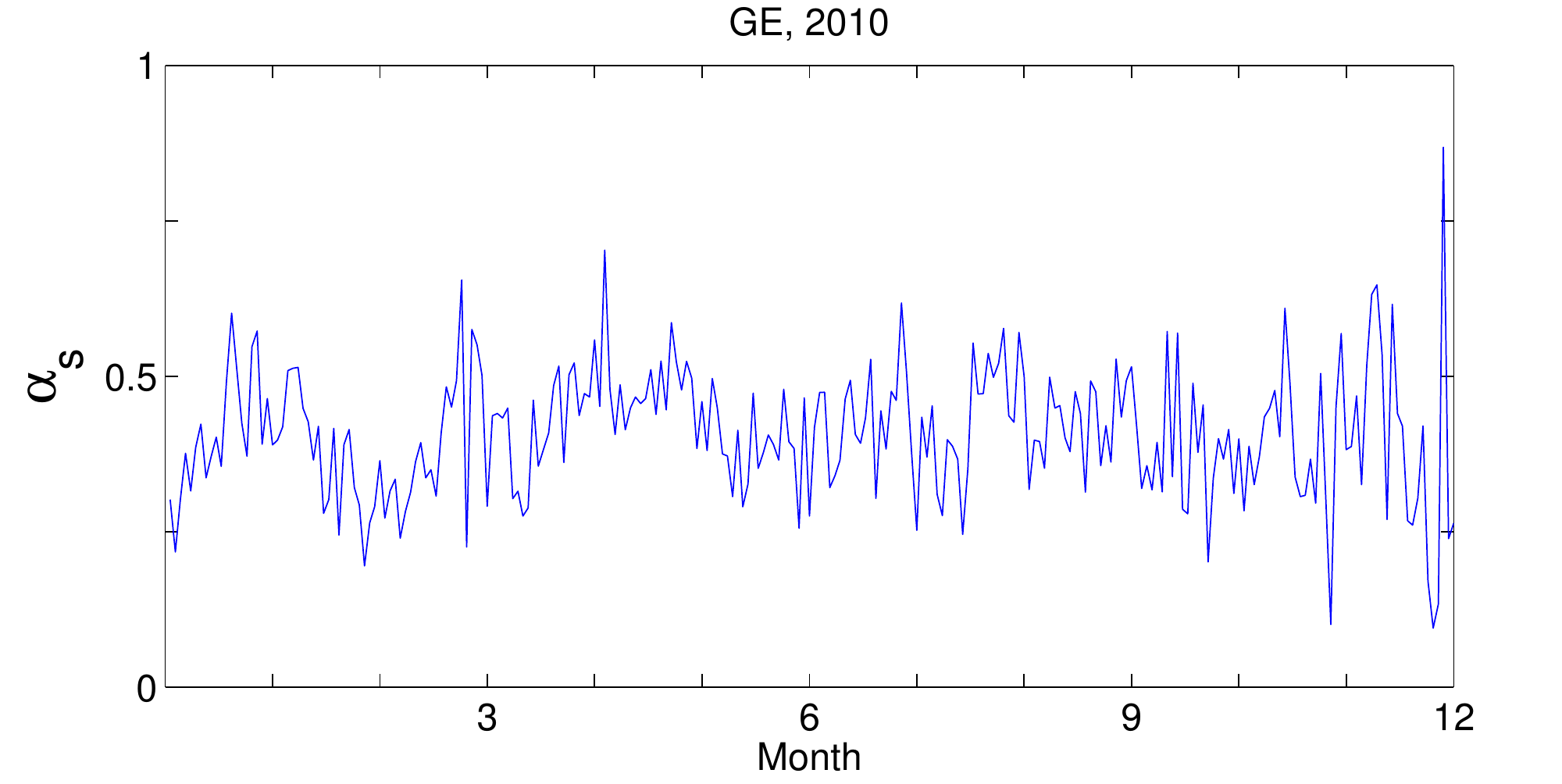}
                \caption{$\alpha_s$}
                \label{Fig:ge2010_alphas}
        \end{subfigure}
	    \centering
        \begin{subfigure}[b]{0.45\textwidth}
                \includegraphics[width=\textwidth]{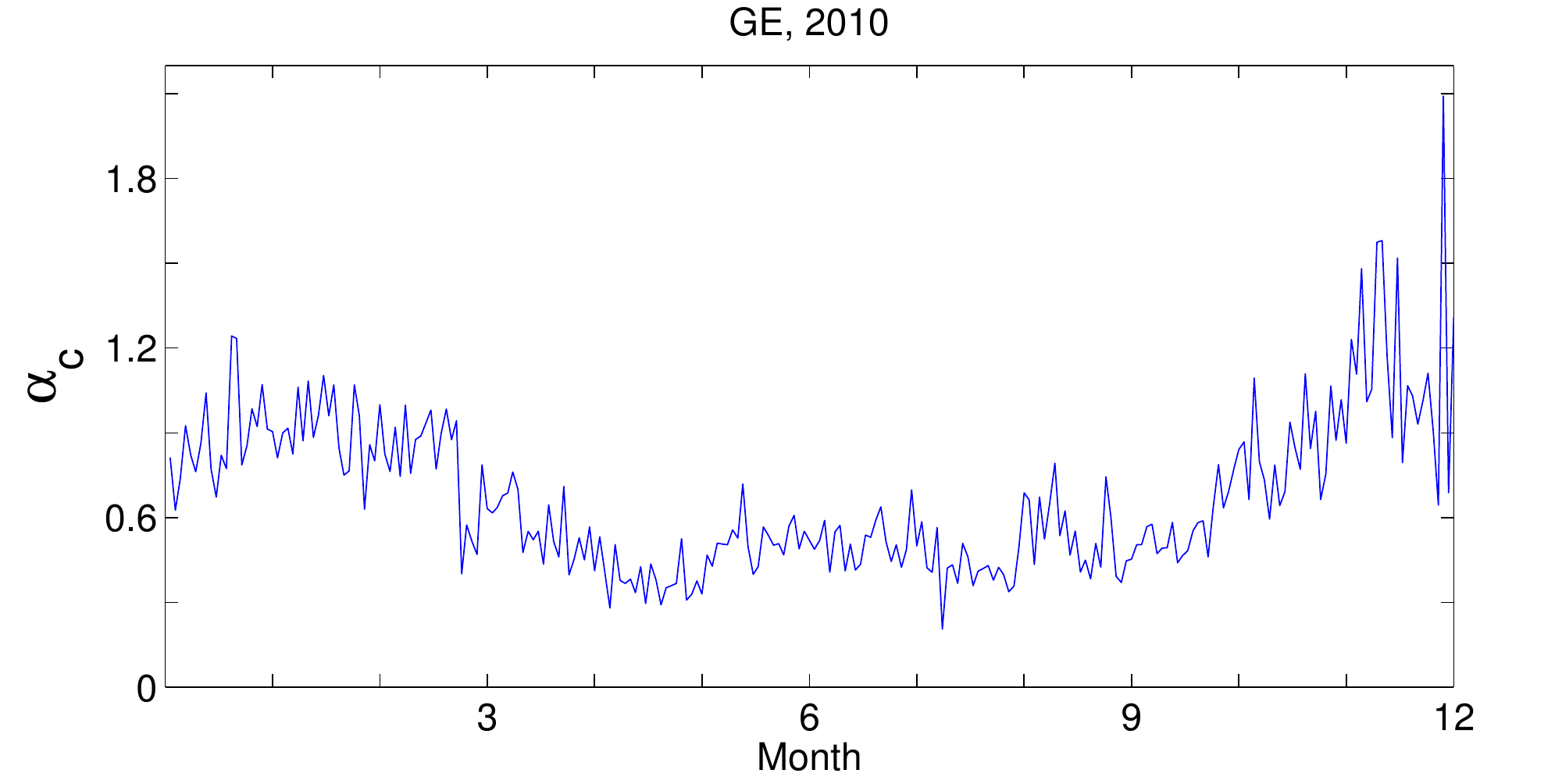}
                \caption{$\alpha_c$}
                \label{Fig:ge2010_alphac}
        \end{subfigure}
	    \centering
        \begin{subfigure}[b]{0.45\textwidth}
                \includegraphics[width=\textwidth]{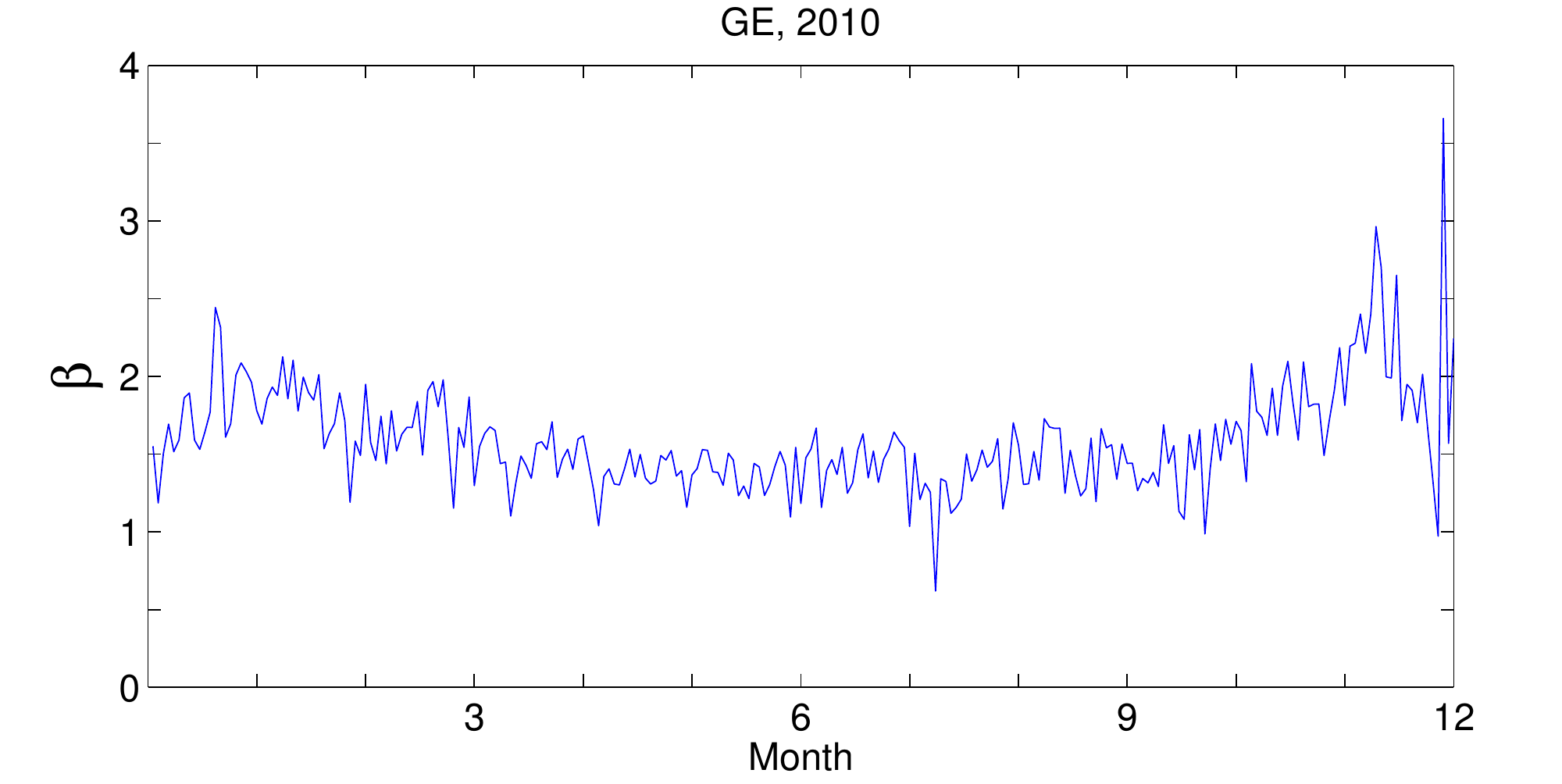}
                \caption{$\beta$}
                \label{Fig:ge2010_beta}
        \end{subfigure}
		\caption{Symmetric Hawkes estimation result, GE, 2010}\label{Fig:EstimationGE2010}
\end{figure}

In addition, in Figure~\ref{Fig:EstimationTMCD2010}, the estimated Hawkes volatility, TSRV and RRV of T (left) and MCD (right) are compared. 
All three volatilities have similar forms of movements during the observed period.
For T, the Hawkes volatility was close to TSRV (right) at the day of Flash Crash.
The estimated Hawkes volatility of MCD at the Flash Crash was larger than the TSRV or RRV.

\begin{figure}
		\centering
        \begin{subfigure}[b]{0.45\textwidth}
                \includegraphics[width=\textwidth]{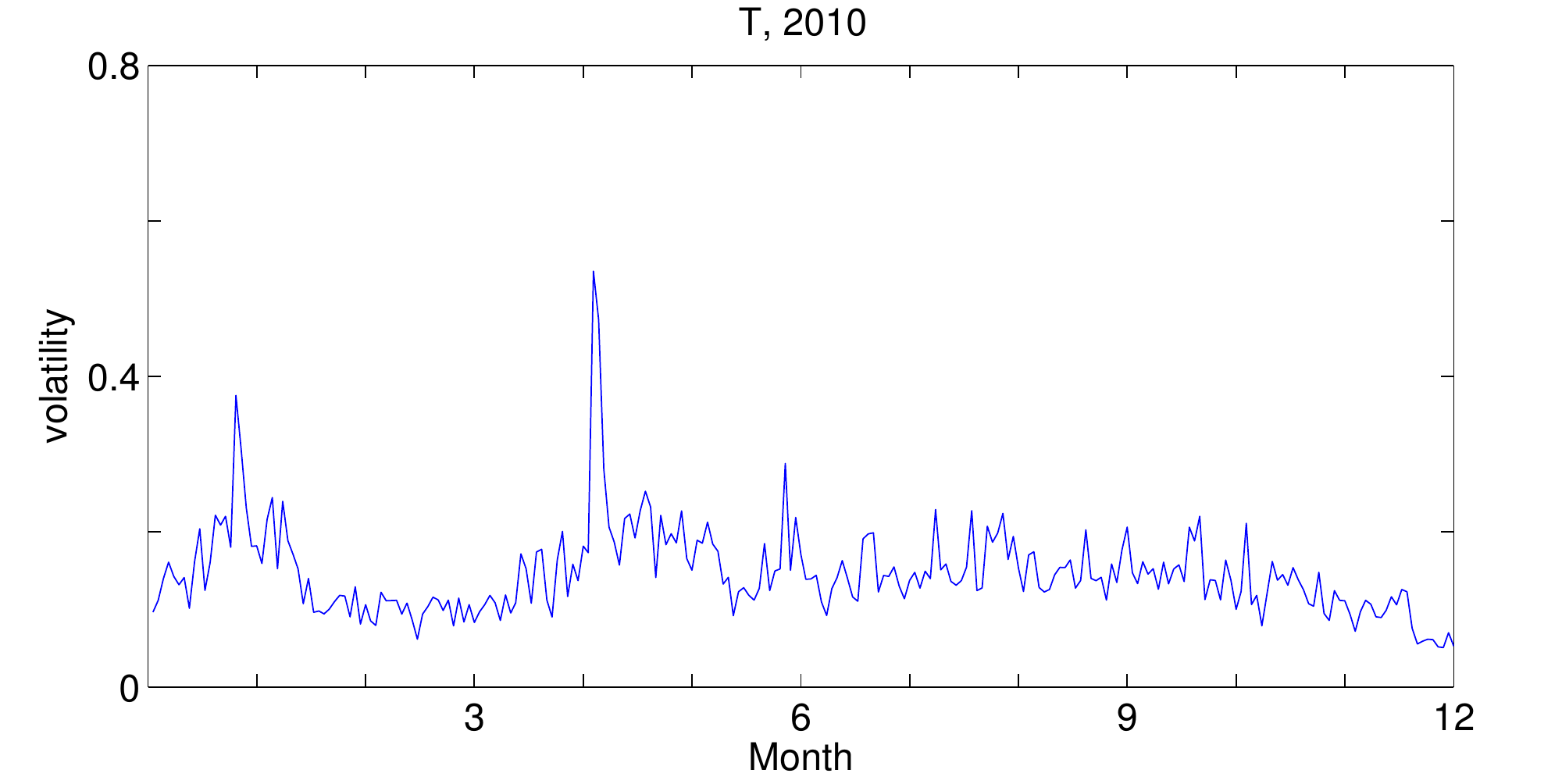}
                \caption{T, volatility}
        \end{subfigure}
        \centering
        \begin{subfigure}[b]{0.45\textwidth}
                \includegraphics[width=\textwidth]{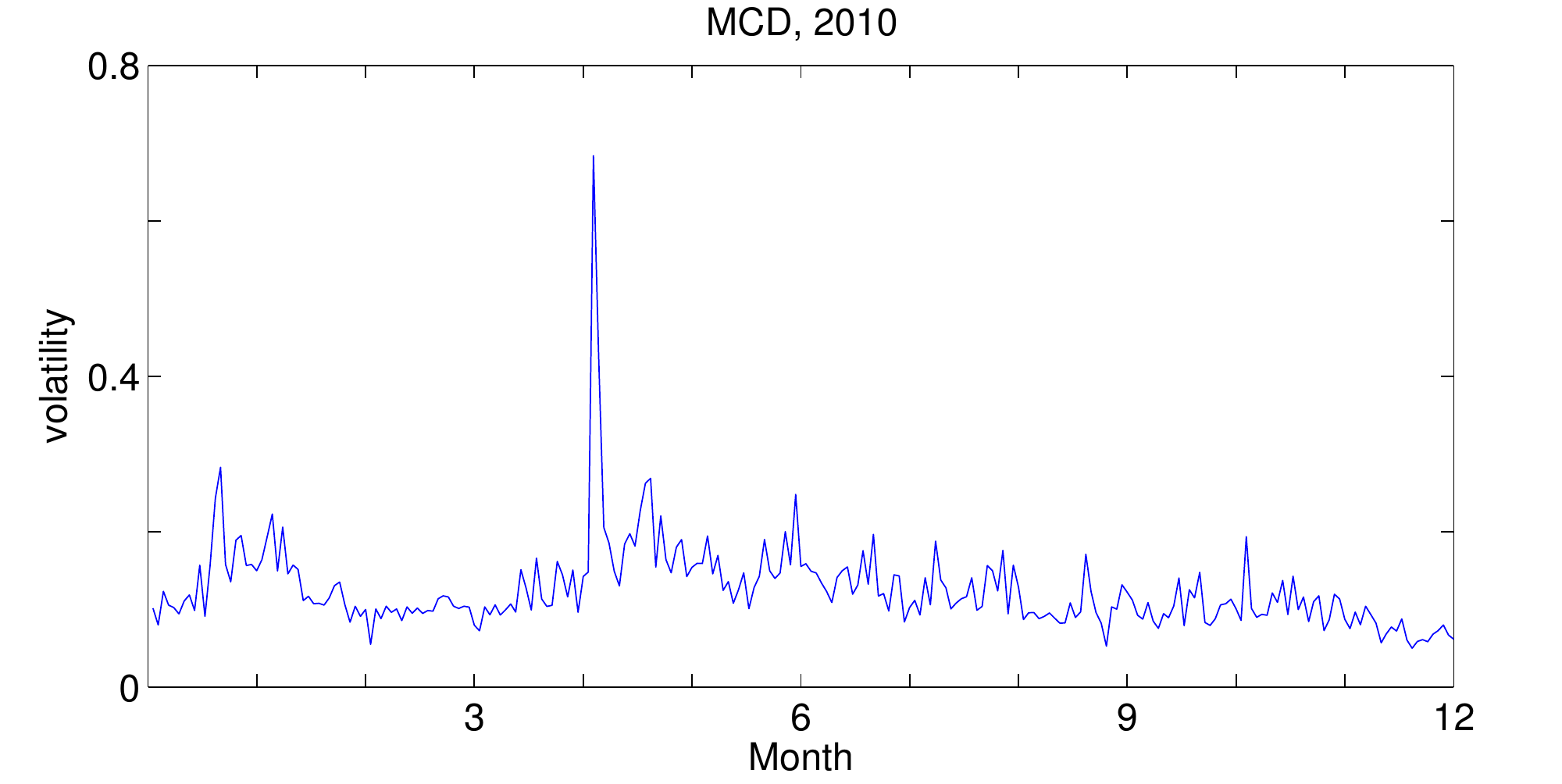}
                \caption{MCD, volatility}
        \end{subfigure}
        \centering
        \begin{subfigure}[b]{0.45\textwidth}
                \includegraphics[width=\textwidth]{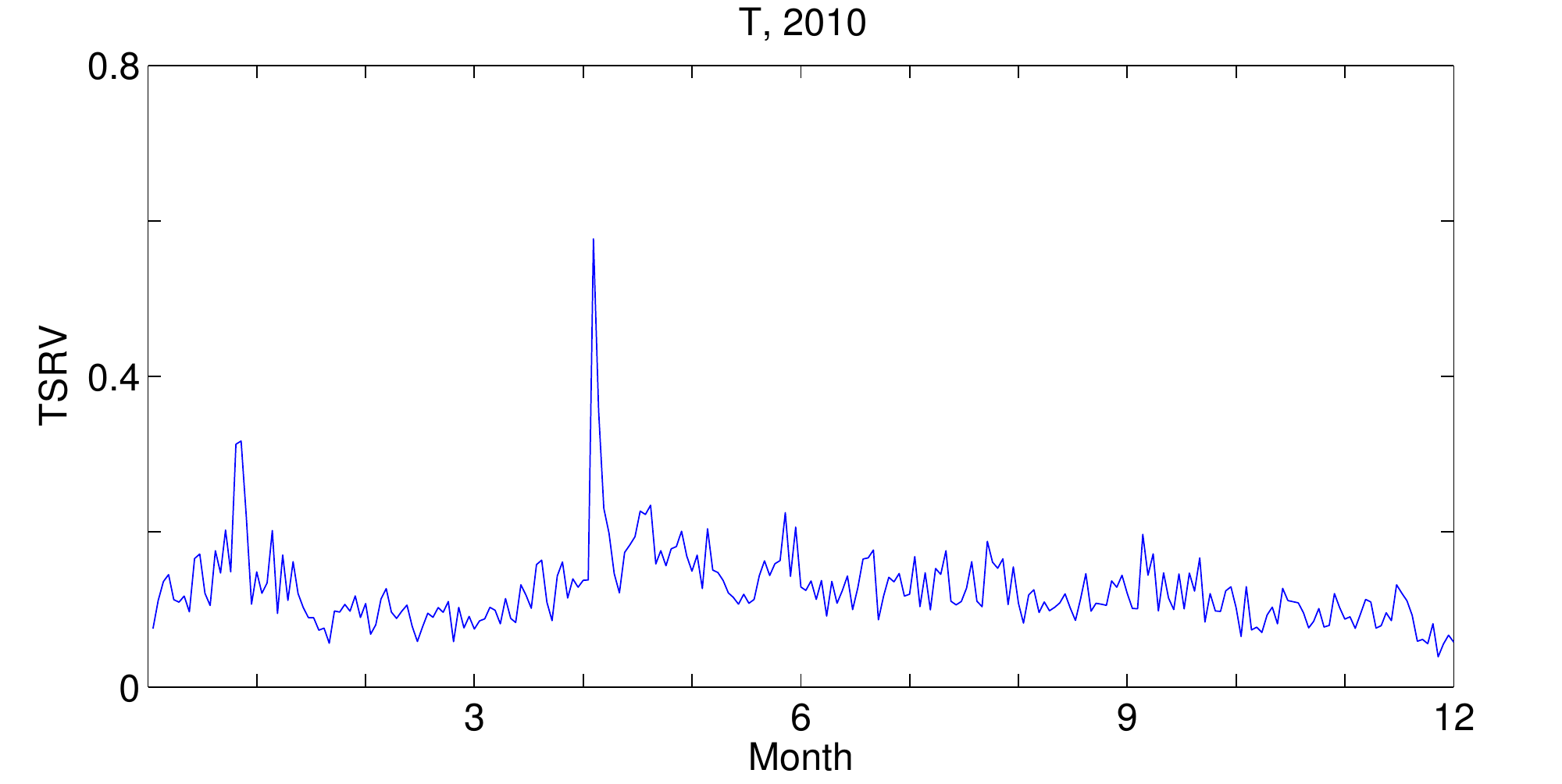}
                \caption{T, TSRV}
        \end{subfigure}
        \centering
        \begin{subfigure}[b]{0.45\textwidth}
                \includegraphics[width=\textwidth]{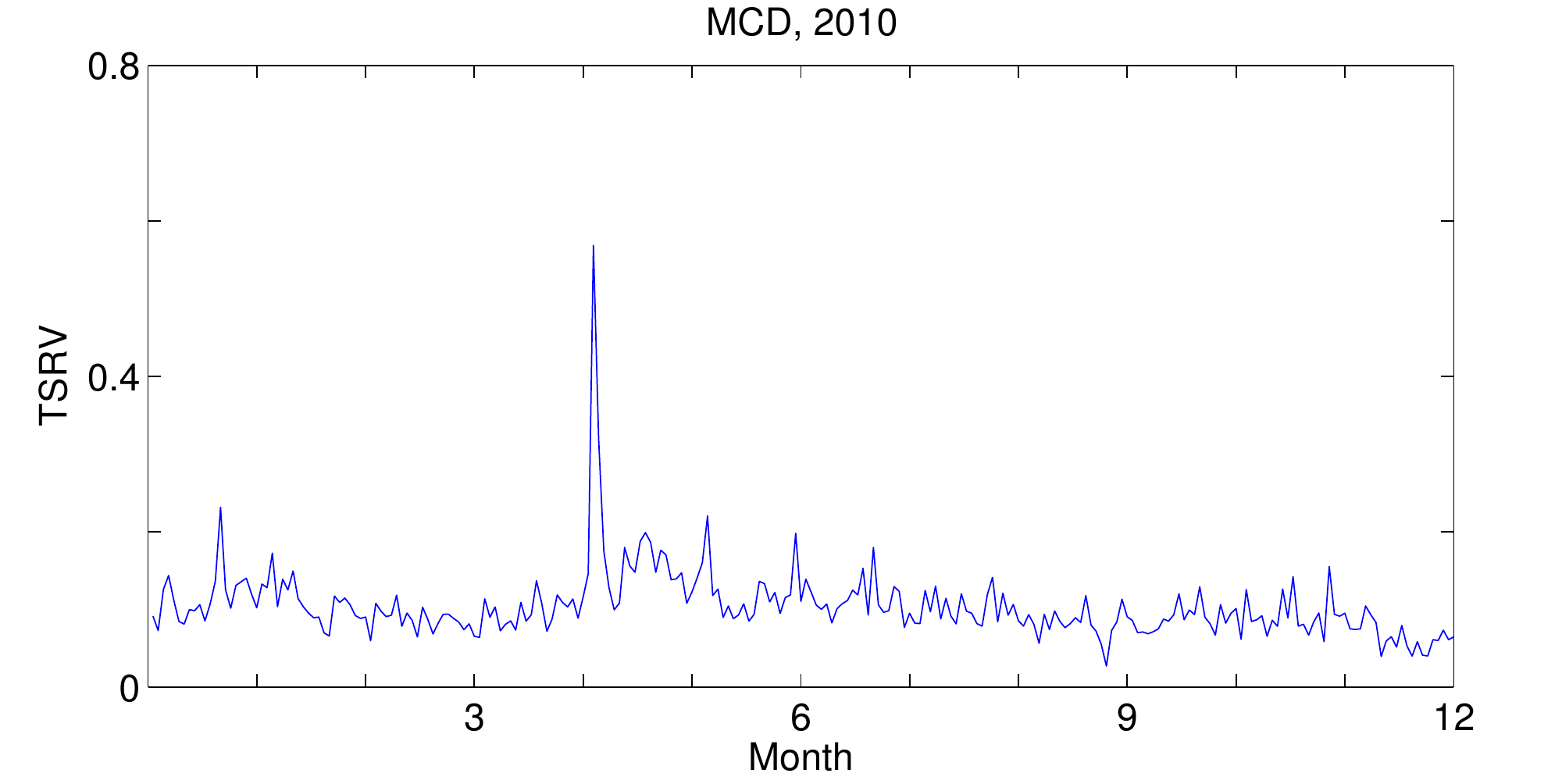}
                \caption{MCD, TSRV}
        \end{subfigure}
        \centering
        \begin{subfigure}[b]{0.45\textwidth}
                \includegraphics[width=\textwidth]{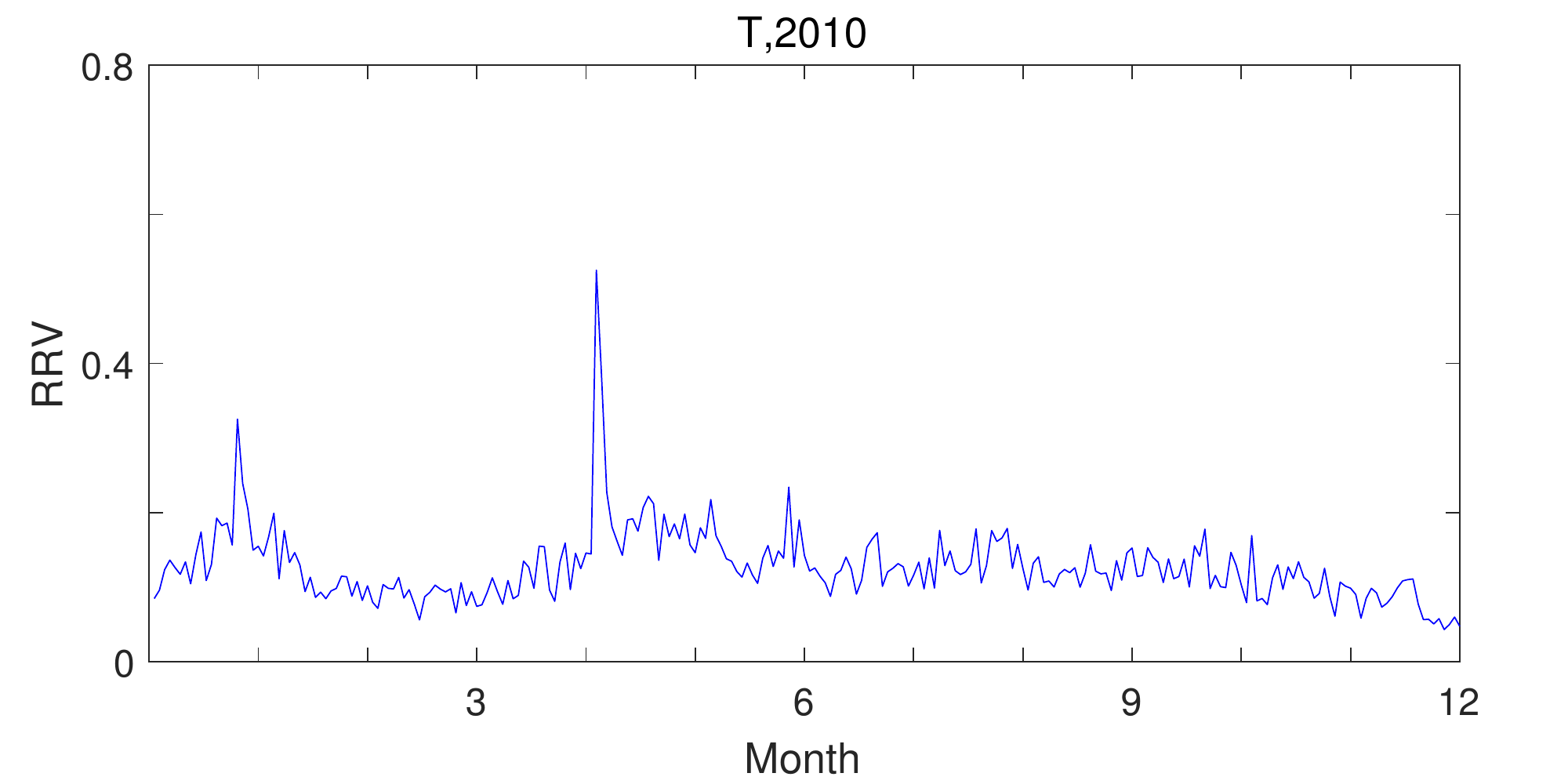}
                \caption{T, RRV}
        \end{subfigure}
        \centering
        \begin{subfigure}[b]{0.45\textwidth}
                \includegraphics[width=\textwidth]{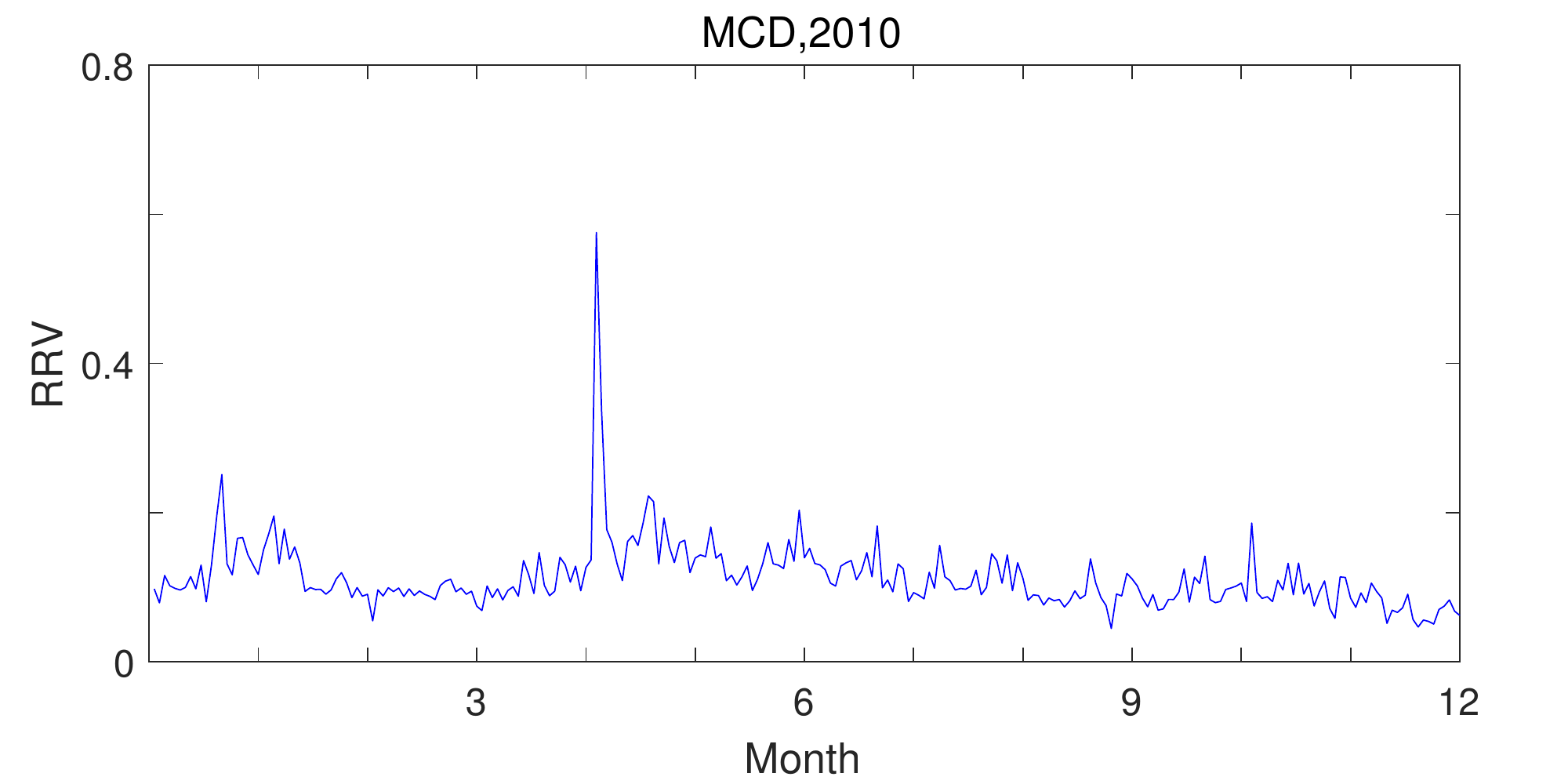}
                \caption{MCD, RRV}
        \end{subfigure}
		\caption{Volatility comparisons with symmetric Hawkes estimation results, T (left) and MCD (right), 2010}\label{Fig:EstimationTMCD2010}
\end{figure}

Figure~\ref{Fig:EstimationGE2008} plots the dynamics of the estimated parameters of GE in 2008, the starting year of the global financial crisis.
The dramatic changes in the Hawkes volatility, TSRV, $\mu, \alpha_s$ and $\beta$ were observed in the beginning of the crisis around August 2008.
The mutual excited parameter, $\alpha_c$, was rather stable.

\begin{figure}
        \centering
        \begin{subfigure}[b]{0.45\textwidth}
                \includegraphics[width=\textwidth]{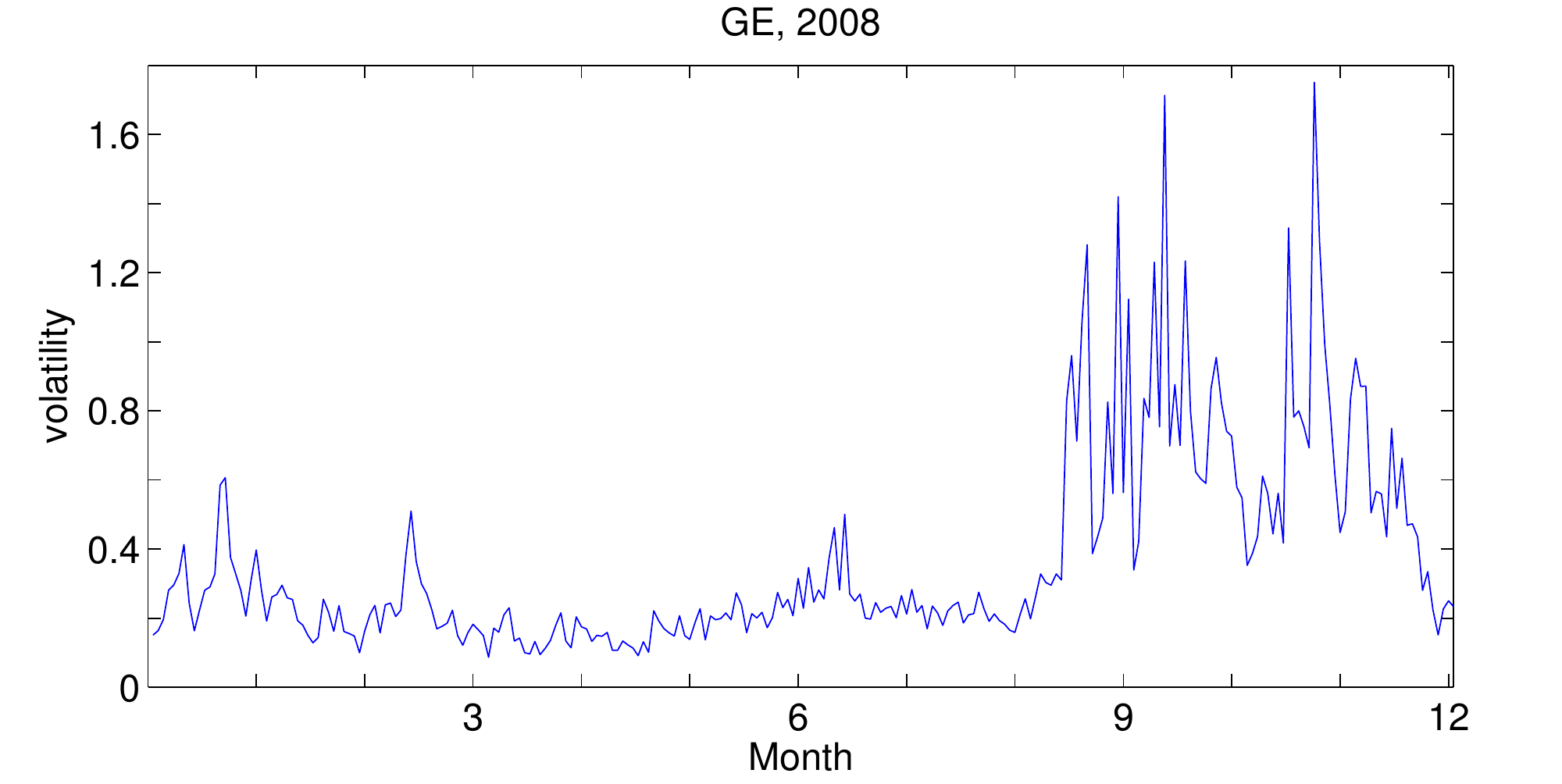}
                \caption{volatility}
                \label{Fig:ge2008_vol}
        \end{subfigure}
        \centering
        \begin{subfigure}[b]{0.45\textwidth}
                \includegraphics[width=\textwidth]{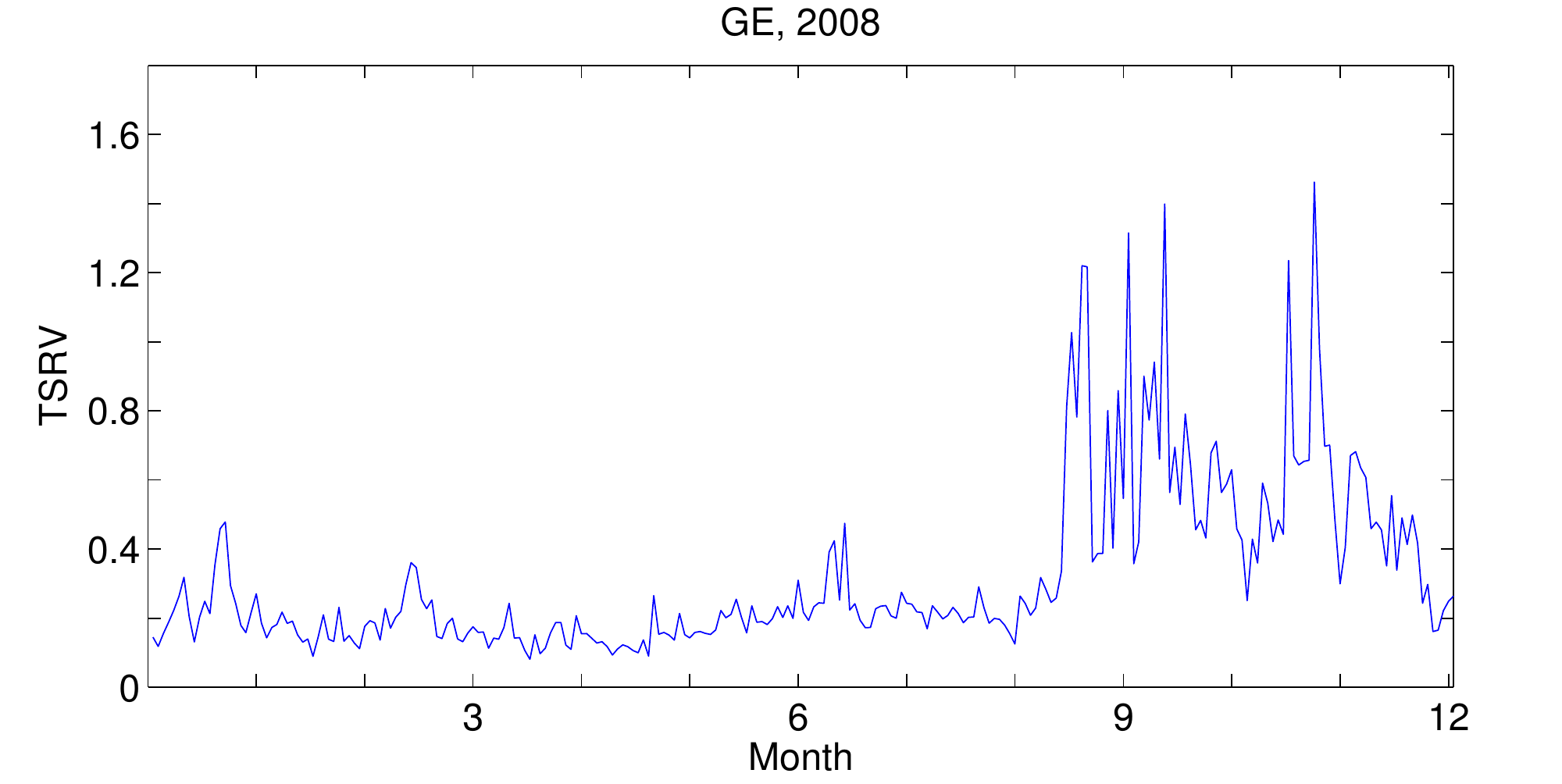}
                \caption{TSRV}
                \label{Fig:ge2008_TSRV}
        \end{subfigure}
	    \centering
        \begin{subfigure}[b]{0.45\textwidth}
                \includegraphics[width=\textwidth]{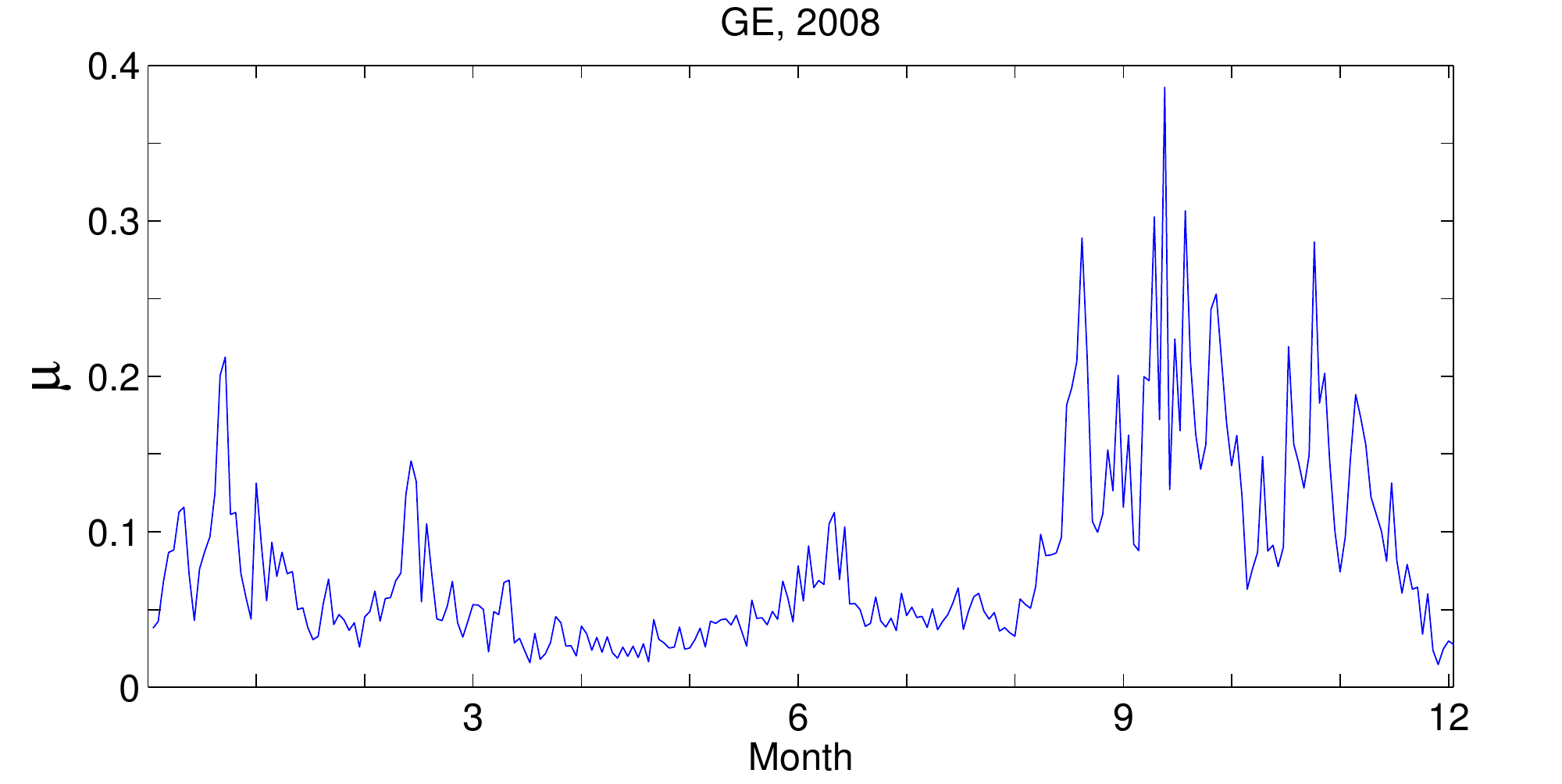}
                \caption{$\mu$}
                \label{Fig:ge2008_mu}
        \end{subfigure}
	    \centering
        \begin{subfigure}[b]{0.45\textwidth}
                \includegraphics[width=\textwidth]{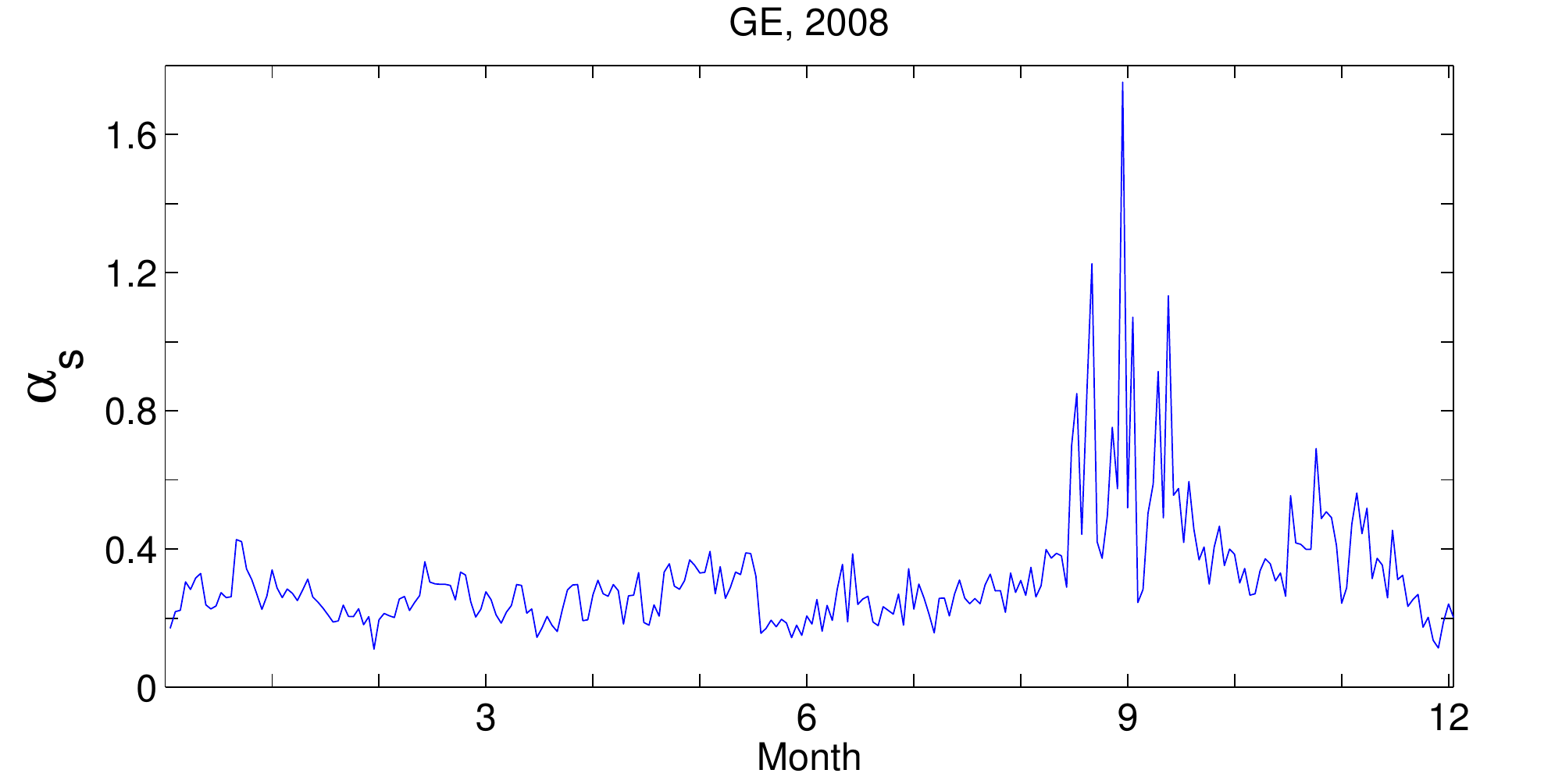}
                \caption{$\alpha_s$}
                \label{Fig:ge2008_alphas}
        \end{subfigure}
	    \centering
        \begin{subfigure}[b]{0.45\textwidth}
                \includegraphics[width=\textwidth]{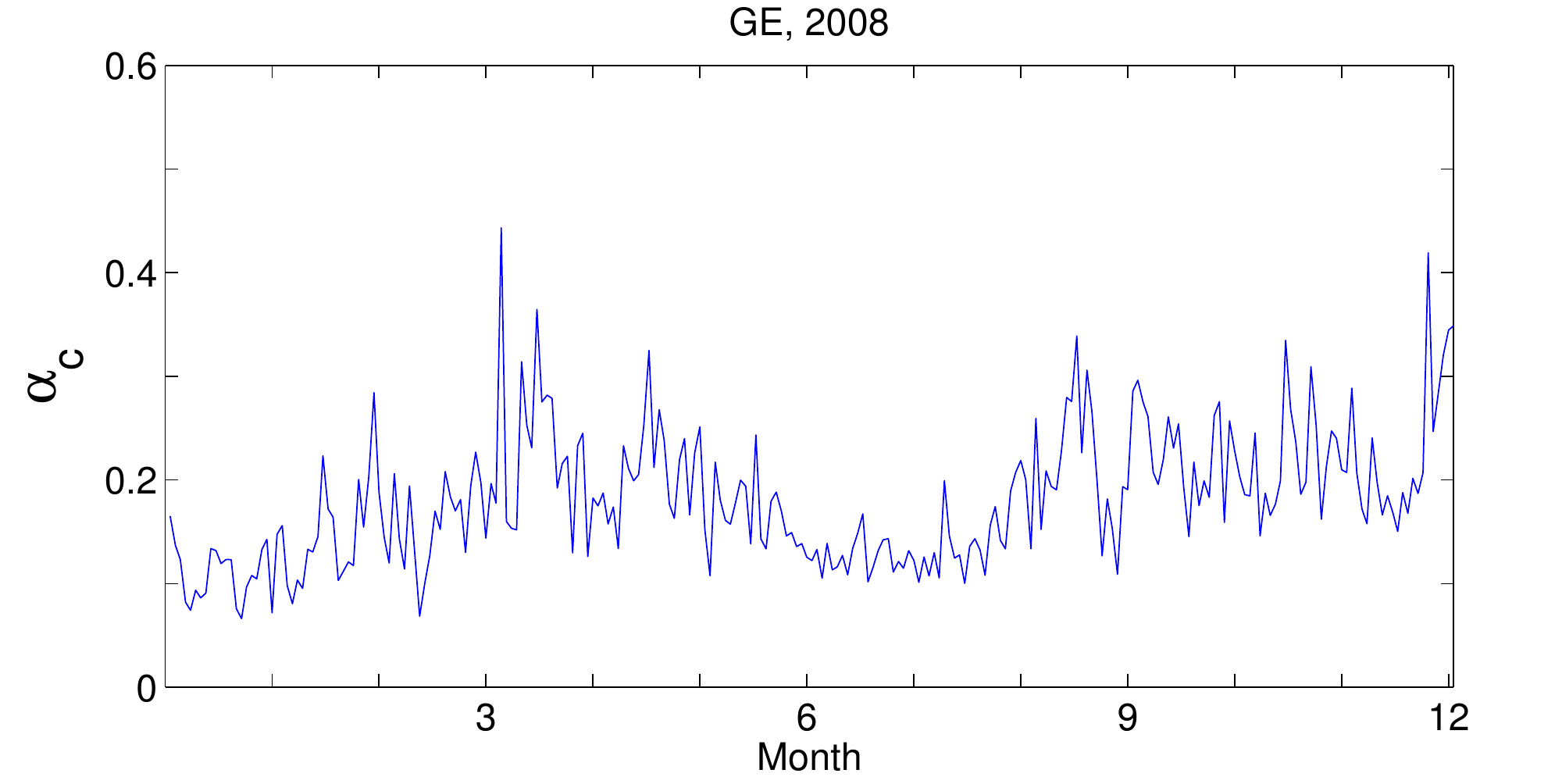}
                \caption{$\alpha_c$}
                \label{Fig:ge2008_alphac}
        \end{subfigure}
	    \centering
        \begin{subfigure}[b]{0.45\textwidth}
                \includegraphics[width=\textwidth]{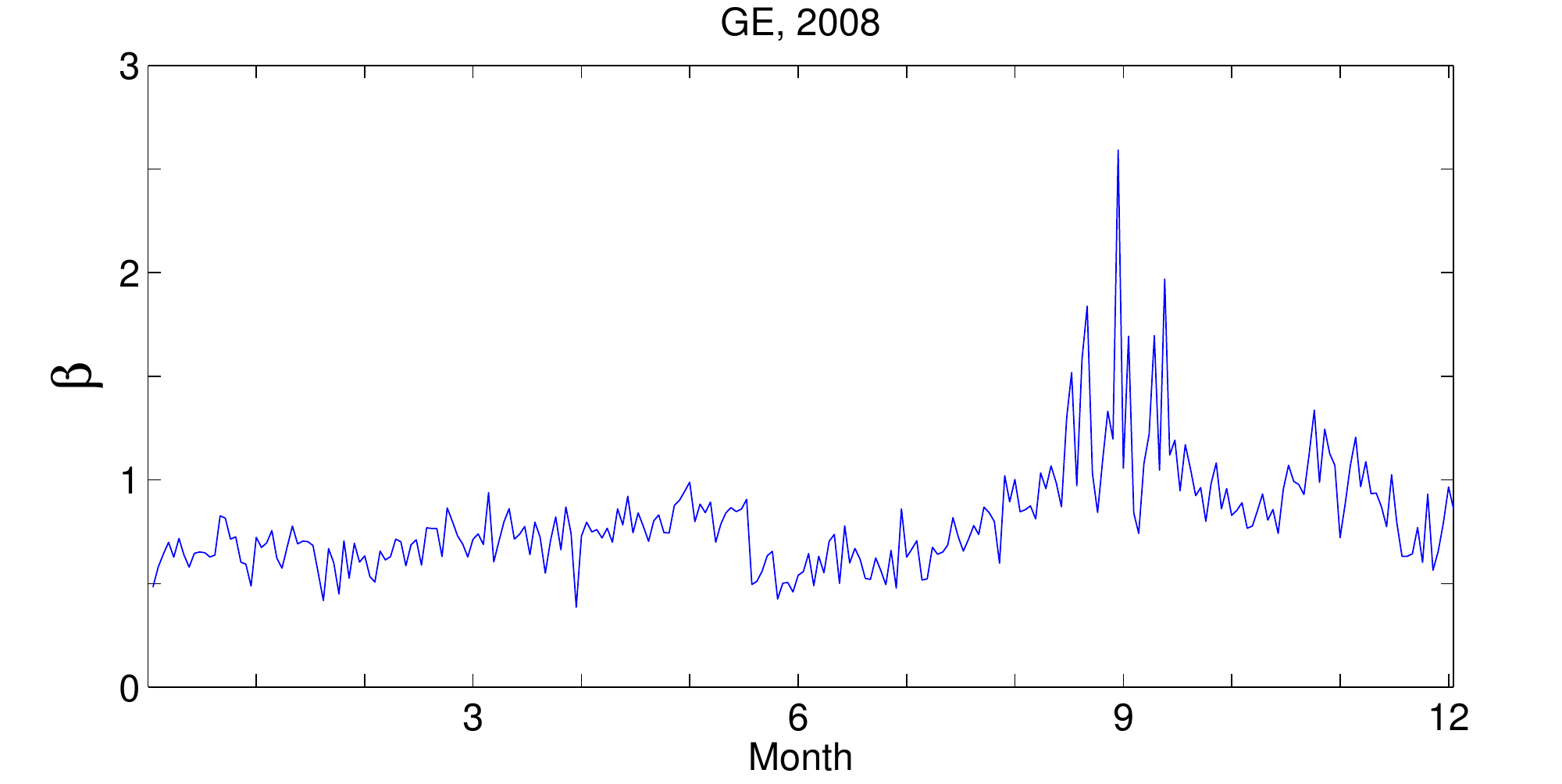}
                \caption{$\beta$}
                \label{Fig:ge2008_beta}
        \end{subfigure}
		\caption{Symmetric Hawkes estimation result, GE, 2008}\label{Fig:EstimationGE2008}
\end{figure}

The volatility estimation results were compared using the Hawkes model and TSRV method in Table~\ref{Table:volatility} with 10 symbols from 2007 to 2011.
In each panel of the table, the mean of the Hawkes volatility and TSRV for given year and mean percentage error are presented.
The volatility estimated by the symmetric Hawkes model is generally larger than the TSRV and the differences between the two volatilities are around 15-25\%.
The reason for the discrepancy between the TSRV and the Hawkes volatility is unknown.
Possible reasons include the intraday variation of the parameters as in the Flash Crash and the restrictions in the parameter condition for the symmetry.
These two issues are examined in the following subsections.

\begin{table}
\caption{Comparison of the volatility estimation by the Hawkes model and realized volatility}\label{Table:volatility}
\centering
\begin{tabular}{ccccccccccc}
\hline
 & BAC & CVX & GE & IBM & JPM & KO & MCD & T & VZ & XOM\\
\hline
2007 \\
H.vol & 0.1725 & 0.2387 & 0.1406 & 0.1674 & 0.2170 & 0.1290 & 0.1494 & 0.1731 & 0.1615 & 0.2327\\
TSRV  & 0.1555 & 0.1834 & 0.1323 & 0.1380 & 0.1899 & 0.1118 & 0.1296 & 0.1615 & 0.1433 & 0.1730\\
MPE(\%) & 14.93 & 20.28 & 12.34 & 17.26 & 15.91 & 15.58 & 15.87 & 14.04 & 14.26 & 23.14\\
\hline
2008 \\
H.vol & 0.7008 & 0.5129 & 0.3646 & 0.3939 & 0.6929 & 0.2583 & 0.3286 & 0.3751 & 0.3806 & 0.4592\\
TSRV  & 0.4880 & 0.3077 & 0.3172 & 0.2658 & 0.4868 & 0.2031 & 0.2440 & 0.3069 & 0.2930 & 0.2827\\
MPE(\%) & 28.14 & 33.17 & 13.57 & 28.41 & 26.61 & 19.60 & 21.34 & 16.09 & 19.97 & 32.53\\
\hline
2009 \\
H.vol & 0.7397 & 0.2608 & 0.3439 & 0.2166 & 0.4701 & 0.1762 & 0.1883 & 0.2407 & 0.2029 & 0.2198\\
TSRV  & 0.5420 & 0.2029 & 0.3367 & 0.1664 & 0.3722 & 0.1509 & 0.1607 & 0.1934 & 0.1810 & 0.1773\\
MPE(\%) & 25.42 & 21.48 & 9.12 & 20.81 & 17.97 & 15.23 & 15.79 & 20.27 & 13.43 & 18.26\\
\hline
2010 \\
H.vol & 0.2869 & 0.1758 & 0.1963 & 0.1395 & 0.2223 & 0.1138 & 0.1258 & 0.1461 & 0.1578 & 0.1603\\
TSRV  & 0.2234 & 0.1376 & 0.1952 & 0.1184 & 0.1985 & 0.1031  & 0.1040 & 0.1255 & 0.1241 & 0.1291\\
MPE(\%) & 21.20 & 21.11 & 11.47 & 17.07 & 13.46 & 13.51 & 18.12 & 15.96 & 17.73 & 18.86\\
\hline
2011 \\
H.vol & 0.3426 & 0.2197 & 0.2389 & 0.1681 & 0.2554 & 0.1334 & 0.1297 & 0.1460 & 0.1543 & 0.1872\\
TSRV  & 0.2648 & 0.1719 & 0.1921 & 0.1334 & 0.2146 & 0.1091 & 0.1122 & 0.1227 & 0.1228 & 0.1553\\
MPE(\%) & 13.79 & 20.46 & 18.78 & 18.49 & 17.06 & 19.24 & 15.73 & 17.47 & 20.38 & 17.38\\
\hline
\end{tabular}
\end{table}

\subsection{Intraday volatility}

One of the interesting applications to modeling the daily price dynamics using the symmetric Hawkes process is that the intraday volatility can be estimated in almost every moment of the day.
This is possible because every arrival time of price change, which are plentiful even during ten minutes, is used and the maximum likelihood estimation is so powerful that the parameters can be estimated with similar or less than ten minutes data.
Figure~\ref{Fig:realvol} shows the dynamics of the intraday volatility of GE with randomly chosen days.
The first estimation of each day was performed using the first ten minutes data of each day.
In this example, it ranged from 10:00 a.m. to 10:10 a.m.

The price movement histories were then updated in every ten minutes and 
the intraday volatilities were re-estimated using the updated data and already existing one.
For the estimation, the reparametrization in Remark~\ref{Rmk:reparm} were used and hence 
the annualized volatility was estimated directly by the maximum likelihood estimation with its numerically computed standard error.
The solid lines in the figure represent the annualized volatilities estimated by the intraday data up to the time and the dotted lines represent the standard errors.
In Figures~\ref{Fig:realvol_ge20110805}~and~\ref{Fig:realvol_ge20111013}, the volatilities are generally large at the beginning of the day and tend to decreases, which is consistent with the seasonality effect in that in the early markets, more trading activities are observed than the middle of the day.

Figure~\ref{Fig:realvol_ge20100506} shows the data for the 2010/05/06 Flash Crash and a dramatic increase was observed in the late part of the day.
Similar behavior is presented in Figure~\ref{Fig:realvol_vz20100506}, which is for the intraday volatility of VZ in the 2010/05/06 Flash Crash.
The real time volatility measurement technique will be very useful in intraday risk managements, because investors can respond to sudden market changes more effectively, if they can compute the exact volatility variation. 

\begin{figure}
        \centering
        \begin{subfigure}[b]{0.45\textwidth}
                \includegraphics[width=\textwidth]{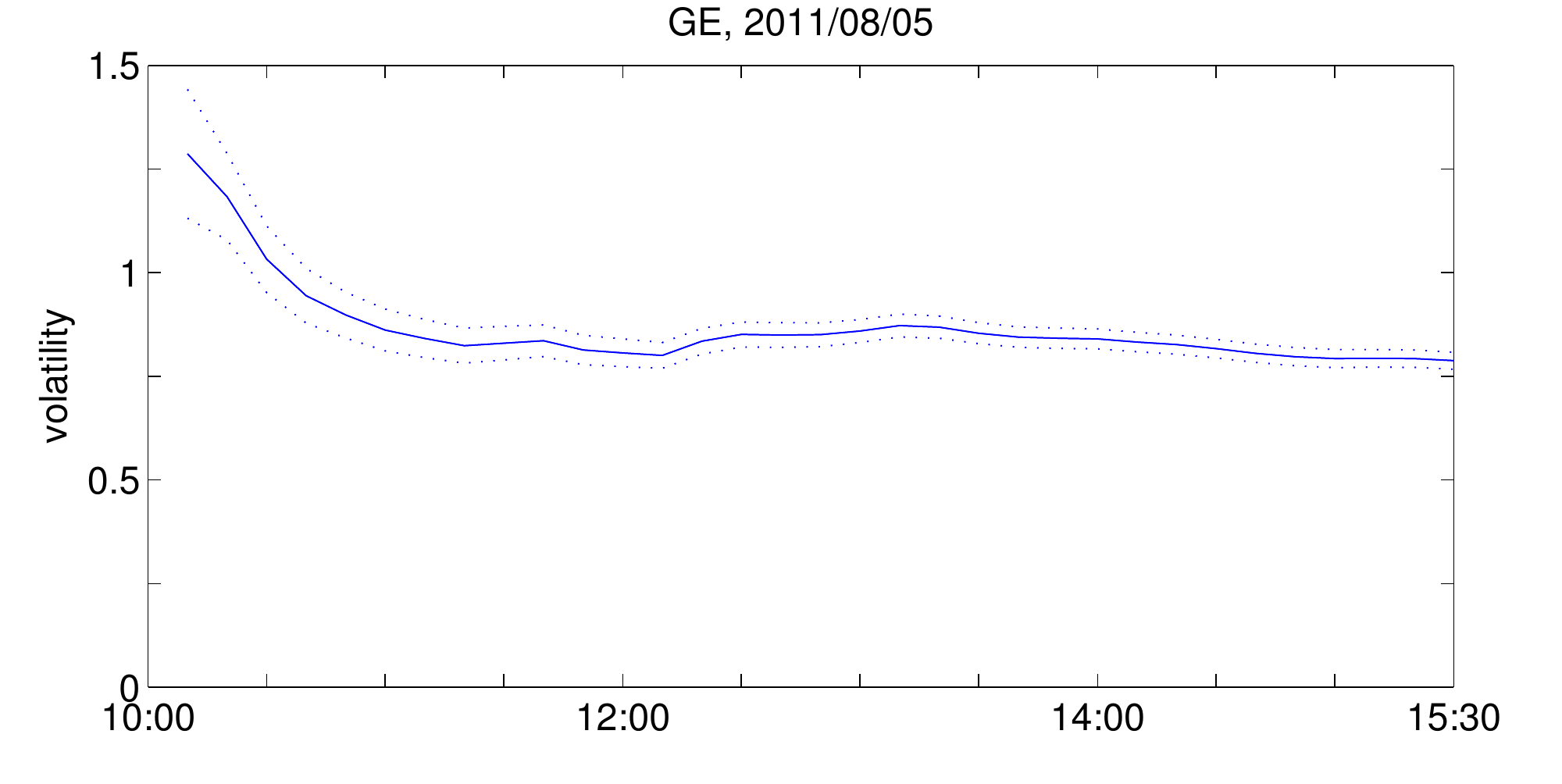}
                \caption{}
                \label{Fig:realvol_ge20110805}
        \end{subfigure}
        \centering
        \begin{subfigure}[b]{0.45\textwidth}
                \includegraphics[width=\textwidth]{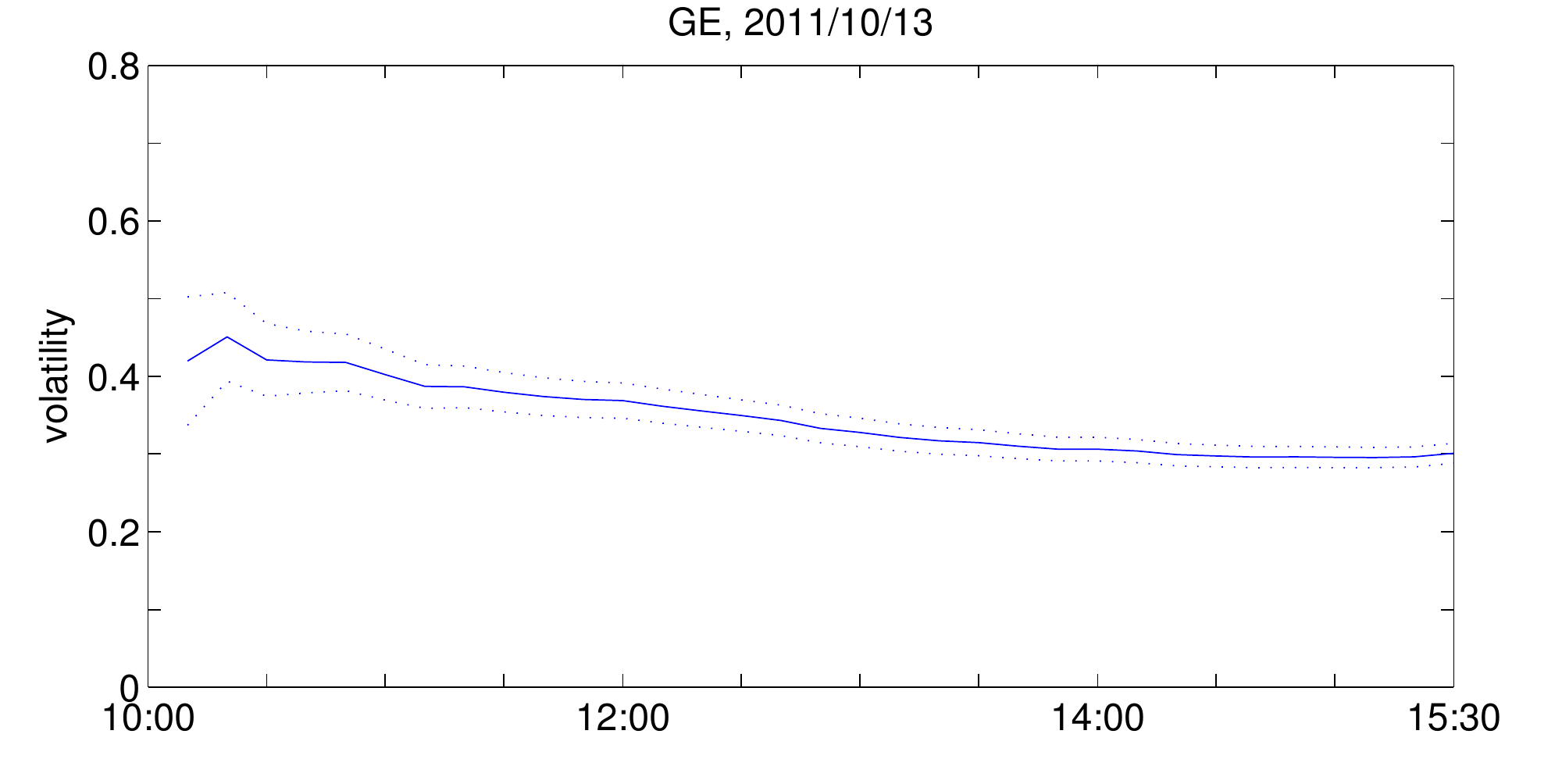}
                \caption{}
                \label{Fig:realvol_ge20111013}
        \end{subfigure}
        \centering
        \begin{subfigure}[b]{0.45\textwidth}
                \includegraphics[width=\textwidth]{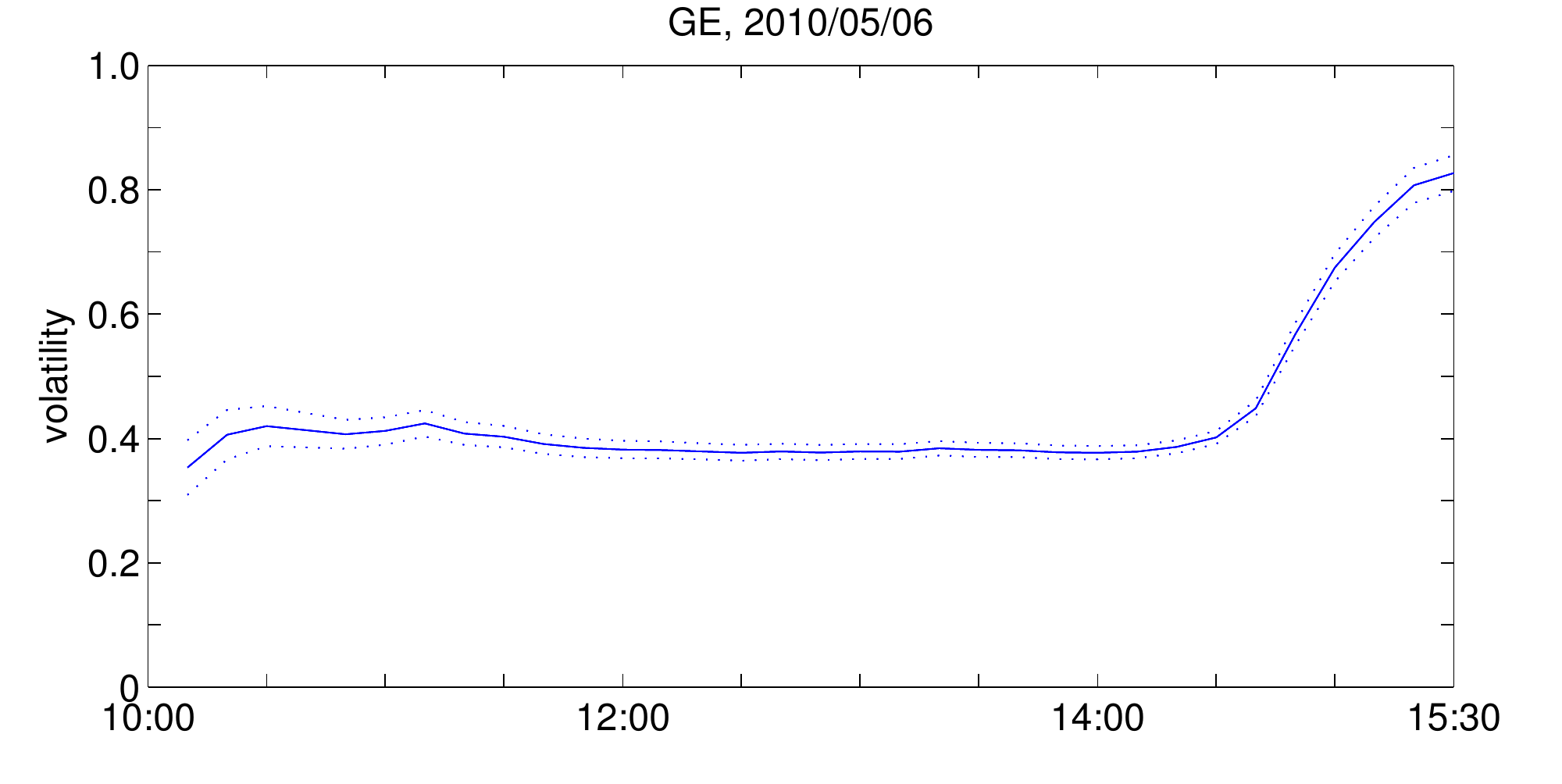}
                \caption{}
                \label{Fig:realvol_ge20100506}
        \end{subfigure}
		\begin{subfigure}[b]{0.45\textwidth}
                \includegraphics[width=\textwidth]{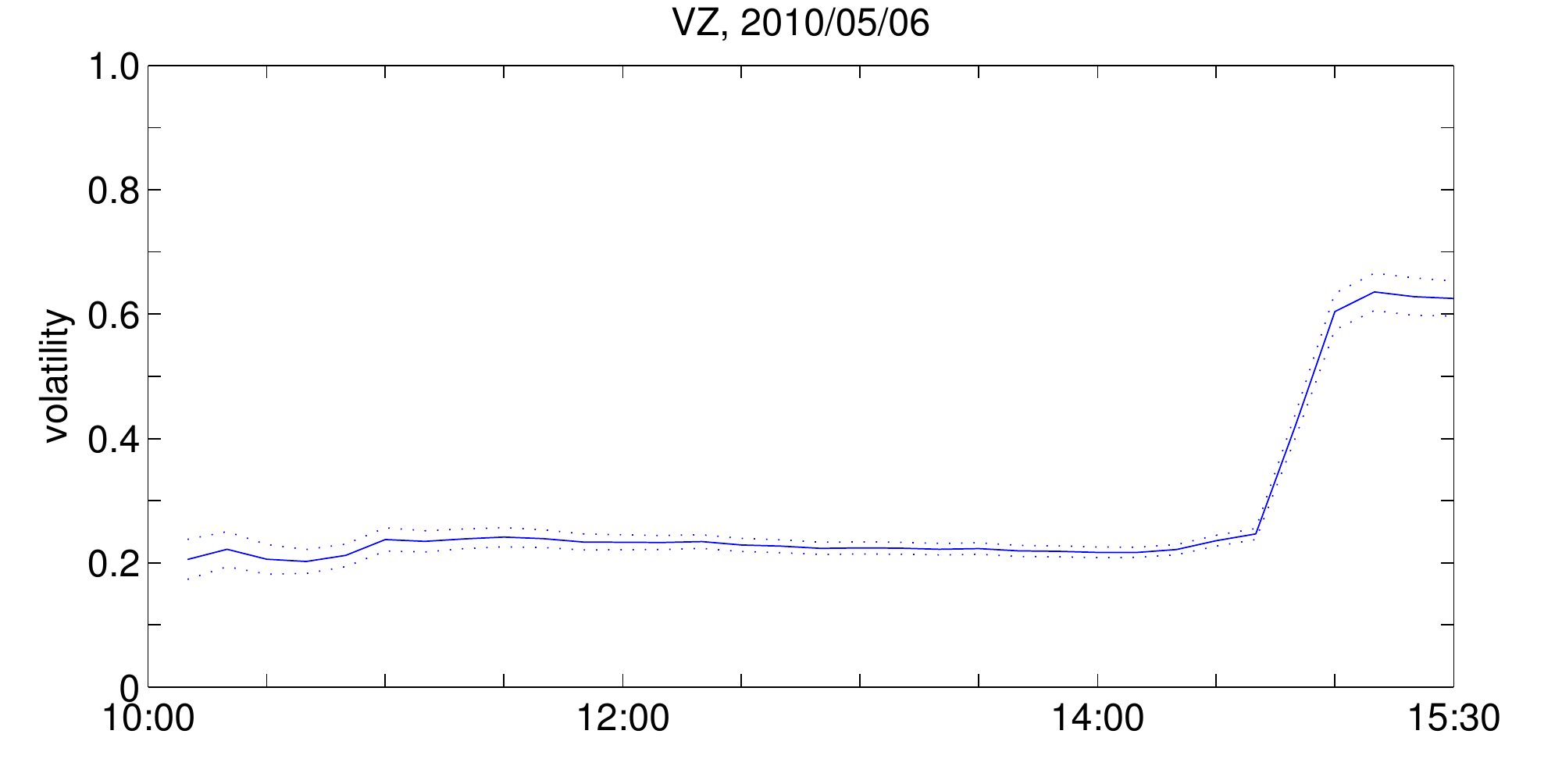}
                \caption{}
                \label{Fig:realvol_vz20100506}
        \end{subfigure}
        \caption{Estimated cumulative intraday volatility (annualized) with every ten minutes update}\label{Fig:realvol}
\end{figure}

\subsection{Fully characterized Hawkes process}
The maximum likelihood estimation of the Hawkes process with full characterization of the parameters as explained in Subsection~\ref{Subsec:full} was performed.
As in the previous subsection, the estimations were employed on a daily basis.
In Figure~\ref{Fig:EstimationGE2011full}, the dynamics of all parameters of the fully characterized Hawkes model of GE in 2011 are plotted.

The dynamics of the parameters $\mu_1, \mu_2$ are close to each other in the mean, as illustrated in Figure~\ref{Fig:ge2011_mu12} which suggest that $\mu_1 = \mu_2$ in the long run sense.
Similarly, each pair of parameters of $(\alpha_{11}, \alpha_{22})$, $(\alpha_{12}, \alpha_{21})$, $(\beta_{11}, \beta_{22})$, and $(\beta_{12}, \beta_{21})$ are close to each other in the mean.
The dynamics of the parameters $\beta_{12}$ and $\beta_{21}$ fluctuate more than $\beta_{11}$ and $\beta_{22}$ over time.
The sample means of $\beta_{12}$ and $\beta_{21}$ of GE in 2011 are quite close to $\beta_{11}$ and $\beta_{22}$ as reported in panel A of Table~\ref{Table:GE2011_full}.
The row `std.' in the table is the sample standard deviation of the time series of each parameter over the time period.

On the other hand, the parameters $\beta_{ij}$ are not always close to among others in the mean.
In the panel B of the table which presents the estimates of the parameters of XOM in 2008, 
the estimates of $\beta_{11}$ and $\beta_{22}$ are close to each other and similarly, the estimates of $\beta_{12}$ and $\beta_{21}$ are close to each other in the mean, 
but the estimates of $\beta_{11}$ and $\beta_{12}$ are significantly different in the mean.
Similarly, the difference in the estimates of $\beta_{21}$ and $\beta_{22}$ are significant.
In this case, the self-excited effects are less persistent than the mutually excited effects.

\begin{figure}
        \centering
        \begin{subfigure}[b]{0.5\textwidth}
                \includegraphics[width=\textwidth]{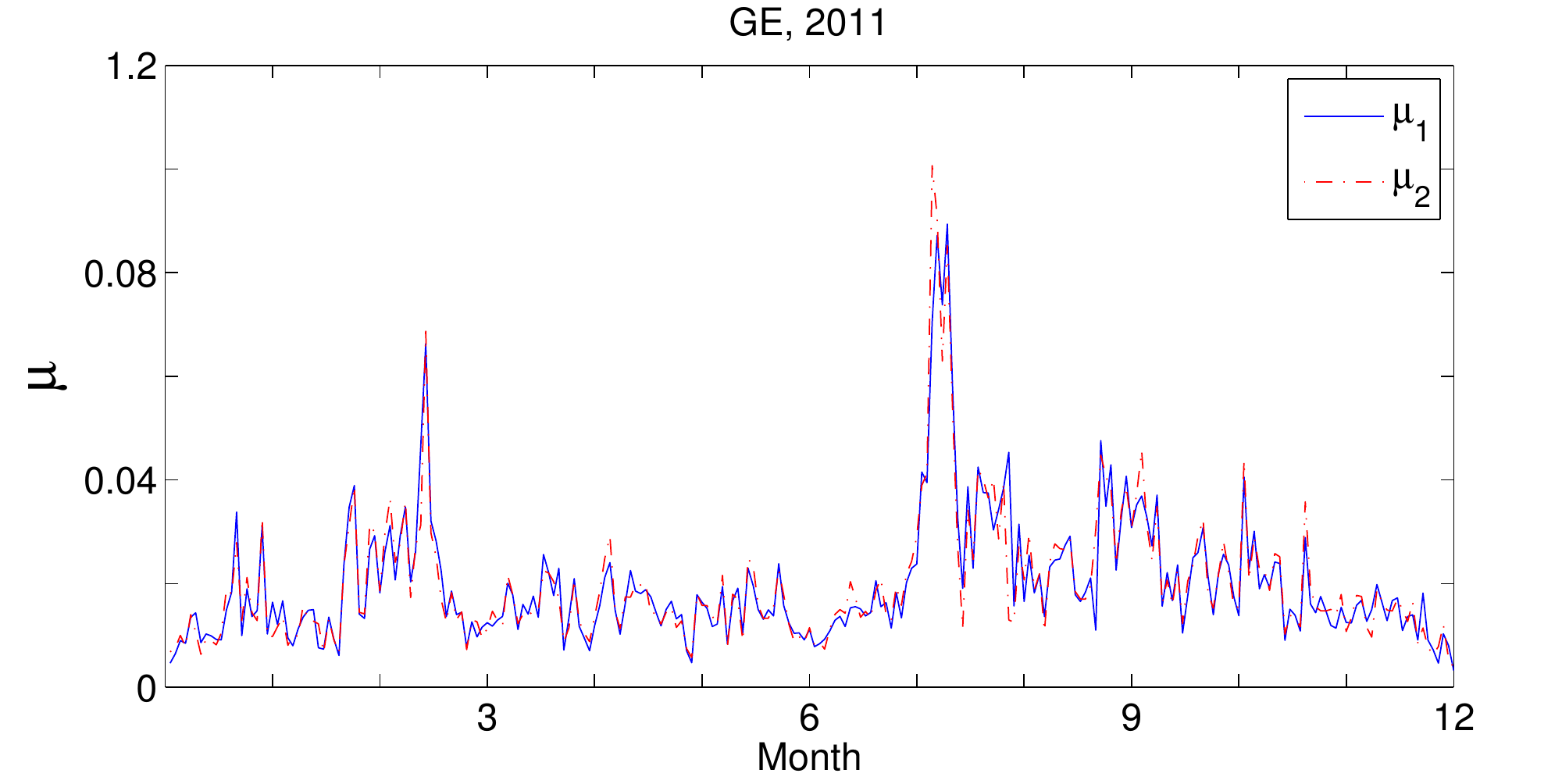}
                \caption{$\mu_1$ and $\mu_2$}
                \label{Fig:ge2011_mu12}
        \end{subfigure}
        \centering
        \begin{subfigure}[b]{0.5\textwidth}
                \includegraphics[width=\textwidth]{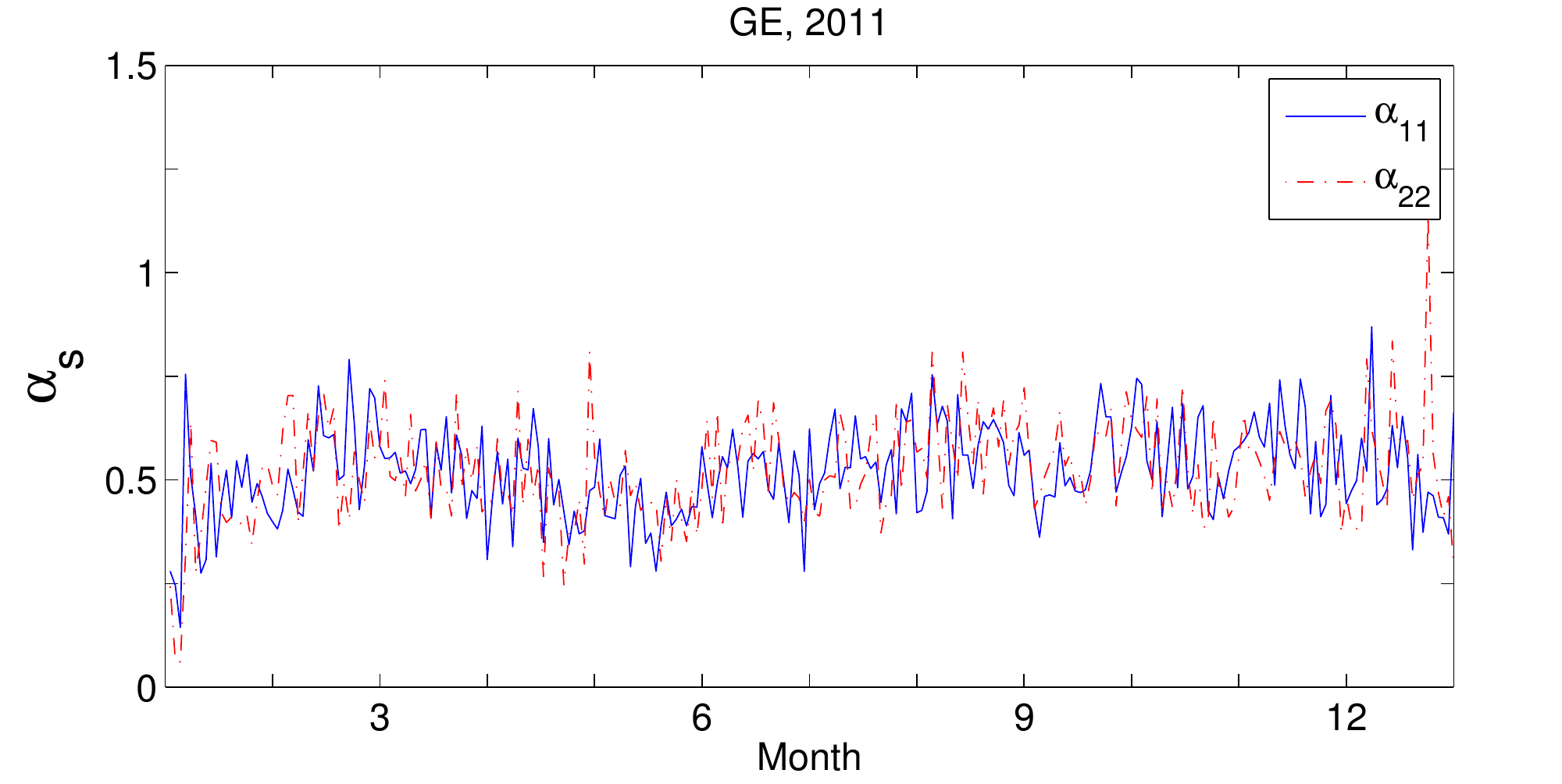}
                \caption{$\alpha_{11}$ and $\alpha_{22}$}
                \label{Fig:ge2011_alpha1122}
        \end{subfigure}
	    \centering
        \begin{subfigure}[b]{0.5\textwidth}
                \includegraphics[width=\textwidth]{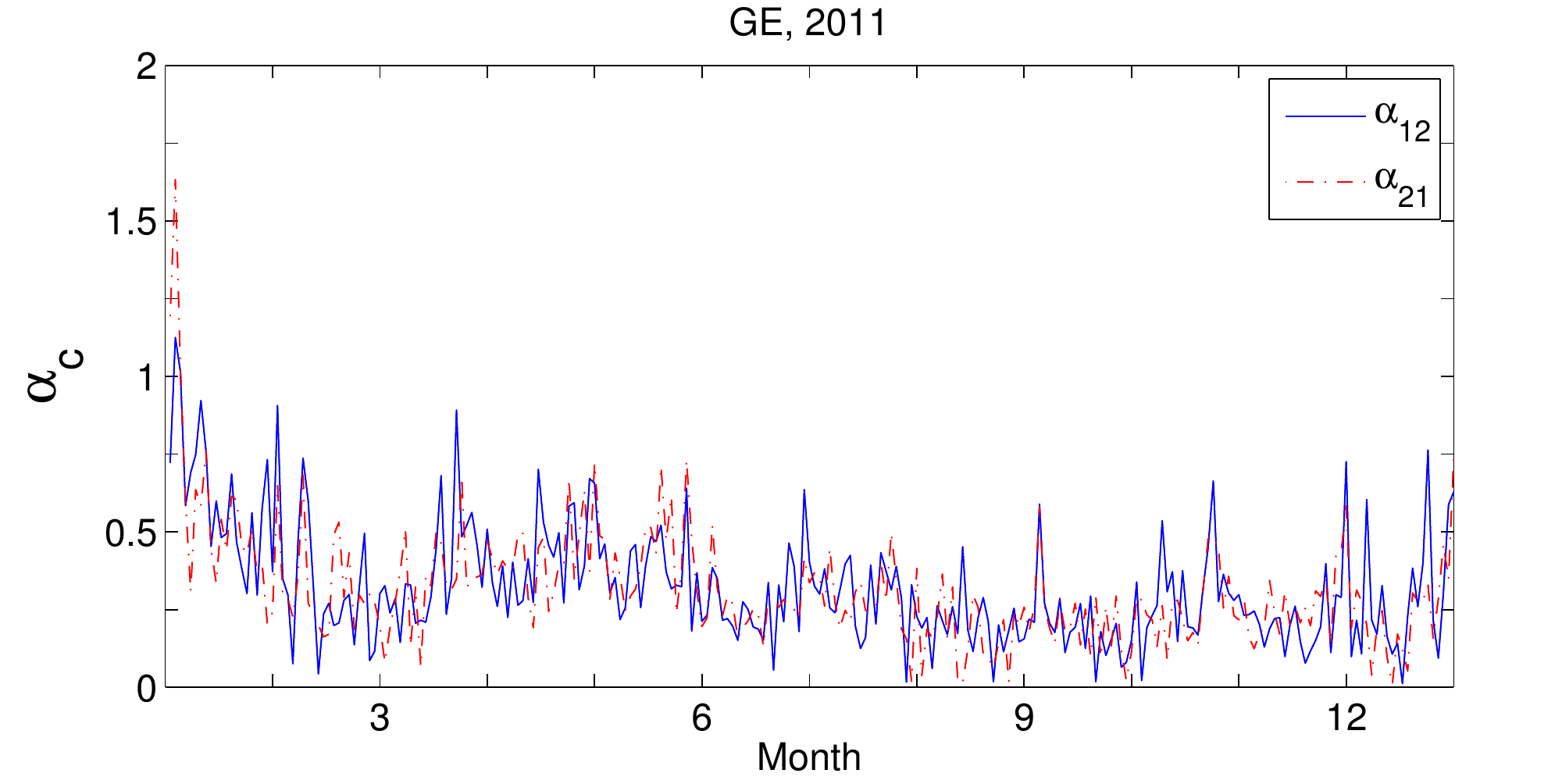}
                \caption{$\alpha_{12}$ and $\alpha_{21}$}
                \label{Fig:ge2011_alpha1221}
        \end{subfigure}
	    \centering
        \begin{subfigure}[b]{0.5\textwidth}
                \includegraphics[width=\textwidth]{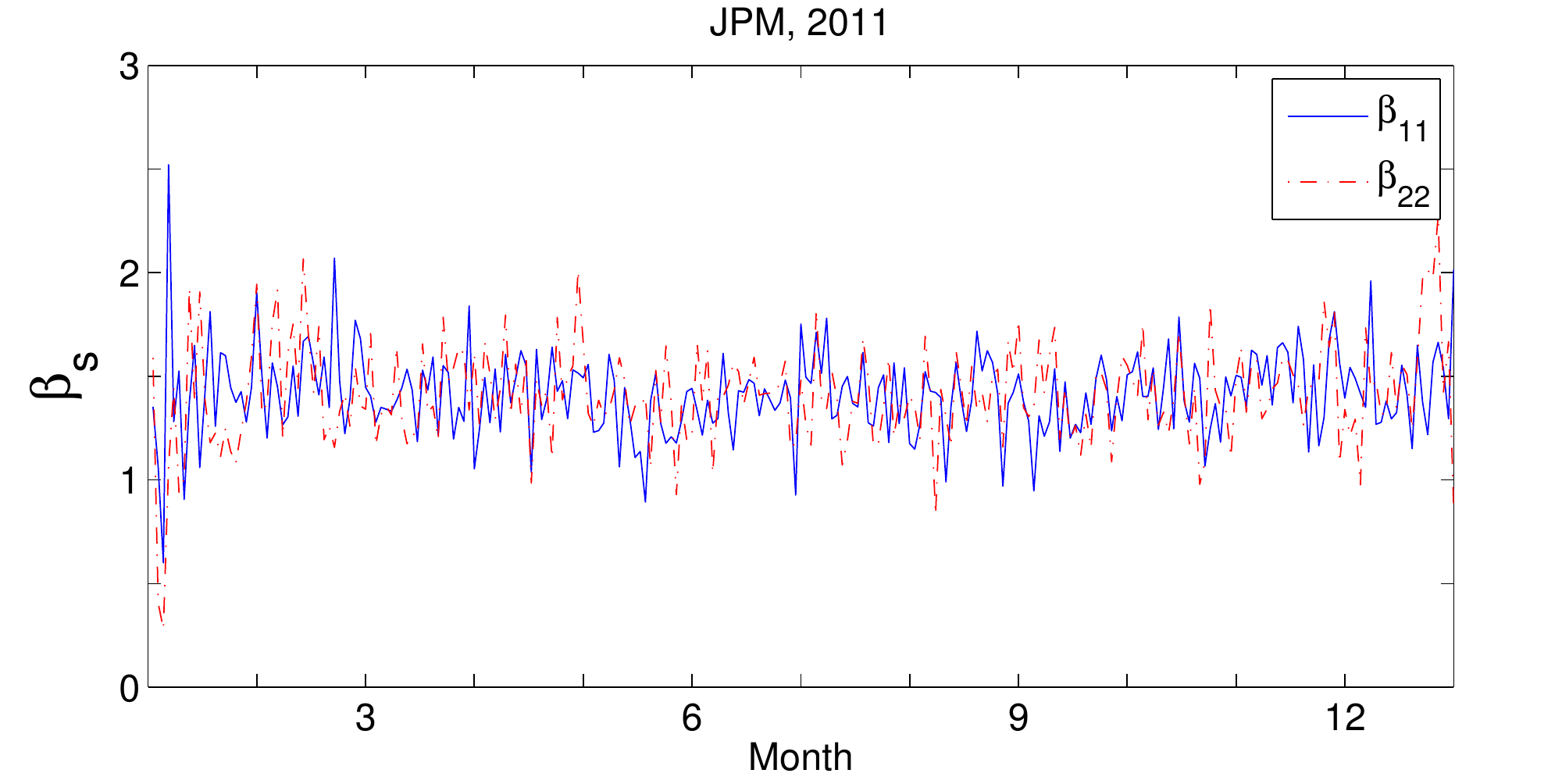}
                \caption{$\beta_{11}$ and $\beta_{22}$}
                \label{Fig:ge2011_beta1122}
        \end{subfigure}
	    \centering
        \begin{subfigure}[b]{0.5\textwidth}
                \includegraphics[width=\textwidth]{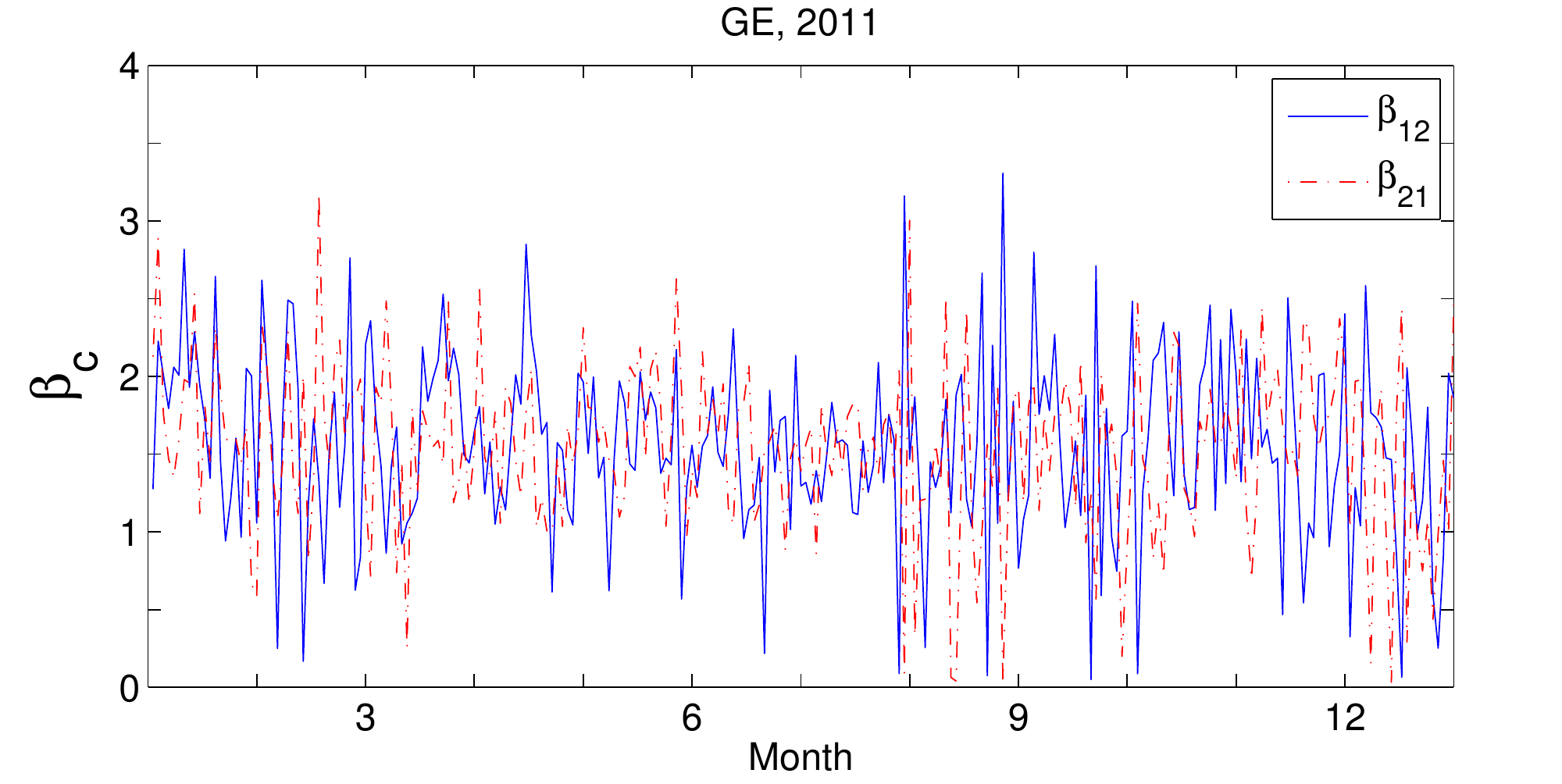}
                \caption{$\beta_{12}$ and $\beta_{21}$}
                \label{Fig:ge2011_beta1221}
        \end{subfigure}
		\caption{Estimation result with the fully characterized Hawkes, GE, 2011}\label{Fig:EstimationGE2011full}
\end{figure}

\begin{table}
\caption{Estimation result of fully characterized self and mutually excited Hawkes process, GE, 2011 in panel A and XOM, 2008 in panel B}\label{Table:GE2011_full}
\centering
\begin{tabular}{ccccccccccccc}

\hline
 & $\mu_1$ & $\mu_2$ & $\alpha_{11}$ & $\alpha_{22}$& $\alpha_{12}$& $\alpha_{21}$ & $\beta_{11}$ &  $\beta_{22}$ & $\beta_{12}$ & $\beta_{21}$\\
\hline
A \\
mean & 0.0198 & 0.0199 & 0.5196 & 0.5228 & 0.3235 & 0.3165 & 1.4145 & 1.4128 & 1.5574 & 1.5378 \\
std. & 0.0124 & 0.0128 & 0.1108 & 0.1241 & 0.1902 & 0.1875 & 0.2097 & 0.2463 & 0.5856 & 0.5402 \\
\hline
B \\
mean & 0.1886 & 0.1727 & 1.0594 & 0.9904 & 0.1369 & 0.1288 & 1.7648 & 1.6916 & 0.8972 & 0.7329\\
std. & 0.1187 & 0.1196 & 0.3274 & 0.3348 & 0.0850 & 0.0786 & 0.4478 & 0.4480 & 0.6104 & 0.5197\\
\hline
\end{tabular}
\end{table}

\subsection{Diffusion parameter}

The parameters of the diffusion model introduced in subsection~\ref{Subsect:diffusion} were estimated using the simulated likelihood estimation explained in subsection~\ref{Subsect:SLE}.
The results of GE, January 2011 are presented in Table~\ref{Table:diff_GE2011} and 
are similar to the results with the Hawkes model in Table~\ref{Table:GE2011}.
In addition, the estimates of the diffusion model with $\rho$ are presented in Table~\ref{Table:diff_GE2011_2}.
The estimates in each model show the similar patterns over the period.

The diffusion estimation has its own pros and cons.
In our setting, because the observed values of the price over one minute intervals are only used, 
which is in contrast to the Hawkes modeling where all times of price changes are used,
the diffusion estimator is less efficient than the Hawkes estimator.
In addition, by the nature of the simulated likelihood estimation, 
it takes longer time to compute the likelihoods and the computed results are not deterministic but depend on the random numbers generated by computers.
On the other hand, 
when the observing times of a price process are limited, i.e., the prices are only available at each one minute interval,
the diffusion model and its estimation are a feasible alternate choice to examine the nature of the price movements in high-frequency.

\begin{table}
\caption{Diffusion model estimation result, GE, January 2011}\label{Table:diff_GE2011}
\centering
\begin{tabular}{cccccccccc}
\hline
Date & $m$ & $a_s$ & $a_c$ & $b$ & volatility\\
\hline
0103 & 0.0137 & 0.1225 & 1.9872 & 2.5512 & 0.1392 \\
0104 & 0.0091 & 0.0805 & 2.9568 & 3.3122 & 0.1546 \\
0105 & 0.0089 & 0.7037 & 2.7975 & 3.5694 & 0.3661 \\
0106 & 0.0082 & 1.0110 & 2.1844 & 3.3689 & 0.2511 \\
0107 & 0.0196 & 0.8574 & 1.3001 & 2.5649 & 0.2557 \\
0110 & 0.0140 & 0.7615 & 0.7941 & 2.3103 & 0.1739 \\
0111 & 0.0069 & 0.5888 & 2.1067 & 2.9101 & 0.1710 \\
0112 & 0.0046 & 0.3888 & 1.5751 & 2.4639 & 0.1247 \\
0113 & 0.0105 & 0.6536 & 0.7633 & 2.3407 & 0.1314 \\
0114 & 0.0043 & 0.9421 & 1.2195 & 2.2284 & 0.2833 \\
0118 & 0.0093 & 0.7737 & 0.7164 & 1.7152 & 0.2325 \\
0119 & 0.0174 & 0.5661 & 1.1345 & 2.3551 & 0.1714 \\
0120 & 0.0173 & 0.7892 & 1.0437 & 2.6670 & 0.1844 \\
0121 & 0.0342 & 0.5358 & 1.1074 & 2.3377 & 0.2215 \\
0124 & 0.0142 & 0.4669 & 0.8490 & 2.0244 & 0.1349 \\
0125 & 0.0208 & 0.7222 & 0.9767	& 2.1042 & 0.2308 \\
0126 & 0.0135 & 0.6058 & 1.0571 & 2.4183 & 0.1387 \\
0127 & 0.0157 & 0.6458 & 0.8145	& 1.8174 & 0.2043 \\
\hline
\end{tabular}
\end{table}

\begin{table}
\caption{Diffusion model estimation result with $\rho$, GE, January 2011}\label{Table:diff_GE2011_2}
\centering
\begin{tabular}{cccccccccc}
\hline
Date & $m$ & $a_s$ & $a_c$ & $b$ & $\rho$\\
\hline
0103 & 0.0140 & 0.0191 & 1.9107 & 2.4863 & 0.2010 \\
0104 & 0.0094 & 0.4947 & 2.3869 & 3.2896 & 0.1106 \\
0105 & 0.0083 & 0.4363 & 3.0567 & 3.5328 & -0.0338 \\
0106 & 0.0097 & 0.7811 & 2.3254 & 3.4391 & -0.0952 \\
0107 & 0.0196 & 0.8110 & 1.5196 & 2.6318 & -0.3633 \\
0110 & 0.0067 & 0.5657 & 1.7872 & 2.3712 & 0.0397 \\
0111 & 0.0072 & 0.7457 & 1.7410 & 2.5917 & 0.0229 \\
0112 & 0.0116 & 0.4343 & 1.4936 & 2.7054 & -0.1048 \\
0113 & 0.0114 & 0.7440 & 0.8413 & 2.2652 & -0.0709 \\
0114 & 0.0051 & 0.8205 & 1.2866 & 2.1705 & -0.0087 \\
0118 & 0.0045 & 0.1585 & 1.6107 & 1.7917 & 0.1269 \\
0119 & 0.0185 & 0.4916 & 1.1835 & 2.3827 & -0.2130 \\
0120 & 0.0153 & 0.5343 & 0.9029 & 1.9994 & 0.1086 \\
0121 & 0.0164 & 0.6817 & 1.1065 & 1.9896 & -0.1220 \\
0124 & 0.0148 & 0.5293 & 1.0546 & 1.9986 & -0.0055 \\
0125 & 0.0252 & 0.6536 & 0.9445 & 2.0343 & 0.1237\\
0126 & 0.0156 & 0.5248 & 1.0506 & 2.3678 & 0.0473 \\
0127 & 0.0167 & 0.5304 & 0.8537 & 1.6645 & 0.0800 \\
\hline
\end{tabular}
\end{table}

\section{Conclusion}\label{Sect:concl}

This paper examined the empirical performance of the symmetric Hawkes process which is a simple model to consider for both clustering property and market microstructure noise in volatility estimation using the stock prices in the S\&P 500.
The daily dynamics of the Hawkes parameters, the comparison between the Hawkes volatility and the realized volatility and the intraday volatility estimation procedure are discussed.
The diffusion analogy of the symmetric Hawkes model was also proposed to provide the analytical simplicity for computing the distributional properties.
The diffusion model also incorporates the clustering effect, market microstrucutre noise, in addition to asymmetric property.

The volatility could be estimated over a relatively short time interval with the Hawkes model and the intraday variations of volatility was demonstrated.
A comparison between the Hawkes volatility and TSRV showed the difference around 15-25\%.
The parameter restriction, asymmetry and parameter variations might be the cause of the discrepancy but more work will be needed to understand the exact reason. 
The estimation results of the diffusion model were provided where similar patterns to the Hawkes model parameters were observed.

\bibliography{Hawkes}
\bibliographystyle{apalike}

\appendix

\section{Expected intensity of fully characterized Hawkes model}\label{Sect:intensity}
Consider the conditional expectation of the intensity processes:
$$\ell_{i}(t|s) = \E [\lambda_{i}(t) | \F_s], \quad \ell_{ij}(t|s) = \E [\lambda_{ij}(t) | \F_s].$$
Then, for each $i$,
\begin{align*}
\ell_{ii}(t|s) &= \E\left[ \left.\e^{-\beta_{ii}(t-s)}\int_{-\infty}^{s} \alpha_{ii} \e^{-\beta_{ii}(s-u)}\D N_i(u) + \int_s^t \alpha_{ii} \e^{-\beta_{ii}(t-u)}\D N_i(u) \right| \F_s\right]\\
&= \lambda_{ii}(s)\e^{-\beta_{ii}(t-s)} + \E \left[ \left. \int_{s}^{t} \alpha_{ii} \e^{-\beta_{ii}(t-u)}\D N_i(u) \right| \F_s \right]\\
&= \lambda_{ii}(s)\e^{-\beta_{ii}(t-s)} + \E \left[ \left. \int_{s}^{t} \alpha_{ii} \e^{-\beta_{ii}(t-u)}(\D N_i(u) - \lambda_i(u) \D u) + \int_{s}^{t} \alpha_{ii} \e^{-\beta_{ii}(t-u)} \lambda_i(u) \D u \right| \F_s \right]\\
&= \lambda_{ii}(s)\e^{-\beta_{ii}(t-s)} + \int_s^t \alpha_{ii} \e^{-\beta_{ii} (t-u)}\ell_{i}(u|s) \D u
\end{align*}
and by differentiating both sides with respect to $t$,
\begin{align*}
\frac{\D \ell_{ii}(t|s)}{\D t} &= -\lambda_{ii}(s) \beta_{ii}\e^{-\beta_{ii}(t-s)} + \alpha_{ii} \ell_{i}(t|s) - \int_s^t \alpha_{ii}\beta_{ii} \e^{-\beta_{ii} (t-u)}\ell_i (u|s) \D u \\
&= \alpha_{ii} \ell_{i}(t|s) - \beta_{ii}\ell_{ii}(t|s)\\
&= \alpha_{ii} \mu_i + (\alpha_{ii} - \beta_{ii} )\ell_{ii}(t|s) + \alpha_{ii}\ell_{ij}(t|s).
\end{align*}
In addition, with the similar method, for $i \neq j$,
\begin{align*}
\frac{\D \ell_{ij}(t|s)}{\D t} = \alpha_{ij} \left( \mu_j + \ell_{ji}(t|s) - \ell_{jj}(t|s) \right) - \beta_{ij}\ell_{ij}(t|s).
\end{align*}
The differential equation system is represented by the matrix form
$$
\begin{bmatrix}
\ell_{11}'(t|s) \\
\ell_{12}'(t|s) \\
\ell_{21}'(t|s) \\
\ell_{22}'(t|s) \\
\end{bmatrix}
= 
\begin{bmatrix}
\alpha_{11}-\beta_{11} & \alpha_{11} & 0 & 0\\
0 & -\beta_{12} & \alpha_{12} & \alpha_{12} \\
\alpha_{21} & \alpha_{21} & -\beta_{21} & 0 \\
0 & 0 & \alpha_{22} & \alpha_{22} - \beta_{22}
\end{bmatrix}
\begin{bmatrix}
\ell_{11}(t|s) \\
\ell_{12}(t|s) \\
\ell_{21}(t|s) \\
\ell_{22}(t|s)
\end{bmatrix}
+
\begin{bmatrix}
\alpha_{11} \mu_{1} \\
\alpha_{12} \mu_{2} \\
\alpha_{21} \mu_{1} \\
\alpha_{22} \mu_{2} 
\end{bmatrix}.
$$
When the eigenvalues of the matrix are negative, the particular solution of the system becomes the long-run expectations of the intensities and are given by
$$
\begin{bmatrix}
\ell_{11}(t|s) \\
\ell_{12}(t|s) \\
\ell_{21}(t|s) \\
\ell_{22}(t|s)
\end{bmatrix}
=
\frac{1}{H}\begin{bmatrix}
\alpha_{11} \beta_{21} \{ (\beta_{22} - \alpha_{22}) \beta_{12} \mu_1 + \alpha_{12} \beta_{22} \mu_2 \}\\
\alpha_{12} \beta_{22} \{ (\beta_{11} -\alpha_{11}) \beta_{21} \mu_{2} + \alpha_{21} \beta_{11} \mu_{1}\}\\
\alpha_{21} \beta_{11} \{ (\beta_{22}-\alpha_{22}) \beta_{12} \mu_1 + \alpha_{11}\beta_{22}\mu_2 \}\\
\alpha_{22} \beta_{12} \{(\beta_{11}-\alpha_{11}) \beta_{21} \mu_{2} + \alpha_{21} \beta_{11} \mu_{1} \}
\end{bmatrix}
$$
as $t\rightarrow \infty$, where $$H = \alpha_{11} \beta_{12} \beta_{21} (\alpha_{22}-\beta_{22})-\beta_{11} (\alpha_{22} \beta_{12} \beta_{21}+\alpha_{12} \alpha_{21} \beta_{22}-\beta_{12} \beta_{21} \beta_{22}).$$ 
The above formulas can be used as the presumed initial values of the intensity processes in the simulations or estimation procedures.
With full characterization of the parameters $\mu_{ij}, \alpha_{ij}, \beta_{ij}$, the system is four dimensional and the solution is rather complicated.

\section{Simulation method}\label{Sect:simul}
If the decaying parameters $\beta_{ij}$ are different from each other, the system of the self and mutually excited Hawkes and intensity processes $(N_1, N_2, \lambda_1, \lambda_2)$ are not Markov.
As shown in Eq.~\eqref{Eq:lambda1}, $\lambda_1(t)$ depends on both $\lambda_{11}(s)$ and $\lambda_{12}(s)$, for $s<t$, and similarly, $\lambda_2(t)$ depends on both $\lambda_{21}(s)$ and $\lambda_{22}(s)$.
On the other hand, the whole system of the processes $(N_1, N_2, \lambda_{11}, \lambda_{12}, \lambda_{21}, \lambda_{22})$ are Markov and to generate the future paths,
it is only important to know the current values of $(N_1, N_2, \lambda_{11}, \lambda_{12}, \lambda_{21}, \lambda_{22})$ not the entire past histories of the processes.
Therefore, for the simulation of the Hawkes process, it is important to compute the distributions of the arrival times determined by each component of the intensities, $\mu_{i}$ and $\lambda_{ij}$.

Suppose that, over a time interval $[s,t)$, there is no jump by $N_1$ and $N_2$;
then the intensities are deterministic and exponentially decaying function is
\begin{align*}
\lambda_{ij}(t) = \lambda_{ij}(s)\e^{-\beta_{ij}(t-s)}.
\end{align*}
Note that $N_i(t) - N_i(s)$ can be represented by the sum of three jump components $N_{i0|s}, N_{ii|s}$, and $N_{ij|s}$ independent upon $\F_s$ with the corresponding intensities $\mu_{i}$, $\lambda_{ij}$, and $\lambda_{ij}$, respectively.
Let $\tau_{ij | s}$ be the first interarrival time of $N_{ij}$ with intensity $\lambda_{ij}$ after $s$.
The probability distribution of $\tau_{ij | s}$ is then represented by
$$ \mathbb P \{\tau_{ij|s} > u \} = \exp \left( -\lambda_{ij}(s) \frac{1 - \e^{-\beta_{ij} u}}{\beta_{ij}} \right). $$
Thus,
$$
\tau_{ij|s} \sim -\frac{1}{\beta_{ij}} \log \left( 1+ \frac{\beta_{ij} \log U}{\lambda_{ij}(s)} \right)
$$
where $U$ is a uniformly distributed random variable over [0,1].
In addition, let $\tau_{i0}$ denote a random variable that follows a Poisson distribution with intensity $\mu_i$.
Then $\min \{ \tau_{10}, \tau_{20}, \tau_{ij|s} \}$ becomes the next jump arrival time after $s$.
After a jump occur, the counting processes are updated accordingly, the intensities are updated, as in Eqs.~\eqref{Eq:lambda1} and \eqref{Eq:lambda2}, and the above procedure is applied repeatedly.

\section{Likelihood function}\label{Sect:likelihood}
Let $t_k$ be the $k$-th jump arrival time of $N_1$ and $\tau_{1 | k}$ be the interarrival time between $k$ and $(k+1)$-th jumps.
Then the conditional cumulative distributions of $\tau_{1 | k}$ at time $t_k$, i.e., with given $\lambda_1(t_k)$, is
$$ F_{\tau_{1 | k}}(u | \lambda_1(t_k)) = 1 - \exp\left( -\int_{t_k}^{t_k+u} \lambda_1(s) \D s\right).$$
Therefore, the conditional density functions is
$$ f_{\tau_{1 | k}}(u | \lambda_1(t_k)) = \lambda_1(t_k) \exp\left( -\int_{t_k}^{t_k + u} \lambda_1(s) \D s\right).$$
Similarly, let $t_m$ be the $m$-th jump arrival time of $N_2$ and $\tau_{2 |m}$ be the interarrival time between the $m$ and $(m+1)$-th jumps.
The conditional density function of $\tau_{2 | m}$ at time $t_m$ is then
$$ f_{\tau_{2 | m}}(u | \lambda_2(t_m)) = \lambda_2(t_m) \exp\left( -\int_{t_m}^{t_m + u} \lambda_2(s) \D s\right).$$

Now consider the interval $[0,T]$ over which the jumps are observed.
The log-likelihood of the realized jump arrivals up to time $T$ is represented by the sum of log-likelihood of all realized arrivals of $N_1$ and $N_2$.
That is
\begin{align*}
L(\theta, T) &= \sum_k \log \left\{ \lambda_1(t_k) \exp\left( -\int_{t_k}^{t_{k+1}} \lambda_1(s) \D s\right) \right\} + \sum_m \log \left\{ \lambda_2(t_m) \exp\left( -\int_{t_m}^{t_{m+1}} \lambda_2(s) \D s\right) \right\}\\
&= \int_{0}^{T} \log \lambda_1(\theta, t)\D N_1(t) +\int_{0}^{T} \log \lambda_2(\theta,t)\D N_2(t) - \int_{0}^{T} (\lambda_1(\theta,t) + \lambda_2(\theta,t))\D t
\end{align*}
where $\theta = \{\mu_{ij}, \alpha_{ij}, \beta_{ij}\}$ denotes the parameter vector. 
The maximum likelihood estimator, $\hat\theta$, is the estimator which maximize $L$ under the observations of realized jump arrivals of $N_1$ and $N_2$.

Define a matrix $I(\theta)$ with each element
$$I_{ij}(\theta) = -\E \left[\int_{0}^T \left(\frac{1}{\lambda_1}\frac{\partial \lambda_1}{\partial \theta_i}\frac{\partial \lambda_1}{\partial \theta_j} + \frac{1}{\lambda_2}\frac{\partial \lambda_2}{\partial \theta_i}\frac{\partial \lambda_2}{\partial \theta_j}\right) \D t\right].$$
The maximum likelihood estimator converges to the true parameter value $\theta_0$ asymptotically normally in distribution with an asymptotic variance-covariance matrix $I^{-1}(\theta_0)$, see \cite{Ogata1978}.
For the maximum likelihood estimation in the statistical package R, consult \cite{Henningsen2011}.

\section{Proof of the variance formula in Proposition~\ref{Prop:vol}}\label{Sect:var_Hawkes}
In this section, the variance formula is derived under the symmetric Hawkes process assumption of the price process.
When the price follows Eq.~\eqref{Eq:price} with symmetric Hawkes process, the variance of the return is represented by
$$ \frac{\delta^2}{S^2(0)} \mathrm{Var}( N_1(t) - N_2(t) - (N_1(0) - N_2(0)).$$
To compute the variance of the return over time interval $[0,t]$, the following results are needed.
The intensities $\lambda_1$ and $\lambda_2$ are assumed to be in the stationary state at time $0$.
Under the assumption, the variance of the price process is derived using the stochastic integration theory.
The quadratic variation of $X$ is defined by
$$ [X]_t = X^2_t - 2\int_0^t X_{s-} \D X_s $$
and the quadratic covariation of $X$ and $Y$ is defined as
$$ [X,Y]_t = X_t Y_t - \int_0^t X_{s-} \D Y_s - \int_0^t Y_{s-} \D X_s.$$
When the processes are quadratic pure jump processes, i.e., the quadratic (co)variation of the continuous part is zero, 
$$ [X]_t = X_0^2 + \sum_{0<s\leq t} (\Delta X_s)^2, \quad [X,Y]_t = X_0 Y_0 + \sum_{0<s\leq t} (\Delta X_s \Delta Y_s).$$
Without a loss of generality, it is assumed that $N_1(0) =N_2(0) = 0$ in this proof.
The next lemma is stated without proof.

\begin{lemma}
Under the stationarity condition of the intensities at time 0,
\begin{align*}
&\textrm{(a) }\E [\lambda_1(t)] = \E [\lambda_2(t)] = \lambda_1(0) = \lambda_2(0)\\
&\textrm{(b) }\E [N_1(t)] = \E [N_2(t)] = \lambda_1(0)t\\
&\textrm{(c) }\E \left[ \left[N_1\right]_t  \right] = \E \left[ \left[N_2\right]_t  \right]= \E \left[ N_1(t) \right] = \lambda_1(0)t\\
&\textrm{(d) }\E \left[\left[\lambda_1\right]_t\right] = \E \left[\left[\lambda_2\right]_t\right] = \lambda_1^2(0) + (\alpha_s^2 + \alpha_c^2) \lambda_1(0)t\\
&\textrm{(e) }\E \left[\left[\lambda_1, \lambda_2\right]_t\right] = \lambda_1^2(0) + 2\alpha_s \alpha_c \lambda_1(0) t\\
&\textrm{(f) }\E \left[ \left[N_1,\lambda_1 \right]_t \right] = \E \left[ \left[N_2,\lambda_2 \right]_t \right] = \alpha_s \E [N_1 (t)] = \alpha_s \lambda_1(0) t\\
&\textrm{(g) }\E \left[ \left[N_1,\lambda_2 \right]_t \right] = \E \left[ \left[N_2,\lambda_1 \right]_t \right] = \alpha_c \E [N_1 (t)] = \alpha_c \lambda_1(0) t
\end{align*}
\end{lemma}

Recall that 
$$
M = 
\begin{bmatrix}
\alpha_s-\beta & \alpha_c \\
\alpha_c & \alpha_s-\beta
\end{bmatrix}.
$$
\begin{lemma}
Under the stationarity condition of the intensities at time 0,
\begin{align*}
\begin{bmatrix}
\E [ \lambda_1^2(t)] \\
\E [ \lambda_1(t)\lambda_2(t) ] 
\end{bmatrix}
= c_1 \begin{bmatrix}-1\\1\end{bmatrix}\e^{2\xi_1 t} + c_2 \begin{bmatrix}1\\1\end{bmatrix} \e^{2\xi_2 t}
-\frac{1}{2}\lambda_1(0) M^{-1}
\begin{bmatrix}
\alpha_s^2 + \alpha_c^2 + 2\beta\mu \\
2(\alpha_s \alpha_c + \beta\mu) 
\end{bmatrix}
\end{align*}
for some constant $c_1$ and $c_2$.
\end{lemma}

\begin{proof}
Note that
\begin{align*}
\E [ \lambda_1^2(t)] &=  \E \left[\left[\lambda_1\right]_t\right] + 2 \E \left[\int_0^t \lambda_1(u) \D \lambda_1(u) \right] \\ 
&= \lambda_1^2(0) + (\alpha_s^2 + \alpha_c^2) \lambda_1(0)t + 2 \E \left[\int_0^t \left\{ \beta\mu\lambda_1(u) + (\alpha_s-\beta)\lambda_1^2(u) + \alpha_c\lambda_1(u)\lambda_2(u) \right\} \D u \right]\\
&= \lambda_1^2(0) + (\alpha_s^2 + \alpha_c^2 + 2\beta\mu) \lambda_1(0) t  + 2\int_0^t \left\{ (\alpha_s-\beta) \E_s [\lambda_1^2(u)] + \alpha_c \E_s [ \lambda_1(u)\lambda_2(u) ] \right\}\D u
\end{align*}
and
\begin{align*}
\E [ \lambda_1(t)\lambda_2(t)] &= \E \left[ \left[\lambda_1,\lambda_2\right]_t\right] + \E \left[ \int_0^t \lambda_1(u) \D \lambda_2(u)\right] + \E \left[ \int_0^t \lambda_2(u) \D \lambda_1(u)\right]\\
&=\lambda_1^2(0) + 2\alpha_s \alpha_c \lambda_1(0)t + 2 \E \left[ \int_0^t \left\{ \beta\mu\lambda_1(u) + \alpha_c  \lambda_1^2 (u) + (\alpha_s- \beta) \lambda_1(u)\lambda_2(u) \right\} \D u \right]\\
&=\lambda_1^2(0) + 2(\alpha_s \alpha_c + \beta\mu) \lambda_1(0)t + 2 \int_0^t \left\{\alpha_c \E[ \lambda_1^2 (u)] + (\alpha_s- \beta) \E [\lambda_1(u)\lambda_2(u)] \right\} \D u.
\end{align*}
Therefore, a system of equations can be derived:
$$
\begin{bmatrix}
\dfrac{\D \E[\lambda_1^2(t)]}{\D t}\\
\dfrac{\D\E[\lambda_1(t)\lambda_2(t)]}{\D t}
\end{bmatrix}
= 2
\begin{bmatrix}
\alpha_s-\beta & \alpha_c \\
\alpha_c & \alpha_s-\beta
\end{bmatrix}
\begin{bmatrix}
\E[\lambda_1^2(t)]\\
\E[\lambda_1(t)\lambda_2(t)]
\end{bmatrix}
+\lambda_1(0)
\begin{bmatrix}
\alpha_s^2 + \alpha_c^2 + 2\beta\mu \\
2(\alpha_s \alpha_c + \beta\mu) 
\end{bmatrix}.
$$
The particular solution of the system is 
\begin{align*}
-\frac{1}{2}\lambda_1(0) M^{-1}
\begin{bmatrix}
\alpha_s^2 + \alpha_c^2 + 2\beta\mu  \\
2(\alpha_s \alpha_c + \beta\mu)
\end{bmatrix}
=\frac{1}{2}\lambda_1(0)
\begin{bmatrix}
\dfrac{2 \beta  \mu  \alpha _c+\alpha _c^2 \left(\beta +\alpha _s\right)+\left(\beta -\alpha _s\right) \left(2 \beta  \mu +\alpha _s^2\right)}{\left(\beta -\alpha _s\right){}^2 -\alpha _c^2}\\
\dfrac{\alpha _c^3+2 \beta  \mu  \left(\beta -\alpha _s\right)+\alpha _c \left(2 \beta  \mu +2 \beta  \alpha _s-\alpha _s^2\right)}{\left(\beta -\alpha _s\right){}^2 - \alpha _c^2}
\end{bmatrix}
\end{align*}
where the inverse matrix of $M$ is represented by
$$ 
M^{-1} =  \begin{bmatrix} \frac{\beta-\alpha_s}{\alpha_c^2 - (\beta-\alpha_s)^2} & \frac{\alpha_c}{\alpha_c^2 - (\beta-\alpha_s)^2} \\ \frac{\alpha_c}{\alpha_c^2 - (\beta-\alpha_s)^2} & \frac{\beta-\alpha_s}{\alpha_c^2 - (\beta-\alpha_s)^2}\end{bmatrix} = 
\frac{1}{\xi_1\xi_2}\begin{bmatrix} \alpha_s - \beta & -\alpha_c \\ -\alpha_c & \alpha_s-\beta \end{bmatrix}.
$$
In addition,  the general solution is
$$ \begin{bmatrix}
\E [ \lambda_1^2(t)] \\
\E [ \lambda_1(t)\lambda_2(t) ] 
\end{bmatrix} 
= c_1 \begin{bmatrix}-1\\1\end{bmatrix}\e^{2\xi_1 t} + c_2 \begin{bmatrix}1\\1\end{bmatrix} \e^{2\xi_2 t}
-\frac{1}{2}\lambda_1(0) M^{-1}
\begin{bmatrix}
\alpha_s^2 + \alpha_c^2 + 2\beta\mu  \\
2(\alpha_s \alpha_c + \beta\mu)
\end{bmatrix}
$$
and with the initial condition of~\eqref{Eq:initial},
$$ c_1 = -\frac{\lambda_1(0)(\alpha_s-\alpha_c)^2}{4\xi_1}, \quad c_2 = \frac{\lambda_1(0)(\alpha_s+\alpha_c)^2 }{4\xi_2}. $$
\end{proof}

\begin{lemma}
Under the stationary state condition of the intensities at time 0, we have
\begin{align*}
\begin{bmatrix}
\E[\lambda_1(t)N_1(t)]\\
\E[\lambda_1(t)N_2(t)]
\end{bmatrix}
={}& d_1 \begin{bmatrix}-1\\1\end{bmatrix}\e^{\xi_1 t} + d_2 \begin{bmatrix}1\\1\end{bmatrix} \e^{\xi_2 t} + \frac{c_1}{\xi_1} \begin{bmatrix}-1\\1\end{bmatrix}\e^{2\xi_1 t} + \frac{c_2}{\xi_2} \begin{bmatrix}1\\1\end{bmatrix} \e^{2\xi_2 t} 
\\
&-\lambda_1(0) \left\{ \beta\mu M^{-1}\begin{bmatrix}1\\1\end{bmatrix}t + M^{-1} 
\begin{bmatrix}
\alpha_s\\
\alpha_c
\end{bmatrix} - \frac{1}{2} (M^{-1})^2 \begin{bmatrix}
\alpha_s^2 + \alpha_c^2  \\
2\alpha_s\alpha_c
\end{bmatrix}\right\}.
\end{align*}
\end{lemma}

\begin{proof}
Note that
\begin{align*}
\E[\lambda_1(t)N_1(t)] ={}& \E\left[ \left[\lambda_1, N_1 \right]_t\right] + \E \left[\int_0^t \lambda_1(u-) \D N_1(u) \right] + \E \left[\int_0^t N_1(u) \D \lambda_1(u) \right] \\
={}& \alpha_s \lambda_1(0) t + \E \left[\int_0^t \lambda^2_1(u) \D u \right] \\
&+ \int_0^t \left\{\beta\mu \E[ N_1(u) ] + (\alpha_s-\beta) \E [\lambda_1(u)N_1(u)] + \alpha_c \E [ \lambda_2(u)N_1(u) ] \right\}\D u\\
={}& \alpha_s \lambda_1(0) t +  \int_0^t \E[\lambda^2_1(u)] \D u  + \int_0^t \beta\mu\lambda_1(0) u \D u\\
&+ \int_0^t \left\{(\alpha_s-\beta) \E [\lambda_1(u)N_1(u)] + \alpha_c \E [ \lambda_1(u)N_2(u) ] \right\}\D u
\end{align*}
and
\begin{align*}
\E[\lambda_1(t)N_2(t)] ={}& \E\left[ \left[\lambda_1, N_2 \right]_t\right] + \E \left[\int_0^t \lambda_1(u-) \D N_2(u) \right] + \E \left[\int_0^t N_2(u) \D \lambda_1(u) \right] \\
={}& \alpha_c \lambda_1(0) t + \E \left[\int_0^t \lambda_1(u)\lambda_2(u) \D u \right] \\
&+ \int_0^t \left\{\beta\mu \E[ N_2(u) ] + (\alpha_s-\beta) \E [\lambda_1(u)N_2(u)] + \alpha_c \E[ \lambda_2(u)N_2(u) ] \right\}\D u\\
={}& \alpha_c \lambda_1(0) t + \int_0^t \E [\lambda_1(u)\lambda_2(u)] \D u + \int_0^t \beta\mu\lambda_1(0) u \D u\\
&+ \int_0^t \left\{(\alpha_s-\beta) \E [\lambda_1(u)N_2(u)] + \alpha_c \E [ \lambda_1(u)N_1(u) ] \right\}\D u.
\end{align*}
In the above, $\E[\lambda_2(u)N_1(u)]$ is replaced with $\E[\lambda_1(u)N_2(u)]$ due to the symmetry of the model and similarly, $\E[\lambda_2(u)N_2(u)]$ is replaced with $\E[\lambda_1(u)N_1(u)]$.
Therefore, following system of equations can be derived:
\begin{align*}
\begin{bmatrix}
\dfrac{\D \E[\lambda_1(t)N_1(t)]}{\D t}\\
\dfrac{\D \E[\lambda_1(t)N_2(t)]}{\D t}
\end{bmatrix}
={}& M
\begin{bmatrix}
\E[\lambda_1(t)N_1(t)] \\
\E[\lambda_1(t)N_2(t)]
\end{bmatrix}
+
\begin{bmatrix}
\alpha_s \lambda_1(0) + \E[\lambda^2_1(t)] + \beta\mu\lambda_1(0) t\\
\alpha_c \lambda_1(0) + \E [\lambda_1(t)\lambda_2(t)] + \beta\mu\lambda_1(0) t
\end{bmatrix}\\
={}& M
\begin{bmatrix}
\E[\lambda_1(t)N_1(t)] \\
\E[\lambda_1(t)N_2(t)]
\end{bmatrix}
+ c_1 \lambda_1(0)\begin{bmatrix}-1\\1\end{bmatrix}\e^{2\xi_1 t} + c_2\lambda_1(0) \begin{bmatrix}1\\1\end{bmatrix} \e^{2\xi_2 t}\\
&+ \lambda_1(0) \left(
\begin{bmatrix}
\beta\mu\\
\beta\mu
\end{bmatrix}t
+ \begin{bmatrix}
\alpha_s \\
\alpha_c 
\end{bmatrix}
-\frac{1}{2} M^{-1}
\begin{bmatrix}
\alpha_s^2 + \alpha_c^2 + 2\beta\mu  \\
2(\alpha_s \alpha_c + \beta\mu)
\end{bmatrix}
\right)
\end{align*}
where the previous lemma is used.
The particular solution is
$$
\begin{bmatrix} A_1\\A_2 \end{bmatrix}t
+
\begin{bmatrix} B_1\\B_2 \end{bmatrix}
+
k_1 \begin{bmatrix}-1\\1\end{bmatrix}\e^{2\xi_1 t} + k_2 \begin{bmatrix}1\\1\end{bmatrix} \e^{2\xi_2 t}
$$
where
$$
\begin{bmatrix} A_1\\A_2\end{bmatrix}
= -\lambda_1(0) \beta \mu M^{-1}
\begin{bmatrix}
1\\
1
\end{bmatrix}
$$
and
$$
\begin{bmatrix}
B_1\\
B_2
\end{bmatrix}
=
-\lambda_1(0) \left(M^{-1}
\begin{bmatrix}
\alpha_s\\
\alpha_c
\end{bmatrix}
-\frac{1}{2}(M^{-1})^2
\begin{bmatrix}
\alpha_s^2 + \alpha_c^2  \\
2\alpha_s\alpha_c
\end{bmatrix}
\right)
$$
and
$ k_1 = c_1/\xi_1, k_2 = c_2/\xi_2 $.
The general solution is
$$ \begin{bmatrix} \E[\lambda_1(t)N_1(t)]\\ \E[\lambda_1(t)N_2(t)] \end{bmatrix} =
d_1 \begin{bmatrix}-1\\1\end{bmatrix}\e^{\xi_1 t} + d_2 \begin{bmatrix}1\\1\end{bmatrix} \e^{\xi_2 t} 
+
\frac{c_1}{\xi_1} \begin{bmatrix}-1\\1\end{bmatrix}\e^{2\xi_1 t} + \frac{c_2}{\xi_2} \begin{bmatrix}1\\1\end{bmatrix} \e^{2\xi_2 t} 
+
\begin{bmatrix}
A_1\\
A_2
\end{bmatrix}t
+
\begin{bmatrix}
B_1\\
B_2
\end{bmatrix} $$
and with the initial condition,
$$ d_1 = \frac{\lambda_1(0) (\alpha_s-\alpha_c)\beta}{2 \xi_1^2 } , \quad d_2 = -\frac{\lambda_1(0)(\alpha_s + \alpha_c)\beta}{2 \xi_2^2}.$$
\end{proof}

\begin{proposition}
Under the stationary state condition of the intensities at time 0,
\begin{align*}
\begin{bmatrix}
\E[N_1^2(t)]\\
\E[N_1(t) N_2(t)]
\end{bmatrix}
={}& \frac{2d_1}{\xi_1} \begin{bmatrix}-1\\1\end{bmatrix}(\e^{\xi_1 t}-1) + \frac{2d_2}{\xi_2} \begin{bmatrix}1\\1\end{bmatrix} (\e^{\xi_2 t}-1) \\
&+
\frac{c_1}{\xi_1^2} \begin{bmatrix}-1\\1\end{bmatrix}(\e^{2\xi_1 t}-1) + \frac{c_2}{\xi_2^2} \begin{bmatrix}1\\1\end{bmatrix} (\e^{2\xi_2 t}-1)
\\
&- \lambda_1(0) \left\{ \beta\mu M^{-1} \begin{bmatrix}1\\1\end{bmatrix}t^2 \right.
+\left.
\left(
2M^{-1}
\begin{bmatrix}
\alpha_s\\
\alpha_c
\end{bmatrix}
-(M^{-1})^2
\begin{bmatrix}
\alpha_s^2 + \alpha_c^2  \\
2\alpha_s\alpha_c 
\end{bmatrix}
- \begin{bmatrix}1\\0\end{bmatrix}
\right)t\right\}.
\end{align*}
In addition,
$$\E[(N_1(t)-N_2(t))^2] = 2\lambda_1(0)\left\{ \frac{\beta^2}{\xi_1^2}t - \left( 1- \frac{\beta^2}{\xi_1^2} \right)\frac{\e^{\xi_1 t}-1}{\xi_1} \right\}.$$
\end{proposition}

\begin{proof}
The following is used:
\begin{align*}
\E[N_1^2(t)] &= \E\left[\left[N_1\right]_t\right] + 2\E \left[\int_0^t N_1(u-) \D N_1(u) \right]\\
&= \lambda_1(0)t + 2\E \left[\int_0^t \lambda_1(u) N_1(u) \D u \right]
\end{align*}
and
\begin{align*}
\E[N_1(t)N_2(t)] &= \E\left[\left[N_1,N_2\right]_t\right] + \E \left[\int_0^t N_1(u-) \D N_2(u) \right] + \E \left[\int_0^t N_2(u-) \D N_1(u) \right]\\
&= 2\E \left[\int_0^t \lambda_1(u) N_2(u) \D u \right]
\end{align*}
along with the previous lemmas.
In the above equation, $\E\left[\left[N_1,N_2\right]_t\right]=0$ as the probability of the simultaneous jumps of $N_1$ and $N_2$ is zero.
In addition, $\E[\lambda_1(u)N_2(u)] = \E[\lambda_2(u)N_1(u)]$ under the symmetry.
Therefore,
\begin{align*}
&\E[(N_1(t)-N_2(t))^2]= 2(\E[N_1^2(t)] - \E[N_1(t) N_2(t)])\\
&= -\frac{8d_1}{\xi_1}(\e^{\xi_1 t}-1)- \frac{4 c_1}{\xi_1^2}(\e^{2\xi_1 t}-1) - 2\lambda_1(0)\left\{\frac{2}{\xi_1}(\alpha_s - \alpha_c) - \frac{1}{\xi_1^2}(\alpha_s - \alpha_c)^2  - 1\right\}t\\
&= -\frac{8d_1}{\xi_1}(\e^{\xi_1 t}-1)- \frac{4 c_1}{\xi_1^2}(\e^{2\xi_1 t}-1)  + 2\lambda_1(0)\left( \frac{\alpha_s - \alpha_c}{\xi_1} - 1\right)^2 t\\
&= -\lambda_1(0)\frac{4(\alpha_s - \alpha_c)\beta}{\xi_1^3}(\e^{\xi_1 t}-1) + \lambda_1(0)\frac{(\alpha_s -\alpha_c)^2}{\xi_1^3}(\e^{2\xi_1 t}-1)  + 2\lambda_1(0)\frac{\beta^2}{\xi_1^2} t  \\
&= \frac{2\lambda_1(0)}{\xi_1^2}\left\{ \beta^2 t - 2(\alpha_s - \alpha_c)\beta\left(\frac{\e^{\xi_1 t}-1}{\xi_1}\right) + (\alpha_s -\alpha_c)^2\left( \frac{\e^{2\xi_1 t}-1}{2\xi_1} \right) \right\}.
\end{align*}
\end{proof}

\section{Proof of Proposition~\ref{Prop:voldiff}}\label{append:proof1}
Note that
\begin{align*}
&\E \left[\left(\int_0^t n_u \D u + \int_0^t \sqrt{V_u} \D W^s_u \right)^2 \right]\\
&=\E \left[\left(\int_0^t n_u \D u \right)^2 \right] + 2 \E \left[\int_0^t n_u \D u \int_0^t \sqrt{V_u} \D W^s_u\right] + \E \left[ \left( \int_0^t \sqrt{V_u} \D W^s_u \right)^2 \right]\\
&=\E \left[\left(\int_0^t n_u \D u \right)^2 \right] + 2 \E \left[\int_0^t n_u \D u \int_0^t \sqrt{V_u} \D W^s_u\right] + \E \left[ \int_0^t V_u \D u  \right]\\
&=\E \left[\left(\int_0^t n_u \D u \right)^2 \right] + 2 \E \left[\int_0^t n_u \D u \int_0^t \sqrt{V_u} \D W^s_u\right] + \theta t.
\end{align*}
Let
\begin{align*}
x(t):= \E \left[\left(\int_0^t n_u \D u \right)^2 \right] = 2 \E \left[\int_0^t n_s \left(\int_0^s n_u \D u \right) \D s \right]
\end{align*}
and
\begin{align*}
y(t):=&{}\E \left[\int_0^t n_u \D u \int_0^t \sqrt{V_u} \D W_u^s \right] \\
=&{} \E \left[\int_0^t\left(\int_0^s n_u \D u \right)\sqrt{V_s} \D W^s_s \right] + \E \left[\int_0^t \left(\int_0^s \sqrt{V_u} \D W^s_u \right)n_s \D s \right]\\
=&{} \E \left[\int_0^t \left(\int_0^s \sqrt{V_u} \D W^s_u \right)n_s \D s \right].
\end{align*}
By assumming $n_0 = 0$,
\begin{align*}
\E\left[n_s \int_0^s \sqrt{V_u} \D W^s_u \right] ={}& \E \left[-\kappa_1 \int_0^s  n_u \D u \int_0^s \sqrt{V_u} \D W^s_u \right] + \E \left[ \phi \int_0^s \sqrt{V_u} \D W^s_u \int_0^s \sqrt{V_u} \D W^s_u \right]\\
={}& -\kappa_1 y(s) + \phi \int_0^s \E[V_u] \D u\\
={}& -\kappa_1 y(s) + \phi \theta s.
\end{align*}
Thus, 
$$ y'(t) = -\kappa_1 y(t) + \phi \theta t$$
and
$$ y(t) = \frac{\phi \theta \left(\kappa_1 t-1+e^{-\kappa t}\right) }{\kappa_1^2}.$$
In addition,
\begin{align*}
\E \left[ n_s \left(\int_0^s n_u \D u \right)\right]&= \E\left[ -\kappa \left(\int_0^s n_u\D u\right)^2 + \phi \int_0^s \sqrt{V_u}\D W_u\int_0^s n_u \D u  \right]\\
&= -\kappa_1 x(s) + \phi y(s)
\end{align*}
Thus,
$$x(t) = 2 \int_0^t (-\kappa_1 x(s) + \phi y(s)) \D s$$
and
$$x'(t) = - 2 \kappa_1 x(t) + 2\phi y(t) .$$
Therefore,
$$ x(t) = \frac{\theta\phi ^2\left( - e^{-2\kappa_1 t} + 4 e^{-\kappa_1 t} - 3 +2  \kappa_1 t\right) }{2 \kappa_1^3}$$
and the desired result is obtained.

\section{Proof of Proposition~\ref{Prop:tmv}}\label{sect:proof2}
As
$$ \D R_t = \frac{1}{S_0}\D S_t$$
and
\begin{align*}
\D R^2_t &= 2 R_t \D R_t + \D [R]_t \\
&= \frac{2(S_t-S_0)}{S_0^2}\D S_t + \frac{1}{S_0^2}\D [S]_t ,
\end{align*}
we have
$$ \D [R,R^2]_t = \frac{2(S_t-S_0)}{S_0^3}\D [S]_t = \frac{2(S_t-S_0)V_t}{S_0^3}\D t.$$
Note that
\begin{align*}
\E [S_t V_t] ={}& \E \left[\left(S_0 + \int_0^t n_s \D s + \int_0^t \sqrt{V_s}\D W^s_s \right)\left(V_0 + \int_0^t \kappa_2 (\theta-V_s)\D s + \int_0^t \gamma \sqrt{V_s}\D W_s^v \right) \right]\\
={}&\E \left[S_0 V_0 + S_0\int_0^t \kappa_2 (\theta-V_s)\D s + S_0 \int_0^t \gamma \sqrt{V_s}\D W_s^v  \right.\\
&+ V_0 \int_0^t n_s \D s + \kappa_2 \theta t\int_0^t n_s \D s - \int_0^t n_s \D s\int_0^t\kappa_2 V_s\D s + \int_0^t n_s \D s\int_0^t \gamma \sqrt{V_s}\D W_s^v \\
&+ V_0\int_0^t\sqrt{V_s} \D W^s_s + \kappa_2\theta t\int_0^t \sqrt{V_s} \D W_s^s -\kappa_2\int_0^t V_s \D s \int_0^t \sqrt{V_s} \D W_s^s \\
&\left.+ \int_0^t\sqrt{V_s}\D W_s^s \int_0^t\gamma\sqrt{V_s} \D W^v_s \right]\\
={}&\E \left[S_0 \theta  - \int_0^t n_s \D s\int_0^t\kappa_2 V_s\D s + \int_0^t n_s \D s\int_0^t \gamma \sqrt{V_s}\D W_s^v \right.\\
&\left. -\kappa_2\int_0^t V_s \D s \int_0^t \sqrt{V_s} \D W_s^s + \int_0^t\sqrt{V_s}\D W_s^s \int_0^t\gamma\sqrt{V_s} \D W^v_s \right].
\end{align*}
For the last equality, the following assumption is used:
$$ \E[V_s]=V_0=\theta, \quad \E[n_s]=0. $$
The following can be derived:
\begin{align*}
w(t) :={}& \E \left[\int_0^t V_s \D s \int_0^t \sqrt{V_s} \D W_s^s \right] = \E \left[\int_0^t V_s \left(\int_0^s \sqrt{V_u} \D W_u^s \right) \D s\right]\\
={}&\E\left[\int_0^t \left( \int_0^s \kappa_2 (\theta - V_u) \D u \int_0^s \sqrt{V_u} \D W_u^s + \int_0^s \gamma \sqrt{V_u} \D W^v_u \int_0^s \sqrt{V_u} \D W_u^s \right) \D s \right]\\
={}& \E \left[\int_0^t -\kappa_2 \left( \int_0^s V_u \D u \int_0^s \sqrt{V_u} \D W_u^s \right) \D s \right] + \rho\gamma\int_0^t\int_0^s \E[V_u] \D u \D s\\
={}&\int_0^t (-\kappa_2 w(s) + \theta\gamma\rho s) \D s\\
={}&\frac{\gamma \rho \theta}{\kappa_2^2}(\kappa_2 t -1 + \e^{-\kappa_2 t}).
\end{align*}
Similarly
\begin{align*}
q(t) := {}& \E \left[ \int_0^t n_s \D s \int_0^t \sqrt{V_s}\D W_s^v \right] = \E \left[\int_0^t n_s \left(\int_0^s \sqrt{V_u} \D W^v_u \right) \D s \right] \\
={}&\E \left[-\kappa_1 \int_0^t \left(\int_0^s  n_u \D u \int_0^s \sqrt{V_u} \D W^v_u \right) \D s\right] + \E \left[ \phi \int_0^t \left( \int_0^s \sqrt{V_u} \D W^v_u \int_0^s \sqrt{V_u} \D W^s_u  \right) \D s\right]\\
={}& \int_0^t (- \kappa_1 q(s) + \phi \rho \theta s) \D s\\
={}& \frac{\phi \rho \theta}{\kappa_1^2}(\kappa_1 t -1 + \e^{-\kappa_1 t}).
\end{align*}
Let
\begin{align*}
z(t) := \E\left[ \int_0^t n_s \D s \int_0^t V_s \D s \right] = \E \left[ \int_0^t \left(\int_0^s n_u \D u \right) V_s \D s \right] + \E \left[ \int_0^t \left(\int_0^s V_u \D u \right) n_s \D s \right]
\end{align*}
and
\begin{align*}
\E \left[ V_s \int_0^s n_u \D u \right] ={}& \E \left[ \left( V_0 + \int_0^s \kappa_2(\theta-V_u) \D u + \int_0^s \gamma \sqrt{V_u} \D W_u^v \right)\int_0^s n_u \D u \right]\\
={}& \E \left[ -\int_0^s \kappa_2 V_u \D u \int_0^s n_u \D u + \int_0^s \gamma \sqrt{V_u} \D W_u^v \int_0^s n_u \D u \right]\\
={}& -\kappa_2 z(s) + \frac{\gamma\rho\phi\theta}{\kappa_1^2}(\kappa_1 s-1+\e^{-\kappa_1 s})
\end{align*}
and
\begin{align*}
\E \left[ n_s \int_0^s V_u \D u \right] ={}& \E \left[ \left(- \int_0^s \kappa n_u \D u + \int_0^s \phi\sqrt{V_u} \D W^s_u\right) \int_0^s V_u \D u \right]\\
={}& -\kappa z(s) + \frac{\gamma \rho \phi \theta}{\kappa_2^2}(\kappa_2 s -1 + \e^{-\kappa_2 s}).
\end{align*}
Thus,
$$ z(t) = \int_0^t \left\{ - (\kappa_1 + \kappa_2) z(s) + \frac{\gamma\rho\phi\theta}{\kappa^2}(\kappa s-1+\e^{-\kappa s})+ \frac{\gamma \rho \phi \theta}{\kappa_2^2}(\kappa_2 s -1 + \e^{-\kappa_2 s}) \right\} \D s$$
and
$$z(t) = \gamma\theta\rho\phi\frac{-\kappa_1^2-\kappa_2^2-\kappa_1\kappa_2 +(\kappa_1^2\kappa_2 + \kappa_1\kappa_2^2)t + (\kappa_2^2+\kappa_1\kappa_2)\e^{-\kappa_1 t} + (\kappa_1^2+\kappa_1\kappa_2)\e^{-\kappa_2 t} -\kappa_1\kappa_2\e^{-(\kappa_1+\kappa_2)t}}{\kappa_1^2\kappa_2^2(\kappa_1+\kappa_2)}.$$
Therefore,
\begin{align*}
\E[S_t V_t] = S_0 \theta - \kappa_2 z(t) + \frac{\gamma\rho\phi\theta}{\kappa_1^2}(\kappa_1 t-1+\e^{-\kappa_1 t}) - \frac{\gamma \rho \theta}{\kappa_2}(\kappa_2 t -1 + \e^{-\kappa_2 t}) + \gamma\theta\rho t
\end{align*}
and hence
\begin{align*}
\int_0^t \E[S_u V_u] \D u ={}& S_0 \theta t + \frac{\gamma\theta\rho}{2} t^2 + \frac{\gamma\rho\phi\theta}{\kappa^3}
\left\{ \frac{\kappa_1^2}{2} t^2 - \kappa_1 t + 1 - \e^{-\kappa_1 t} \right\} - \frac{\gamma \rho \theta}{\kappa_2^2}\left\{ \frac{\kappa_2^2}{2} t^2 - \kappa_2 t +1 - \e^{-\kappa_2 t} \right\}\\
&- \frac{ \gamma\theta\rho\phi }{\kappa_1^2 \kappa_2 (\kappa_1 +\kappa_2)}\left\{ (-\kappa_1^2-\kappa_2^2-\kappa_1\kappa_2)t + \frac{1}{2} (\kappa_1^2\kappa_2 + \kappa_1\kappa_2^2)t^2 \right.\\
&\left. - \frac{  \kappa_2^2+\kappa_1\kappa_2}{\kappa_1} (\e^{-\kappa_1 t}-1) - \frac{\kappa_1^2+\kappa_1\kappa_2}{\kappa_2} (\e^{-\kappa_2 t}-1) + \frac{\kappa_1\kappa_2}{\kappa_1+\kappa_2}(\e^{-(\kappa_1+\kappa_2) t}-1)\right\}.
\end{align*}
Finally,
\begin{align*}
\E[[R,R^2]_t] = \frac{2}{S_0^3}\int_0^t \E[S_u V_u] \D u - \frac{2 \theta t}{S_0^2}.
\end{align*}

\end{document}